\def\DRAFT{}

\documentclass[11pt,fleqn]{article}%
\usepackage{eurosym}
\usepackage{amsfonts,amsmath,amssymb,graphicx,setspace}
\usepackage[usenames, dvipsnames]{xcolor}
\usepackage[colorlinks=true,linkcolor=blue,citecolor=Mahogany,urlcolor=blue]%
{hyperref}
\usepackage[top=1in, bottom=1in, left=1in, right=1in]{geometry}
\usepackage[normalem]{ulem}
\usepackage{subcaption}
\usepackage{amsmath}
\usepackage{amsfonts}
\usepackage{amssymb}
\usepackage{amsthm}
\usepackage{mathtools}
\usepackage{graphicx}
\usepackage[longnamesfirst]{natbib}
\usepackage{dsfont} 
\usepackage{nicefrac}
\usepackage{comment}
\usepackage{xcolor}
\usepackage{mathrsfs}
\usepackage{tikz}
\usetikzlibrary{intersections,arrows.meta,calc,decorations.markings, 3d, positioning, arrows.meta, decorations.pathreplacing}

\usepackage[font=normalsize]{caption}
\DeclareCaptionLabelFormat{captionf}{\textsc{#1~#2}}
\DeclareCaptionFont{sloppy}{\sloppy}

\usepackage{pgfplots}
\pgfplotsset{compat=1.17}
\usetikzlibrary{patterns,fillbetween}

\usepackage{graphics}
\usepackage{color}
\usepackage{float}

\usepackage[normalem]{ulem}
\definecolor{accentcolor}{RGB}{26, 148, 49}
\usepackage{hyperref}
\hypersetup{pdfborderstyle={},allcolors=accentcolor,colorlinks}

\usepackage{titling}
\usepackage{bbm}
\usepackage{apptools}

\usepackage{etoolbox}
\newcommand{\zerodisplayskips}{%
  \setlength{\abovedisplayskip}{8pt}%
  \setlength{\belowdisplayskip}{7pt}%
  \setlength{\abovedisplayshortskip}{8pt}%
  \setlength{\belowdisplayshortskip}{7pt}}
\appto{\normalsize}{\zerodisplayskips}
\appto{\small}{\zerodisplayskips}
\appto{\footnotesize}{\zerodisplayskips}

\AtAppendix{\counterwithin{lemma}{section}}
\AtAppendix{\counterwithin{claim}{section}}
\usepackage[
]{ifdraft}

\usepackage{lscape}
\usepackage{placeins}
\providecommand{\U}[1]{\protect\rule{.1in}{.1in}}
\usepackage[noamsmath]{kpfonts}
\usepackage{inconsolata}

\definecolor{ForestGreen}{rgb}{.13,.54,.13}
\definecolor{BrickRed}{rgb}{.80,.26,.33}

\ifdefined\DRAFT
\newcommand{\fed}[1]{\textcolor{BrickRed}{(\textbf{Fedor:} #1)}}
\newcommand{\kir}[1]{\textcolor{BrickRed}{(\textbf{Kirill:} #1)}}
\newcommand{\lee}[1]{\textcolor{ForestGreen}{(\textbf{Leeat:} #1)}}

\else
\newcommand{\fed}[1]{}
\newcommand{\kir}[1]{}
\newcommand{\lee}[1]{}

\fi

\usepackage[T1]{fontenc}
\usepackage{pgfplots}
\providecommand{\U}[1]{\protect\rule{.1in}{.1in}}
\newtheorem{theorem}{Theorem}
\newtheorem{proposition}{Proposition}
\newtheorem{corollary}{Corollary}
\newtheorem{lemma}{Lemma}

\newcommand\E{{\rm E}}
\renewcommand{\E}{\mathbb{E}}

\newcommand{\R}{\mathbb{R}}
\newcommand{\Z}{\mathbb{Z}}

\newcommand{\nash}{\mathrm{NE}}
\newcommand{\ce}{\mathrm{CE}}

\newcommand{\conv}{\mathrm{conv}}
\newcommand{\supp}{\mathrm{supp}\,}

\newcommand{\G}{\Gamma}

\usepackage{tikz}
\usetikzlibrary{trees}

\begin{document}

\title{Extreme Equilibria: The Benefits of Correlation
\thanks{We thank S.~Nageeb Ali, Benjamin Brooks, Ron Holzman, Emir Kamenica, Navin Kartik, Richard McLean, Stephen Morris, Stephan Lauermann,  Jacopo Perego, Doron Ravid, Vasiliki Skreta, \'{E}va Tardos, and Alexander Wolitzky for helpful comments and suggestions. We also appreciate the input of seminar and conference audiences at ASSA~2025, Chicago, Columbia, Collegio Carlo Alberto, CUNY, EC 2025, Georgia, Harvard, HSE, MIT, Ohio State, Oxford, Penn, Penn State, Rochester, Rutgers, Stanford, Stony Brook, UC Berkeley, and Wisconsin-Madison. We gratefully acknowledge financial support from the National Science Foundation through grant SES-1949381.}}
\author{Kirill Rudov\thanks{Analysis Group; \href{mailto:karudov@gmail.com}{\texttt{karudov@gmail.com}}} \hspace{3mm}\hspace{2mm} Fedor Sandomirskiy\thanks{Princeton University;
\href{mailto:fsandomi@princeton.edu}{\texttt{fsandomi@princeton.edu}}} \hspace{3mm}\hspace{2mm} Leeat Yariv\thanks{Princeton University, CEPR, and NBER; \href{mailto:lyariv@princeton.edu}{\texttt{lyariv@princeton.edu}}}}
\date{\today }

\pretitle{\begin{flushleft}\LARGE} 
\posttitle{\end{flushleft}}
\preauthor{\begin{flushleft}\large} 
\postauthor{\end{flushleft}}
\predate{\begin{flushleft}} 
\postdate{\end{flushleft}}
\settowidth{\thanksmarkwidth}{*}
\setlength{\thanksmargin}{-\thanksmarkwidth}

\maketitle
\thispagestyle{empty}

\renewenvironment{abstract}
 {\par\noindent\textbf{\abstractname.}\ \ignorespaces}
 {\par\medskip}
\begin{abstract}
Correlated equilibria arise naturally when agents communicate or rely on intermediaries such as recommendation systems. We study when a given Nash equilibrium can be improved within the set of correlated equilibria for general objectives. Our key insight is a detail-free criterion: any Nash equilibrium with three or more randomizing agents is generically improvable. We refine this insight to specific classes of games and objectives, including Pareto and utilitarian welfare, and provide constructive methods to obtain improvements. Our findings underscore the ubiquity of improvable Nash equilibria and the crucial role of correlation in enhancing strategic outcomes.

\vspace{2mm}

 \vspace{2mm}


\end{abstract}

\newpage
\pagenumbering{arabic}

\section{Introduction}\label{sec:intro}

Correlated equilibria are a powerful generalization of Nash equilibrium, offering both computational and strategic advantages. Unlike Nash equilibria, which can be notoriously difficult to compute in general games, correlated equilibria are computationally simple. They also naturally arise in settings where agents can communicate or collude, making them a plausible solution concept for many applications. Moreover, they are mechanism-implementable: a mediator can credibly induce agents to follow correlated strategies without requiring external enforcement. Given these advantages, a fundamental question arises: when can Nash equilibria be strictly improved through correlation?

This paper examines conditions under which correlated equilibria strictly improve upon Nash equilibria, revealing a key insight: when a Nash equilibrium involves substantial randomization, it is generically improvable. We establish this connection between randomization and improvability for general classes of strategic interactions and social objectives. Randomization arises in a variety of strategic settings, such as bidders randomizing in auctions, firms adopting stochastic pricing strategies, or voters deciding probabilistically whether to turn out. Our findings demonstrate that such strategic randomness often leaves room for improvement. We also show that in symmetric environments, such improvements exhibit a simple structure, admit a geometric construction, and allow for straightforward implementation.

When analyzing the improvability of Nash equilibria, we adopt an agnostic stance toward the social planner’s objective. The set of correlated equilibria is convex, so any generic objective attains its unique optimum at an extreme point of this set. It follows that unless a Nash equilibrium is itself extreme, it can be improved, regardless of the specific objective. One might expect that any equilibrium can be made uniquely optimal under some (possibly contrived) objective. The existing literature confirms this view in several settings, particularly for two-agent games or pure-strategy equilibria in larger games; see our literature review. A key contribution of this paper is to broaden this perspective: in games with three or more agents, extremality provides a meaningful and nontrivial distinction between equilibria that are improvable and those that are not. 

To build intuition, consider a game where $n$ agents each choose between two actions. Such a setup arises in various strategic settings: voters deciding whether to cast a costly ballot, firms choosing whether to engage in costly R\&D, or individuals in a social network deciding whether to adopt a new technology. A correlated equilibrium is a distribution over action profiles under which each agent, upon receiving her recommended action, has no incentive to deviate. The feasibility of such an equilibrium is governed by at most $2n$ linear constraints: one per agent per action. A basic linear programming principle \citep[see, e.g.,][]{winkler1988extreme} implies that any extreme point of this set must be supported on at most $2n+1$ action profiles. Now, consider a Nash equilibrium where all agents use mixed strategies, implying that every action profile is played with positive probability. Since this yields $2^n$ active action profiles, the equilibrium can only be an extreme point if $2^n \leq 2n+1$. This condition holds only for $n = 1,2$, meaning that for any binary-action game with three or more agents who all mix, Nash equilibria are necessarily non-extreme, and hence, improvable.

\paragraph{General Improvability} Section~\ref{sec:improvability} shows that this conclusion extends to games in which agents have larger and potentially heterogeneous action sets. We show that, for generic strategic interactions, a Nash equilibrium is extreme if and only if at most two agents randomize. Accordingly, a social planner's objective is improvable when three or more agents mix. This result yields a detail-free criterion for assessing the potential to improve equilibria: regardless of the overall number of agents or the specifics of the strategic environment, the sole statistic that determines whether correlation, or a shift to a different, ``less random'' Nash equilibrium, may be beneficial is the number of agents who randomize.\footnote{We also formulate versions of this result that do not require the genericity of the environment.} When all agents randomize over action sets of the same size, the basic arguments from our example above extend naturally. The core challenge in the proof lies in handling asymmetric equilibria, where agents randomize over action sets of differing cardinalities.

\paragraph{Welfare and Pareto Improvability} It is often convenient to represent the outcomes of a game in terms of the agents’ payoffs induced by equilibrium play. In Section~\ref{sec:payoff_extremality}, we show that, generically, Nash equilibria involving at least three agents who randomize are not payoff-extreme---they are not extreme points of the (convex) set of correlated equilibrium payoffs---and allow for improvements in utilitarian welfare. We further demonstrate that when a sufficient number of agents mix, equilibria can be Pareto improved. Notably, the number of randomizing agents required for Pareto improvability grows only logarithmically with the total number of agents.

These results are grounded in convexity principles. The set of payoffs achievable through correlated equilibria arises as a linear projection of the correlated equilibrium polytope onto the space of agents' payoffs. The extreme points of a projected set are contained within the projection of the extreme points of the original set. Given that utilitarian welfare is almost always maximized at a unique extreme point, it follows that Nash equilibria with excess randomness leave room for welfare improvements via correlation. Pareto improvability is more stringent and thus requires additional constraints beyond those guaranteeing non-extremality.

\paragraph{Symmetric Games and Symmetric Improvements} Symmetric games are prevalent in economic and strategic settings, appearing in auctions, voting models, and competitive markets, where identical agents face identical strategic choices. In such games, symmetric equilibria, with all agents adopting the same strategy, often emerge as focal solutions.

In Section \ref{sec:symmetric}, we show that in symmetric games with at least three agents, any symmetric mixed equilibrium can be improved, even when restricting attention to symmetric correlated equilibria. This result implies that even when coordination is limited to symmetric recommendations---such as a common policy guideline for firms or a uniform turnout mobilization effort in elections---improvements remain possible. Only pure-strategy equilibria resist such improvements.

Furthermore, we establish that extreme symmetric correlated equilibria in symmetric games can be constructed through a simple procedure. In games involving a large number of agents, these equilibria can be represented using a small number of payoff-irrelevant states, each generating conditionally independent and identically distributed signals that guide individual behavior. 

In Section~\ref{sec_binary_action}, we introduce a geometric method for explicitly characterizing improving symmetric correlated equilibria. We use this approach to identify a class of games for which correlated equilibria can deliver higher expected payoffs to all agents than those attainable under \emph{any} Nash equilibrium.

\paragraph{Summing up} This paper provides a systematic framework for understanding when correlation can enhance strategic interactions. Whether in competitive markets, voting mechanisms, or algorithmically mediated platforms, correlation offers a natural way to improve strategic outcomes. Once communication or intermediation allow for correlation across agents' strategies, we show that Nash equilibria involving substantial randomization become fragile. This highlights a new channel for instability of mixed Nash equilibria and underscores the importance of correlation in strategic environments.

\subsection*{Related literature}
First introduced by \cite*{aumann1974subjectivity,aumann1987correlated}, correlated equilibria have received substantial attention in the literature; see \cite*{forges2024subjectivity} for a recent overview. Correlated equilibria offer a reduced-form way of capturing pre-play communication or mediation without explicitly modeling the communication phase \citep{forges2020correlated}. Similarly, they result from a variety of learning heuristics  \citep{foster1997calibrated, fudenberg1999conditional, hart2000simple} and are computationally simpler than Nash equilibria \citep{papadimitriou2008computing}.

Correlated equilibria have been analyzed in a variety of specific applications, including Cournot competition \citep{gerard1978correlation}, abatement games \citep{moulin2014coarse}, quadratic games \citep{dokka2023equilibrium}, various auctions \citep{lopomo2011bidder,feldman2016correlated,agranov2018collusion,pavlov2023correlated,ahunbay2024uniqueness}, voting \citep{gerardi2007deliberative}, and congestion \citep*{rivera2018efficiency}.
Their empirical relevance has been seen in various experimental settings: for example, in voting contexts \citep{goeree2011experimental}, in bargaining \citep{agranov2014communication}, in auctions  \citep{agranov2018collusion} and algorithmic bidding \citep{kolumbus2022auctions}, as well as in symmetric bimatrix games \citep{georgalos2020nash,friedman2022empirical}. 

While \cite*{aumann1974subjectivity} presents a $2 \times 2$ game of chicken where correlation enhances utilitarian welfare, such examples have largely been confined to specific economic environments. Beginning with \cite*{rosenthal1974correlated}, a substantial literature has explored the potential for correlation to improve upon Nash equilibria. 

Most general findings suggest that Nash equilibria are, in fact, unimprovable in various classes of games, as surveyed by \cite*{forges2024subjectivity}. In games where each agent's utility is concave in her own action, the correlated equilibrium is unique, implying that the (also unique) Nash equilibrium cannot be improved upon. This has been shown for potential games by \cite*{neyman1997correlated} and more generally by \cite*{ui2008correlated} and \cite*{cao2025correlated}. See also \cite*{einy2022strong}, \cite*{wu2008correlated}, and \cite*{jann2015correlated} for specific applications exhibiting unique correlated equilibria. Beyond the singleton correlated equilibrium case, \cite*{nau2004geometry} show that all Nash equilibria lie on the boundary of the correlated equilibrium polytope, indicating limited scope for improvement.\footnote{\cite*{peeters1999structure}, \cite*{hendrickx2002relation}, and \cite*{calvo2006} also study the structure of the correlated equilibrium set of $2 \times 2$ games.} For two-agent games with arbitrary action sets, \cite*{cripps1995extreme}, \cite*{evangelista1996note}, and \cite*{canovas1999nash} establish that all Nash equilibria are extreme points of the correlated equilibrium set. 

Using a computer science approach and focusing primarily on $2 \times 2$ games, \cite*{ashlagi2008value} are, to our knowledge, the only ones to provide general results regarding games where Nash equilibria can be improved via correlation. They quantify the maximal magnitude of a utilitarian welfare improvement that can be attained by correlation via the ``mediation value'', a concept akin to the price of anarchy. The mediation value is defined as the ratio between the highest utilitarian welfare attainable through correlated equilibria and that achievable through Nash equilibria. They obtain precise bounds on the mediation value for $2 \times 2$ games. While this paper provides important steps in bounding the extent to which Nash equilibrium outcomes can be improved within certain classes of games, general conditions under which correlation improves upon a given Nash equilibrium have remained elusive, a gap this paper aims to fill.

The connection between randomization and inefficiency was also examined by \cite*{dubey1978finiteness}, who compares Nash equilibria to the first-best unconstrained benchmark, abstracting away from incentive constraints. He shows that in a generic game, a necessary condition for the Pareto optimality of a Nash equilibrium, within the space of all distributions over action profiles, is that at least one agent plays a pure strategy. In contrast, we characterize when Nash equilibria can be improved relative to incentive-compatible outcomes, as represented by correlated equilibria, either for arbitrary objectives or for particular ones such as utilitarian or Pareto efficiency. While our paper connects the value of correlation and randomization in games, \cite*{kamenica2024commitment} establish a similar link between the value of commitment and randomization in one-agent persuasion problems.\footnote{Choosing an information structure that induces a joint distribution over the state and the receiver's action is equivalent to choosing a correlated equilibrium in an auxiliary complete-information game with three agents. The random state is simulated by two auxiliary agents engaged in a zero-sum game. For instance, two equally likely states can be simulated by matching pennies. More generally, a state $\theta \in \Theta$ drawn from prior $p$ can be simulated by setting $A_1 = A_2 = \Theta$ and defining the first auxiliary agent's utility as $u(a_1, a_2) = \mathbf{1}_{a_1 = a_2} \cdot \frac{1}{p(a_1)}$. The sender's problem of finding an optimal information structure then reduces to selecting an optimal correlated equilibrium in the three-agent game comprising these two auxiliary agents and the receiver.}

Methodologically, our paper contributes to the growing literature on extreme-point techniques in economic theory; see, for example,  \cite*{manelli2007multidimensional}, \cite*{bergemann2015limits}, \cite*{kleiner2021extreme}, \cite*{nikzad2022constrained}, \cite*{arieli2023optimal}, \cite*{kleiner2024extreme}, \cite*{yang2024monotone}, \cite*{lahr2024extreme}, \cite*{augias2025economics}, and \cite*{yang2026stochastic}. Results on Pareto-improving correlated equilibria require new insights into the combinatorial problem of cylinder packing studied by~\cite{felzenbaum1993packing}, while those on the structure of symmetric correlated equilibria rely on versions of de Finetti's theorem for finite exchangeable distributions obtained by \cite*{diaconis1980finite}. Our construction of improving correlated equilibria in Section~\ref{sec_binary_action} uses convexification techniques reminiscent of those employed in the information design literature \citep{aumann1995repeated, kamenica2011bayesian}.

\section{Model \label{sec:model}}

Consider a game $\G=\big(N,(A_i)_{i\in N}, (u_i)_{i\in N}\big)$, where $N=\{1,\ldots,n\}$ is a finite set of agents, $A_i$ is a finite set of actions of agent~$i$, and $A= \prod_{i\in N}A_i$ is the set of action profiles. The function  $u_i\colon A\to\R$ represents the utility of each agent $i \in N$. 

The set $\ce(\G)$ of correlated equilibria consists of all probability distributions $\mu\in\Delta(A)$ 
such that for all $i\in N$ and for all distinct $a_i,a_i'\in A_i$ we have
\begin{equation}\label{eq_CE_definition}
 \sum_{a_{-i}\in A_{-i}} \mu(a_i,a_{-i}) u_i(a_i,a_{-i}) \geq \sum_{a_{-i}\in A_{-i}} \mu(a_i,a_{-i}) u_i(a_i',a_{-i}).
\end{equation}
One can interpret $\mu$ as the distribution of actions recommended by a mediator, ensuring that each agent $i$ finds it optimal to follow the recommended action $a_i$, as captured by the incentive constraint~\eqref{eq_CE_definition}. Beyond the mediation interpretation, condition~\eqref{eq_CE_definition} ensures that, given the distribution governing play, no agent can improve her expected payoff by unilaterally deviating once she takes into account the information encoded in that distribution.

Correlated equilibria $\mu$ that are product distributions---$\mu=\mu_1\times \ldots \times \mu_n$ with $\mu_i \in \Delta(A_i)$---form the set of Nash equilibria $\nash(\G)$. The set $\nash(\G)$ is non-empty \citep{nash1950non} and thus $\ce(\G)$ is non-empty as well. Since $\ce(\G)$ is cut from $\Delta(A)$ by a finite number of incentive constraints~\eqref{eq_CE_definition}, it is a non-empty convex polytope and can be described as the convex hull of its vertices, commonly referred to as extreme points.\footnote{A point $x$ of a convex set $X$ is called extreme if it cannot be represented as a non-trivial convex combination of other points from $X$, i.e., $x=\alpha y+(1-\alpha)y'$ with $\alpha\in (0,1)$ and $y,y'\in X$ can only hold if $y=y'=x$.}

If a Nash equilibrium is an extreme point of the set of correlated equilibria $\ce(\G)$, we call it an \textbf{extreme Nash equilibrium}. Our main goal is to provide conditions under which Nash equilibria are extreme. As we show below, when Nash equilibria are not extreme, the agents or the social planner would potentially have opportunities to improve outcomes via channels of communication, the use of intermediaries, or the design of mechanisms. 

To make general structural insights about extreme Nash equilibria possible, we need to rule out trivial examples, such as degenerate games where an agent is indifferent across all actions. We exclude such examples using the classical notion of regularity introduced by \cite*{harsanyi1973oddness}. In essence, a \textbf{regular Nash equilibrium} is an equilibrium that remains stable under small perturbations of utility functions.\footnote{A Nash equilibrium $\nu$ is regular if the incentive constraints outside the support of $\nu$ are not active, and the conditions of the implicit function theorem are satisfied on the support of $\nu$. A formal definition of regularity is given in Appendix~\ref{app_polygon}. Regularity ensures that the equilibrium weights depend smoothly on utility functions in a small neighborhood of $\nu$ and rules out equilibria that place positive weight on weakly dominated actions.}

Arguably, only regular equilibria are relevant for economic modeling, which makes regularity a standard assumption \citep{van1991stability}.\footnote{For example, regularity of all equilibria is assumed in classical results on the oddness of Nash equilibria \citep{harsanyi1973oddness,wilson1971computing} or purification \citep{harsanyi1973games,govindan2003short}.} Furthermore, games with all regular equilibria are prevalent.

We say that a \textbf{generic game} has a certain property if this property holds for an open everywhere dense set of games with the complement having zero Lebesgue measure. As shown by \cite*{harsanyi1973oddness}, all equilibria are regular in a generic game.

\paragraph{Extremality and improvability}
Non-extremality offers a conservative perspective on the improvability of an equilibrium: it can be improved for \textit{all} objectives in a large class. We now formalize this idea. First, consider a designer maximizing a linear objective~$W$---for instance, utilitarian welfare---over the polytope $\ce(\G)$ of correlated equilibria. Crucially, linearity here pertains to linearity in probabilities, not in the actions themselves. In particular, common objectives---such as expected welfare, revenue, or the likelihood of a given action profile---are linear.

\begin{figure}
\begin{center}
\begin{tikzpicture}[scale=0.8]
\filldraw[black!10, thick] (0,0) -- (2.5,0) -- (3,2) -- (1,3) -- (0,0);
\draw[black!90,thick] (0,0) -- (2.5,0) -- (3,2) -- (1,3) -- (0,0);
\draw[black!50] (1.5,3.5) -- (4.5,0.5);
\draw[->, thick] (2.5,2.5) -- (3.5,3.5) node[below right] {Objective};
\fill (3,2) circle (4pt);

\begin{scope}[xshift=10cm]
\filldraw[black!10, thick] (0,0) -- (2.5,0) -- (3,2) -- (1,3) -- (0,0);
\draw[black!90,thick] (0,0) -- (2.5,0) -- (3,2) -- (1,3) -- (0,0);
\draw[black!50] (4.5,1.25) -- (0,3.5);
\draw[->, thick] (2,2.5) -- (2.75,4) node[below right] {Objective};
\fill (3,2) circle (4pt);
\fill (1,3) circle (4pt);
\draw[line width=2pt] (3,2) -- (1,3);
\end{scope}
\end{tikzpicture}
\end{center}
\captionsetup{justification=raggedright,singlelinecheck=false}
\caption{A non-degenerate objective (left panel) and a degenerate one (right panel).}
\label{fig_degenerate}
\end{figure}
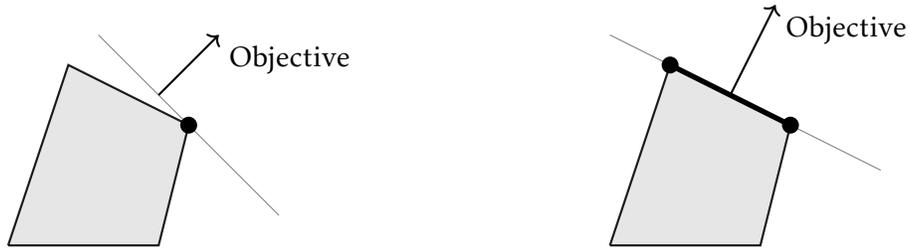

The left panel of Figure~\ref{fig_degenerate} illustrates a representative case of a linear objective, which is maximized at a unique point on the polytope. 
The right panel depicts a knife-edge case, where the level hyperplanes of the objective are parallel to a face of the polytope, and thus the optimum is not unique. Unless we are in this knife-edge case, the value of a linear objective can be strictly improved at any non-extreme point of the polytope.

This phenomenon is more general. Consider an arbitrary (weakly) convex objective $W$ on $\ce(\G)$. Convexity of~$W$ may capture aversion to concentrated probability on bad outcomes: for example, a regulator may seek to avoid action profiles triggering bank runs, assigning increasingly high costs to equilibria that put substantial weight on such outcomes. Suppose a Nash equilibrium $\nu$ is non-extreme. Hence, $\nu$ can be represented as $\alpha \mu+(1-\alpha)\mu'$ with $\alpha\in (0,1)$ and distinct $\mu,\mu'\in\ce(\G)$. By convexity, we obtain $W(\nu)\leq\alpha W(\mu)+(1-\alpha) W(\mu')$ and, thus, one of $\mu,\mu'$ gives at least as high value to $W$ as~$\nu$. It follows that any convex objective is maximized at an extreme point of $\ce(\G)$. This reflects Bauer's Maximum Principle \citep[see, e.g., ][Theorem~7.69]{guide2006infinite}. 

We call an objective that attains a unique optimum over $\ce(\G)$ a \textbf{non-degenerate objective} for $\G$. By Bauer's Principle, a non-degenerate objective $W$ can be strictly improved within $\ce(\G)$ at any non-extreme Nash equilibrium.

Although the optima of a degenerate objective may include non-extreme points, a small perturbation of such a degenerate objective suffices to rule out all non-extreme equilibria. Indeed, consider an $\varepsilon$-perturbed objective $W_\varepsilon(\mu)=W(\mu)+\sum_{a\in A} \varepsilon_a\cdot \mu(a)$, where $\varepsilon\in \R^A$ is a vector of small shocks. If the shocks $\varepsilon_a$ are random with any absolutely continuous distribution (for instance, uniformly in a small ball), any non-extreme equilibrium is optimal for a measure-zero set of realized objectives.
\medskip

\paragraph{Aumann's example and improvability} Consider a version of \cite*{aumann1974subjectivity}'s classical example illustrating the power of correlated equilibria. There are two agents, who face the following game (each entry represents an action profile and contains the corresponding payoff pair, where the first is the row agent's payoff and the second is the column's payoff): 
\begin{equation}\label{eq_Aumann_game}
\begin{pmatrix} 
0,0 &  4,1\\
1,4 &  3,3
\end{pmatrix}
\end{equation}
The designer aims to maximize the social welfare $W(\mu)=\sum_{a} \big(u_1(a)+u_2(a)\big)\cdot \mu(a)$ over the set of all correlated equilibria $\mu$. For this game, the incentive constraints~\eqref{eq_CE_definition} defining the set of correlated equilibria 
take an especially simple form:
the weight of each payoff profile $(4,1)$ or $(1,4)$ must be at least the weight of each of $(0,0)$ or $(3,3)$.

This game has a mixed Nash equilibrium where each agent randomizes uniformly. Its utilitarian welfare level of $4$ can be improved in an incentive-compatible way by reducing the weight of the $(0,0)$ outcome. In fact, there is a welfare maximizing correlated equilibrium in which each of the non-zero payoff pairs is reached with probability $1/3$. This correlated equilibrium generates a utilitarian welfare of $16/3$, which exceeds that of the mixed Nash equilibrium, as well as that of the two pure Nash equilibria.

Although none of the Nash equilibria are optimal for the utilitarian welfare objective, they are all extreme and thus non-improvable according to our conservative perspective. 
Indeed, for some objectives, these equilibria are the unique  maximizers. For instance, if the designer cares about the total cubed payoffs of agents, the pure asymmetric equilibria become optimal. If the objective is to maximize the total weight assigned to the action profiles $(0,0)$ and $(3,3)$, the mixed Nash equilibrium becomes the unique optimum.

This example highlights the observation that our notion of improvability is demanding, and suggests that Nash equilibria may perhaps rarely, if ever, be non-extreme. As we will see, this intuition is limited to two-agent games.

\section{A Characterization of Extreme Nash Equilibria} \label{sec:improvability}

We show that whether or not a Nash equilibrium is extreme depends on the amount of randomization agents invoke in their strategies. Roughly speaking, equilibria with substantial uncertainty over the action profiles implemented cannot be extreme.

\begin{theorem}\label{th_almost_pure}
A regular Nash equilibrium is extreme if and only if 
at most two agents randomize.
\end{theorem}

Since, by \cite{harsanyi1973oddness}, all Nash equilibria are regular in a generic game, the conclusion of Theorem~\ref{th_almost_pure} holds generically for \emph{all equilibria}. Furthermore, equilibria with extensive randomization are prevalent: \cite{mckelvey1997maximal} study the number of regular mixed Nash equilibria with a given support in generic games and show that this
number grows sharply with the number of agents who randomize.

Theorem~\ref{th_almost_pure} has a direct implication for welfare comparisons. For any social planner with a strictly convex objective, if three or more agents randomize at a Nash equilibrium, the outcome can be improved, either by introducing correlation---by allowing agents to communicate or introduce mediation---or by selecting an alternative pure or ``almost pure'' equilibrium when such exist.  

For instance, suppose the planner values aggregate efficiency but dislikes uncertainty over outcomes. Welfare may then take the form of aggregate expected payoff net of the entropy of outcomes. Theorem \ref{th_almost_pure} implies that the resolution to the trade-off between payoff efficiency and uncertainty is, in many ways, detail-free. Naturally, a dislike of uncertainty would push the social planner to impose limited mixing by agents. The theorem indicates that, regardless of the number of agents or the payoff structure, the social planner could always improve upon a Nash equilibrium in which more than two agents mix.

Pure equilibria are trivially extreme: they are the unique optimum for an objective that values the weight on the particular action profile that the equilibrium prescribes. Equilibria with exactly one randomizing agent are ruled out under the regularity assumption, as we explain below. For equilibria involving two randomizing agents, Theorem \ref{th_almost_pure} asserts that any regular Nash equilibrium cannot be improved upon either. This echoes \cite*{cripps1995extreme}, \cite*{evangelista1996note}, and \cite*{canovas1999nash}, who demonstrated that Nash equilibria are extreme in two-agent games. Our result indicates that universal non-improvability is special to two-agent games.\footnote{To illustrate how richer randomization creates opportunities to correlate actions, consider the product of two unrelated two-agent games, one played by agents 1 and 2 and the other by agents 3 and 4, with each pair randomizing at a regular mixed Nash equilibrium in their game. The product of these two equilibria is a four-agent Nash equilibrium. Although each component equilibrium is extreme in its own correlated-equilibrium polytope, the product distribution is not extreme in the four-agent polytope: the two components can be correlated in different ways while preserving the marginals and all incentive constraints.}

\medskip

Theorem \ref{th_almost_pure} also has implications for games with a unique correlated equilibrium. A unique correlated equilibrium is necessarily a Nash equilibrium. Uniqueness ensures resilience of this Nash equilibrium to pre-play communication or intermediation that, in some contexts, may underlie collusive arrangements \citep{agranov2018collusion}. Moreover, Nash equilibria in games with a unique correlated equilibrium exhibit robustness to small payoff perturbations or payoff uncertainty; for example, see \cite*{kajii1997robustness}, \cite*{viossat2008having}, and \cite*{einy2022strong}.
In particular, \cite*{viossat2008having} shows that the set of games with a unique correlated equilibrium is an open set within the set of all games. Several classes of games are known to display this property, including Bertrand oligopoly  \citep{milgrom1990rationalizability}, Cournot oligopoly \citep{liu1996correlated, neyman1997correlated}, and generic two-agent conflicting-interest (constant-sum) games; see our literature review for more examples. In all these environments, the Nash equilibria entail limited randomization. Theorem \ref{th_almost_pure} implies that this observation is true in general.

\begin{corollary}\label{cor_unique}
If $\G$ is a game with a unique correlated equilibrium $\nu$, then $\nu$ is either a pure Nash equilibrium or entails precisely two agents mixing.\footnote{Corollary~\ref{cor_unique} answers negatively an open question posed by \cite*{he2021private} on the existence of a three-agent game with a unique mixed correlated equilibrium in which all agents randomize.}
\end{corollary}

For a generic game, Corollary~\ref{cor_unique} follows immediately from Theorem~\ref{th_almost_pure}. However, the corollary requires neither genericity nor regularity assumptions. Indeed, Theorem~\ref{th_almost_pure} implies the result for a dense subset of games within the open set of games with a unique correlated equilibrium.
To extend it to the full set, we use the upper hemicontinuity of the Nash correspondence. Combined with uniqueness, hemicontinuity implies that the Nash correspondence is, in fact, continuous on this set. Since the limit of Nash equilibria each involving at most two mixing agents cannot itself involve more mixers, the conclusion carries over to all games with a unique correlated equilibrium.

\medskip

The proof of the theorem is contained in Appendix~\ref{app_th1_proof}. There, we formulate a slightly stronger result that also provides a lower bound on the dimension of the face carrying a non-extreme Nash equilibrium.\footnote{For generic games, Nash equilibria lie on the boundary of the correlated equilibrium polytope, as shown by \cite*{nau2004geometry}. Our result bounds the dimension of that boundary face.} This dimension grows exponentially with the number $k\geq 3$ of randomizing agents:
\begin{equation}\label{eq_dimension_bound_intuition}
\dim\big(\text{the face carrying a Nash equilibrium}\big)\geq 2^{k-3}.  
\end{equation}
The dimension of a face represents the number of directions in which one can perturb the Nash equilibrium without violating the incentive constraints. It plays a key role in our ability to improve several objectives simultaneously, an insight we use when we turn to Pareto improvements of Nash equilibria in Section~\ref{sec:payoff_extremality}.

\medskip

To glean high-level intuition for why equilibria with many agents mixing cannot be extreme, consider an $n$-agent game $\Gamma$ and a Nash equilibrium $\nu$ supported on $S = S_1 \times \cdots \times S_n$. When a large number of agents mix at $\nu$, their joint randomization yields a rich set of directions at which their actions can be jointly correlated. Among these feasible directions, some necessarily preserve all incentive constraints. This, in turn, implies that the equilibrium does not constitute an extreme point.

Formally, we consider a perturbation of $\nu$ in a direction $\tau\in\R^S$ such that both $\nu+\varepsilon\tau$ and $\nu-\varepsilon\tau$ remain correlated equilibria for sufficiently small $\varepsilon>0$. Geometrically, the set of such $\tau$ forms the tangent space to the set of correlated equilibria at $\nu$. If this tangent space contains any nonzero vector---that is, it has positive dimension---then the equilibrium admits nontrivial perturbations that preserve all incentive constraints and therefore cannot be extreme.

What constraints must a vector $\tau$ from the tangent space satisfy? Since the incentive constraints for $\nu$ hold as equalities on the support, they must also be satisfied by $\tau$, yielding one constraint for each pair of distinct actions $a_i,a_i'\in S_i$ of each agent $i$. The incentive constraints outside the support of $\nu$ do not play a role, since they are slack by the regularity of $\nu$. In addition, $\tau$ must have zero total mass to ensure that $\nu\pm\varepsilon\tau$ are probability distributions. In total, we impose $1+\sum_i |S_i|(|S_i|-1)$ linear constraints on a vector in a space of dimension $|S|=\prod_i |S_i|$. Thus the dimension of the tangent space is at least
\begin{equation}\label{eq_dimension_counting}
    \prod_i |S_i|-\left(1+\sum_i |S_i|(|S_i|-1) \right).
\end{equation}
For example, if $k\geq 3$ agents randomize over $2$ actions each, the dimension is at least 
\begin{equation*}
    2^k-2k-1\geq 1
\end{equation*}
and we conclude that such a Nash equilibrium is not extreme in agreement with Theorem~\ref{th_almost_pure}. A version of this argument can also be used to deduce the conclusion of the theorem for any Nash equilibrium in which all agents randomize over (roughly) the same number of actions.

Consider now equilibria in which agents mix over dramatically different numbers of actions, say where one agent randomizes over many actions, while the others randomize over only a few. For such equilibria, \eqref{eq_dimension_counting} does not guarantee that the tangent space has a positive dimension. Indeed, even for $n=3$, with $|S_1|=4$ and $|S_2|=|S_3|=2$, the expression~\eqref{eq_dimension_counting} equals $-1$, and so the bound holds only trivially. However, regular Nash equilibria do not allow for such dramatic asymmetry in support sizes, and must satisfy:
\begin{equation}\label{eq_polygon_intuition}
   |S_i|-1\leq \sum_{j\ne i} (|S_j|-1)\qquad \text{for all}\quad i=1,\ldots, n.
\end{equation}
For example, a regular equilibrium cannot have exactly one mixing agent; and if exactly two agents mix, then they must mix over the same number of pure actions.

We term~\eqref{eq_polygon_intuition} the \textbf{polygon inequality}. Geometrically, this condition means that one can construct an $n$-sided polygon, where the length of the $i$-th side corresponds to $|S_i|-1$. In the three-agent case ($n=3$) depicted in Figure~\ref{fig:polygon}, this simplifies to the familiar triangle inequality that, for instance, rules out regular equilibria with $|S_1|=4$ and $|S_2|=|S_3|=2$ since $4-1\not\leq (2-1)+(2-1)$. The intuition behind the polygon inequality relies on counting degrees of freedom. 
Since each agent must be indifferent between actions she takes with positive probability, each additional action in her support imposes an additional constraint on opponents' strategies. The opponents collectively possess only $\sum_{j\ne i}(|S_j|-1)$ degrees of freedom. Therefore, the polygon inequality simply states that the number of imposed constraints cannot exceed the available degrees of freedom.\footnote{While this inequality has appeared in the literature, we could not find a short self-contained proof. For readers' convenience, we  include a straightforward proof in the appendix; see Lemma~\ref{lm_polygon}. The particular case of \eqref{eq_polygon_intuition} for three agents was obtained by \cite*{chin1974structure}. In full generality, the result can be deduced from the analysis by \cite*{mckelvey1997maximal}, who bound the maximal possible number of Nash equilibria with a given support in a generic game. They show that this number is~$0$ for equilibria violating~\eqref{eq_polygon_intuition}. \cite*{kreps1981} obtains a ``dual'' result: under condition~\eqref{eq_polygon_intuition} on the support of $\nu$, there exist utilities $u_1,\ldots, u_n$ making $\nu$ a unique Nash equilibrium.}

\begin{figure}
    \centering
\begin{tikzpicture}[
    scale=1.5,  
    dot/.style={circle, fill=black, inner sep=1pt},
    every label/.append style={font=\footnotesize},
    edge label/.append style={fill=white, inner sep=1pt, text height=1.5ex, text depth=.25ex}
]

\node[dot] (A) at (0,0) {};
\node[dot] (B) at (3,0) {};
\node[dot] (C) at (1.2,1.8) {};

\draw[thick] (A) -- (B) node[midway, below, font=\small] {$|S_1| - 1$};
\draw[thick] (A) -- (C) node[midway, sloped, above, font=\small] {$|S_2| - 1$};
\draw[thick] (C) -- (B) node[midway, sloped, above, font=\small] {$|S_3| - 1$};

\end{tikzpicture}
    \caption{The polygon inequality for a three-agent game: given a regular Nash equilibrium with support $S_1 \times S_2 \times S_3$, values $|S_i| - 1$ are the side lengths of a triangle.}
    \label{fig:polygon}
\end{figure}
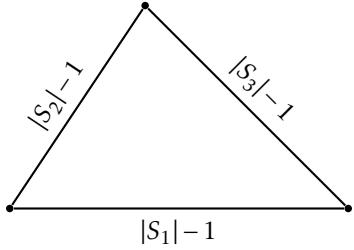

The remaining, and most technical, part of the proof of Theorem~\ref{th_almost_pure} is to show that the dimension bound~\eqref{eq_dimension_counting}, combined with the polygon inequality~\eqref{eq_polygon_intuition}, implies that for $k\geq 3$ mixing agents, the dimension of the tangent space is always positive and, moreover, admits a lower bound~\eqref{eq_dimension_bound_intuition}. The details are provided in the appendix.
\medskip

To conclude that a mixed Nash equilibrium $\nu$ with three or more mixing agents is not extreme, Theorem~\ref{th_almost_pure} requires verifying that $\nu$ is regular, a condition that may be difficult to check in a given game $\Gamma$. However, this verification can be circumvented. As the above discussion suggests, the role of regularity in the theorem is to ensure two conditions: that all incentive constraints corresponding to actions outside the support are inactive, and that the support of $\nu$ satisfies the polygon inequality~\eqref{eq_polygon_intuition}. Thus, regardless of regularity, if a mixed equilibrium $\nu$ with three or more mixing agents satisfies the ``inactivity condition'' and~\eqref{eq_polygon_intuition}, it is not extreme.

In some cases, even these checks are unnecessary. For instance, if $\nu$ is totally mixed---that is, if every action is played with positive probability---the inactivity condition becomes vacuous. Moreover, the polygon inequality~\eqref{eq_polygon_intuition} is automatically satisfied if at least $k\geq 2$ agents randomize over the same number of actions, while the remaining $n - k$ agents choose actions deterministically. Combining these observations yields the following corollary, which is particularly easy to apply.
\begin{corollary}\label{cor_fully_mixed_asymmetric}
In a game with $n \geq 3$ agents, each of whom has the same number of actions, any totally mixed Nash equilibrium is not extreme.
\end{corollary}

While Theorem~\ref{th_almost_pure} rules out extreme Nash equilibria with more than two agents mixing, the situation is different for extreme correlated equilibria. These may involve more complex forms of randomization, but only to a limited extent: an extreme correlated equilibrium must still place weight on a relatively small collection of action profiles. Indeed, if each of the $n$ agents has at most $m \ge 2$ actions, the number of incentive constraints in~\eqref{eq_CE_definition} is at most $n \cdot m(m-1)$. By \cite{winkler1988extreme}, this bound implies that the support of an extreme correlated equilibrium $\mu$ cannot exceed
\begin{equation*}
    |\supp(\mu)|\leq n\cdot m(m-1)+1.
\end{equation*}
That is, the number of action profiles in the support grows linearly in the number of agents~$n$, while the total number of action profiles $m^n$ is exponential in $n$.

Our results have implications for strategic environments in which Nash equilibria require substantial randomization. In such settings, outcomes may improve when agents coordinate, either directly or through intermediaries. For example, in costly voting models, agents may mix between abstaining and participating \citep{palfrey1983strategic}, and some formulations yield a unique  mixed equilibrium as such \citep{adachi2004costly}. 
Similar binary participation incentives lead to mixed equilibria in discrete public-good provision and the volunteer dilemma \citep{palfrey1984participation,diekmann1985volunteer}. In oligopolistic markets with search or information frictions, firms randomize over prices to avoid being systematically undercut \citep{shilony1977mixed,varian1980model,stahl1989oligopolistic,vives1999oligopoly}. Related forces appear in complete-information all-pay auctions and contests, where the incentive to slightly overbid, much like the incentive to undercut, drives equilibrium randomization \citep{baye1993rigging,baye1996all,siegel2009all}. Colonel Blotto games provide another canonical illustration: players allocate resources across battlefields and must randomize to avoid predictable weaknesses \citep{roberson2006colonel}. Across these settings, equilibrium mixing reflects strategic vulnerability to being anticipated. Our results show that allowing for correlation can strictly improve outcomes relative to Nash equilibria. While the distinction between coordination and collusion can be subtle, this perspective highlights environments in which the gains from coordination, and the risks of collusion, are likely to be particularly pronounced.

\section{Utilitarian and Pareto Improvements}\label{sec:payoff_extremality}

The preceding discussion characterizes extreme equilibria in terms of their structure within the space of distributions over actions. However, economic analysis often focuses on the payoffs associated with equilibria, rather than the distributions themselves. We now show that utilitarian welfare of an equilibrium---the sum of agents' utilities---can generically be improved whenever more than two agents mix. Further, we demonstrate that an equilibrium can be Pareto improved when the number of mixing agents is logarithmic in the size of the population.

Formally, for each $\mu\in \Delta(A)$, assign the expected payoff vector $u(\mu)\in \R^n$ given by $u_i(\mu)=\sum_a u_i(a)\mu(a)$. For any game $\G$, let $U^{CE}(\G) \subset \R^n$ be the (convex) set of correlated equilibrium payoff vectors. We say that a Nash equilibrium $\nu$ with payoff vector $u(\nu)$ is \textbf{payoff-extreme} if $u(\nu)$ is an extreme point of $U^{CE}(\G)$. 

Payoff-extremality is tightly related to welfare improvability. Indeed, consider any linear objective over the space of payoffs. Such an objective corresponds to a weighted utilitarian welfare function with weights~$\alpha=(\alpha_1,\ldots, \alpha_n)$: 
\begin{equation*}
    W_\alpha(\mu)=\sum_i \alpha_i \sum_a u_i(a)\mu(a)= \sum_i \alpha_i u_i(\mu).
\end{equation*} 
Thus, if a Nash equilibrium is not payoff-extreme, $W_\alpha$ can be improved unless $\alpha$ is orthogonal to the corresponding face of $U^{CE}(\G)$. For a generic game~$\G$, this exception occurs only for a zero-measure set of weights $\alpha$. Since this zero-measure set depends on $\G$, there is no a priori guarantee that for a given weight vector $\alpha$---for example, $\alpha=(1,\ldots, 1)$, corresponding to standard utilitarian welfare---the induced objective $W_\alpha$ admits an improvement in a generic game. 

The following proposition illustrates that, for a generic game, equilibria involving sufficient randomization are not payoff-extreme and their utilitarian welfare can, in fact, be improved.

\begin{proposition}
\label{prop:payoff_extreme}
In a generic game, a Nash equilibrium $\nu$ with more than two agents randomizing is not payoff-extreme. Furthermore, its utilitarian welfare, $\sum_i u_i(\nu)$, can be strictly improved.
\end{proposition}

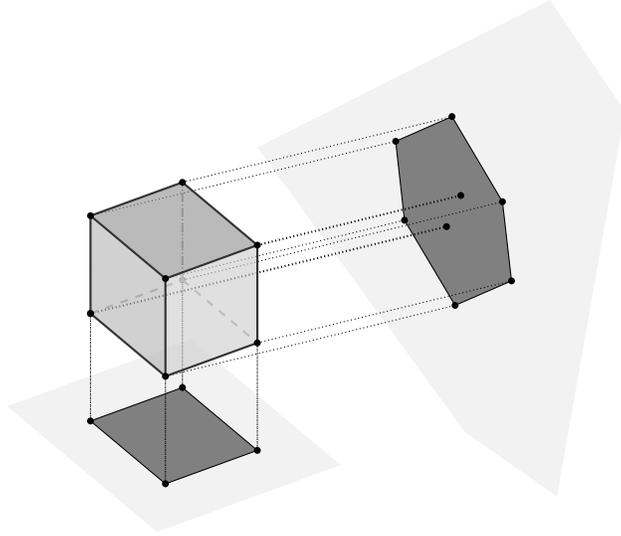
\begin{figure}
    \centering
\begin{tikzpicture}[scale=1.3,
    font=\small,
    color_cube_face/.style={fill=gray!40, fill opacity=0.6},
     color_cube_face_left/.style={fill=gray!60, fill opacity=0.6},
      color_cube_face_top/.style={fill=gray!90, fill opacity=0.6},
    color_cube_edge/.style={thick, draw=black!80},
    color_hidden_edge/.style={thick, draw=gray!70, dashed},
    color_proj_plane/.style={fill=gray!20, fill opacity=0.5},
    color_proj_shape/.style={fill=gray!100},
    color_proj_line/.style={draw=black, thin, densely dotted},
    label_style/.style={font=\tiny, text=black},
    label_style_emph/.style={font=\tiny\bfseries, text=black},
    arrow_style/.style={-latex, thick}
]

\pgfmathsetmacro{\xAngle}{-40}
\pgfmathsetmacro{\yAngle}{20}
\pgfmathsetmacro{\zAngle}{90}

\begin{scope}[x={({cos(\xAngle)*1cm},{sin(\xAngle)*1cm})}, y={({cos(\yAngle)*1cm},{sin(\yAngle)*1cm})}, z={({cos(\zAngle)*1cm},{sin(\zAngle)*1cm})}]

        \coordinate (P000) at (0,0,0); \coordinate (P100) at (1,0,0);
        \coordinate (P010) at (0,1,0); \coordinate (P110) at (1,1,0);
        \coordinate (P001) at (0,0,1); \coordinate (P101) at (1,0,1);
        \coordinate (P011) at (0,1,1); \coordinate (P111) at (1,1,1);

        \coordinate (P'00) at (0,0,-1.1); \coordinate (P'10) at (1,0,-1.1);
        \coordinate (P'01) at (0,1,-1.1); \coordinate (P'11) at (1,1,-1.1);

        \fill[color_proj_plane] (-0.5,-0.5,-1.1) -- (1.5,-0.5,-1.1) -- (1.5,1.5,-1.1) -- (-0.5,1.5,-1.1) -- cycle;

        \filldraw[color_proj_shape] (P'00) -- (P'10) -- (P'11) -- (P'01) -- cycle;

        \draw[color_proj_line] (P000) -- (P'00); \draw[color_proj_line] (P001) -- (P'00);
        \draw[color_proj_line] (P100) -- (P'10); \draw[color_proj_line] (P101) -- (P'10);
        \draw[color_proj_line] (P010) -- (P'01); \draw[color_proj_line] (P011) -- (P'01);
        \draw[color_proj_line] (P110) -- (P'11); \draw[color_proj_line] (P111) -- (P'11);

\coordinate (P'000) at (2.8616, 1.5409, 2.2013);
\coordinate (P'100) at (3.3270, 1.2547, 1.7925);
\coordinate (P'010) at (2.5723, 2.3852, 1.9811);
\coordinate (P'001) at (2.4528, 1.3208, 2.8868);
\coordinate (P'110) at (3.0440, 2.0991, 1.5723);
\coordinate (P'101) at (2.9214, 1.0346, 2.4780);
\coordinate (P'011) at (2.1667, 2.1667, 2.6667);
\coordinate (P'111) at (2.6352, 1.8805, 2.2579);

        \fill[color_proj_plane] (5,1,1) -- (1,5,1) -- (0,5,1.5) -- (1,1,2) -- (5,0,2) -- cycle;
        \filldraw[color_proj_shape] (P'100) -- (P'110) -- (P'010) -- (P'011) -- (P'001) -- (P'101) -- cycle;

        \draw[color_proj_line] (P100) -- (P'100);
        \draw[color_proj_line] (P010) -- (P'010);
        \draw[color_proj_line] (P001) -- (P'001);
        \draw[color_proj_line] (P110) -- (P'110);
        \draw[color_proj_line] (P101) -- (P'101);
        \draw[color_proj_line] (P011) -- (P'011);
        \draw[color_proj_line, thick] (P000) -- (P'000);
        \draw[color_proj_line, thick] (P111) -- (P'111);

        \fill[gray] (P010) circle (1.0pt);         
        \draw[color_hidden_edge] (P000) -- (P010);
        \draw[color_hidden_edge] (P010) -- (P011);
        \draw[color_hidden_edge] (P010) -- (P110);
        \fill[color_cube_face] (P100) -- (P110) -- (P111) -- (P101) -- cycle; 
        \fill[color_cube_face_top] (P001) -- (P101) -- (P111) -- (P011) -- cycle; 
        \fill[color_cube_face_left] (P000) -- (P100) -- (P101) -- (P001) -- cycle; 
        \draw[color_cube_edge] (P100) -- (P110);
        \draw[color_cube_edge] (P101) -- (P111) -- (P011);
        \draw[color_cube_edge] (P100) -- (P101); \draw[color_cube_edge] (P110) -- (P111); 
        \draw[color_cube_edge] (P001) -- (P101); \draw[color_cube_edge] (P001) -- (P011);
        \draw[color_cube_edge] (P000) -- (P100);
        \draw[color_cube_edge] (P000) -- (P001);

        \fill (P'000) circle (1.0pt);
        \fill (P'111) circle (1.0pt);
         \fill (P000) circle (1.0pt);
        \fill (P111) circle (1.0pt);
        \fill (P100) circle (1.0pt);
        \fill (P101) circle (1.0pt);
        \fill (P110) circle (1.0pt);
        \fill (P001) circle (1.0pt);
        \fill (P011) circle (1.0pt);

 \fill (P'10) circle (1.0pt);
        \fill (P'01) circle (1.0pt);
        \fill (P'00) circle (1.0pt);
        \fill (P'11) circle (1.0pt);

                \fill (P000) circle (1.0pt);
        \fill (P'100) circle (1.0pt);
        \fill (P'101) circle (1.0pt);
        \fill (P'110) circle (1.0pt);
        \fill (P'001) circle (1.0pt);
        \fill (P'011) circle (1.0pt);
        \fill (P'010) circle (1.0pt);
    \end{scope}

\end{tikzpicture}
    \caption{Projections of a cube onto two planes:
A projection along the cube's edge collapses pairs of vertices, so each corner of the resulting square traces back to the whole edge. By contrast, a generic projection forms a hexagon whose corners trace back to a single vertex of the cube.} 
    \label{fig:cube}
\end{figure}

The proof of Proposition~\ref{prop:payoff_extreme} appears in Appendix~\ref{app_payoff_space}. To build intuition for this result, observe that the set $U^{CE}(\G)$ arises as the linear projection of the correlated equilibrium polytope onto a lower-dimensional payoff space via the mapping $\mu \to u(\mu)$. Much like the sharpest points of a shadow must trace back to the sharpest points or edges of the object casting it, the extreme points of a projection of a convex set are contained in the projection of the set’s extreme points. Accordingly, the extreme payoffs in $U^{CE}(\G)$ trace back to extreme correlated equilibria under the payoff mapping. However, there is no a priori guarantee that an entire edge, or a lower-dimensional face, of the correlated equilibrium polytope will project to an extreme payoff. Just as in Figure~\ref{fig:cube}, a vertical edge of the cube collapses to a single corner of its square shadow.

Nonetheless, under a generic projection, the sharpest points of an object’s shadow trace back to unique sharp points of the original object, as in the hexagonal shadow shown in Figure~\ref{fig:cube}. Projection to the payoff space is, of course, not a generic projection. However, one can show that for a generic game, each extreme payoff vector corresponds to a unique point in the correlated equilibrium polytope. The key observation is that correlated equilibria are unchanged if we transform each agent's utility function from $u_i(a)$ to 
\begin{equation*}
u_i'(a) = u_i(a) + \delta_i(a_{-i})\quad  \text{for some}\quad \delta_i\colon A_{-i}\to \R.    
\end{equation*}
In other words, adding a function of others' actions yields a strategically equivalent game. While this transformation leaves the set of correlated equilibria intact, it \emph{does} alter the direction in which their payoffs are projected. Although these perturbations result only in a restricted set of directions, they are sufficient to establish that, for a generic game, extreme payoff points can arise only from extreme correlated equilibria.

Now, consider a Nash equilibrium in a generic game involving more than two agents mixing. If it is payoff-extreme, the discussion above implies that it must originate from a single extreme point of the correlated equilibrium polytope, contradicting Theorem~\ref{th_almost_pure}. This establishes the first statement of Proposition~\ref{prop:payoff_extreme}. The need to guarantee a unique projection explains why we need the genericity requirement in the proposition.\footnote{Curiously, the precise restriction depends on whether exactly $3$ or more agents mix. When more than $3$ agents mix, we show that for any regular Nash equilibrium in any game, welfare can be improved in all but a measure-zero set of strategically equivalent games. In other words, with $4$ or more mixers, any regular Nash equilibrium---up to strategic equivalence---admits an improvement. With exactly $3$ mixing agents, however, there exist degenerate strategic interactions for which no improvement is possible even after the transition to a strategically equivalent game. Genericity allows us to exclude such degenerate cases.}

While every extreme point of a projection arises as the image of an extreme point, the converse does not hold. Specifically, not every extreme point projects to an extreme point. For example, two vertices of the cube project to the interior of the hexagon in Figure~\ref{fig:cube}. Accordingly, not every extreme point of the correlated equilibrium polytope yields an extreme point of $U^{CE}$. In particular, pure Nash equilibria and those involving exactly two randomizing agents---while generically extreme in the set of correlated equilibria by Theorem~\ref{th_almost_pure}---are not necessarily payoff-extreme. For instance, the mixed Nash equilibrium in Aumann’s game~\eqref{eq_Aumann_game}, discussed at the end of Section~\ref{sec:model}, lies in the interior of $U^{CE}$ despite being an extreme correlated equilibrium. Whether a Nash equilibrium with at most two agents mixing is payoff-extreme depends on the structure of the game and the equilibrium itself. There is no general, detail-free characterization.\footnote{Relatedly, \cite{lehrer2011equilibrium} show that, in two-agent games, correlated and Nash equilibrium payoffs are essentially unconstrained.}

We now explain why the generic improvability of the utilitarian welfare $W(\nu) =  \sum_i u_i(\nu)$ of a Nash equilibrium $\nu$ with more than two agents mixing follows from non-extremality of such equilibria in payoff space. Since $\nu$ is not an extreme point of $U^{CE}$, the value of any weighted welfare function $W_\alpha$ can be strictly improved via a correlated equilibrium, except when the weight vector $\alpha$ lies in a game-dependent set of measure zero. Hence, the set of pairs $(\Gamma, \alpha)$ for which $W_\alpha$ cannot be improved has measure zero in the product space. By Fubini's theorem, it follows that for almost every $\alpha$, the corresponding set of such games $\Gamma$ has measure zero.\footnote{Fubini's theorem implies that for a subset $E$ of a product space $X\times Y$ with measure $\mu=\mu_X\times \mu_Y$ such that $\mu(E)=0$, the section $E_x=\{y\colon (x,y) \in E\}$ has $\mu_Y$-measure zero for $\mu_X$-almost all $x$.} In particular, there exists some $\alpha^* \in \R_{>0}^n$ such that $W_{\alpha^*}$ can be strictly improved at any Nash equilibrium that is not payoff-extreme in a generic game. Maximizing $W_{\alpha^*}$ in a game with utility functions $u_i$ is equivalent to maximizing the utilitarian welfare $W$ in the rescaled game with utilities $u_i' = \alpha_i^* \cdot u_i$. Thus, the utilitarian welfare of a Nash equilibrium with more than two mixing agents can be strictly improved in a generic game.\footnote{The formal proof, presented in the appendix, proceeds more directly: we establish that the payoff vector of a Nash equilibrium belongs to a face of $U^{CE}(\G)$ not parallel to the hyperplanes $\sum_i u_i=\mathrm{const}$. This guarantees that, along a particular direction, there is a perturbation of the Nash equilibrium that improves utilitarian welfare, and one that decreases it, thus ensuring both claims of the proposition at once.} This conclusion is not specific to the case of equal weights: the same reasoning applies to any fixed weighted welfare objective with strictly positive weights~$\alpha \in \R_{>0}^n$.

\medskip

The welfare improvement in Proposition~\ref{prop:payoff_extreme} can in fact be strengthened to Pareto improvement if the number of mixing agents exceeds a threshold growing logarithmically with the total number of agents~$n$.

\begin{proposition}\label{prop_Pareto}
In a generic game with $n$ agents, for any Nash equilibrium with at least $8+\log_2 n$ mixing agents, there is a correlated equilibrium that strictly improves all agents' utilities. In particular, any such equilibrium is Pareto dominated by a correlated equilibrium.
\end{proposition}
For a large number of agents $n$, this proposition shows that even a negligible fraction of mixing agents renders a Nash equilibrium Pareto dominated.

\begin{figure}
    \centering
\begin{tikzpicture}[scale=1.3,
    color_cube_face/.style={fill=gray!40, fill opacity=0.6},
    color_cube_face_left/.style={fill=gray!60, fill opacity=0.6},
    color_cube_face_top/.style={fill=gray!90, fill opacity=0.6},
    color_cube_edge/.style={thick, draw=black!80},
    color_hidden_edge/.style={thick, draw=gray!70, dashed},
    color_proj_plane/.style={fill=gray!20, fill opacity=0.5},
    color_proj_shape/.style={fill=gray!100},
    color_proj_line/.style={draw=black, thin, densely dotted},
    label_style/.style={font=\tiny, text=black},
    label_style_emph/.style={font=\tiny\bfseries, text=black},
    arrow_style/.style={-Latex, thick},
    color_point_face/.style={fill=red!60!black},
    color_nhood_face/.style={fill=red, fill opacity=0.3, draw=red!50!black, thin},
    color_nhood_proj_face/.style={fill=red, fill opacity=0.3, draw=red!50!black, thin},
    color_nhood_proj/.style={fill=red!60, fill opacity=0.7, draw=red!30!black, thin},
    color_proj_line_emph/.style={draw=red!70!black, thick, densely dotted}
]

\pgfmathsetmacro{\xAngle}{-40}
\pgfmathsetmacro{\yAngle}{20}
\pgfmathsetmacro{\zAngle}{90}

\pgfmathsetmacro{\r}{0.15} 
\pgfmathsetmacro{\nsq}{1.3*1.3 + 0.7*0.7 + 1.0*1.0} 
\pgfmathsetmacro{\ax}{1.3} \pgfmathsetmacro{\ay}{0.7} \pgfmathsetmacro{\az}{1.0} \pgfmathsetmacro{\dval}{7} 
\pgfmathsetmacro{\qx}{1} \pgfmathsetmacro{\qy}{0.5} \pgfmathsetmacro{\qz}{0.5} 
\pgfmathdeclarefunction{projx}{1}{%
  \pgfmathparse{(\qx + ((\dval - (\ax*\qx + \ay*(\qy+\r*cos(#1)) + \az*(\qz+\r*sin(#1))))/\nsq)*\ax)}%
}
\pgfmathdeclarefunction{projy}{1}{%
  \pgfmathparse{(\qy+\r*cos(#1) + ((\dval - (\ax*\qx + \ay*(\qy+\r*cos(#1)) + \az*(\qz+\r*sin(#1))))/\nsq)*\ay)}%
}
\pgfmathdeclarefunction{projz}{1}{%
  \pgfmathparse{(\qz+\r*sin(#1) + ((\dval - (\ax*\qx + \ay*(\qy+\r*cos(#1)) + \az*(\qz+\r*sin(#1))))/\nsq)*\az)}%
}

\begin{scope}[x={({cos(\xAngle)*1cm},{sin(\xAngle)*1cm})}, y={({cos(\yAngle)*1cm},{sin(\yAngle)*1cm})}, z={({cos(\zAngle)*1cm},{sin(\zAngle)*1cm})}]

        \coordinate (P000) at (0,0,0); \coordinate (P100) at (1,0,0);
        \coordinate (P010) at (0,1,0); \coordinate (P110) at (1,1,0);
        \coordinate (P001) at (0,0,1); \coordinate (P101) at (1,0,1);
        \coordinate (P011) at (0,1,1); \coordinate (P111) at (1,1,1);
        \coordinate (P'000) at (2.8616, 1.5409, 2.2013);
        \coordinate (P'100) at (3.3270, 1.2547, 1.7925);
        \coordinate (P'010) at (2.5723, 2.3852, 1.9811);
        \coordinate (P'001) at (2.4528, 1.3208, 2.8868);
        \coordinate (P'110) at (3.0440, 2.0991, 1.5723);
        \coordinate (P'101) at (2.9214, 1.0346, 2.4780);
        \coordinate (P'011) at (2.1667, 2.1667, 2.6667);
        \coordinate (P'111) at (2.6352, 1.8805, 2.2579);
        \coordinate (Q) at (1, 0.5, 0.5);
        \coordinate (Q') at (2.9827, 1.5676, 2.0252);

        \fill[color_proj_plane] (5,1,2) -- (1,5,1) -- (0,5,1) -- (1,1,2) -- (5,0,2) -- cycle;
        
        \filldraw[color_proj_shape] (P'100) -- (P'110) -- (P'010) -- (P'011) -- (P'001) -- (P'101) -- cycle;

        \filldraw[color_nhood_proj] plot[smooth, variable=\t, domain=0:360, samples=50]
            ({projx(\t)}, {projy(\t)}, {projz(\t)});
            
        \draw[color_proj_line] (P100) -- (P'100);
        \draw[color_proj_line] (P010) -- (P'010);
        \draw[color_proj_line] (P001) -- (P'001);
        \draw[color_proj_line] (P110) -- (P'110);
        \draw[color_proj_line] (P101) -- (P'101);
        \draw[color_proj_line] (P011) -- (P'011);
        \draw[color_proj_line, thick] (P111) -- (P'111);
        \draw[color_proj_line_emph] (Q) -- (Q');

        \fill[gray] (P010) circle (1.0pt); 
        \draw[color_hidden_edge] (P000) -- (P010);
        \draw[color_hidden_edge] (P010) -- (P011);
        \draw[color_hidden_edge] (P010) -- (P110);
        \fill[color_nhood_face] (P100) -- (P110) -- (P111) -- (P101) -- cycle; 
        \fill[color_cube_face_top] (P001) -- (P101) -- (P111) -- (P011) -- cycle; 
        \fill[color_cube_face_left] (P000) -- (P100) -- (P101) -- (P001) -- cycle; 
        \draw[color_cube_edge] (P100) -- (P110);
        \draw[color_cube_edge] (P101) -- (P111) -- (P011);
        \draw[color_cube_edge] (P100) -- (P101); \draw[color_cube_edge] (P110) -- (P111); 
        \draw[color_cube_edge] (P001) -- (P101); \draw[color_cube_edge] (P001) -- (P011);
        \draw[color_cube_edge] (P000) -- (P100);
        \draw[color_cube_edge] (P000) -- (P001);

        \fill (P000) circle (1.0pt); 
        \fill (P'010) circle (1.0pt);
        \fill (P001) circle (1.0pt); \fill (P'001) circle (1.0pt);
        \fill (P011) circle (1.0pt); \fill (P'011) circle (1.0pt);

               \fill[color_nhood_proj_face] (P'100) -- (P'101) -- (P'111) -- (P'110) -- cycle;
              
        \filldraw[color_nhood_face] plot[variable=\t, domain=0:360, samples=50]
            (1, {0.5+\r*cos(\t)}, {0.5+\r*sin(\t)});
        \fill[color_point_face] (Q) circle (1.2pt);
        \node[label_style, above right=2pt, color=red!30!black] at (Q) {$\nu$};
        
        \fill[color_point_face] (Q') circle (1.2pt);
                \fill (P100) circle (1.0pt); \fill (P'100) circle (1.0pt);
        \fill (P110) circle (1.0pt); \fill (P'110) circle (1.0pt);
        \fill (P111) circle (1.0pt); \fill (P'111) circle (1.0pt);
        \fill (P101) circle (1.0pt); \fill (P'101) circle (1.0pt);

    \end{scope}

\end{tikzpicture}
\caption{Under a generic projection of a polytope to an $n$-dimensional space, any point in the relative interior of an $n$-dimensional face is projected to the interior of the image.}
 \label{fig:cube_Pareto}
\end{figure}
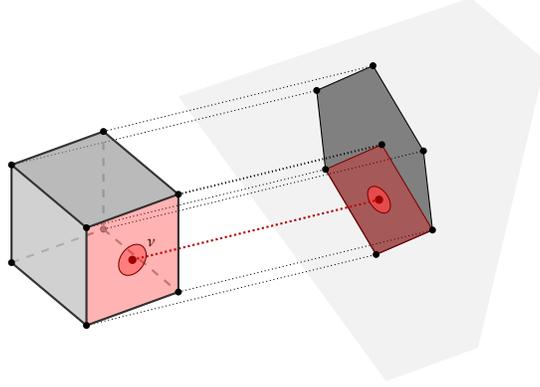
The proposition is proved in Appendix~\ref{app_Pareto} where we also show that the logarithmic order of the bound is tight. The intuition behind the result is geometric: to conclude that a Nash equilibrium $\nu$ admits a Pareto improvement, it suffices to show that the corresponding payoff vector $u(\nu)$ lies in the interior of the correlated equilibrium payoff set $U^{\mathrm{CE}}$. 

Recall that $U^{\mathrm{CE}}$ can be thought of as a projection of the correlated equilibrium polytope onto the $n$-dimensional space of payoffs. Under a generic projection of a polytope to an $n$-dimensional space, interior points of $n$-dimensional faces are mapped to interior points of the image, as illustrated in Figure~\ref{fig:cube_Pareto}.
By formula~\eqref{eq_dimension_bound_intuition}, a Nash equilibrium with $k$ mixing agents lies in the relative interior of a face of the correlated
equilibrium polytope of dimension $d \geq 2^{k-3}$. The image of this face
under a collection of linear objectives $f_1,\ldots, f_d$, viewed as a projection onto
$\R^d$, has nonempty interior whenever the collection is generic,
meaning it does not collapse the dimension of the face.
\begin{corollary}\label{cor_multiple_objectives}
    Take any game and a regular Nash equilibrium with $k\geq 3$ mixing agents. For any collection of at most $2^{k-3}$ generic linear objectives, there is a correlated equilibrium that strictly improves these objectives simultaneously.
\end{corollary}
Consequently, in general, one additional mixing agent increases the number of objectives that can be improved simultaneously by a factor of two. 

Pareto improvements correspond to improving $n$ objectives, $u_1,\ldots, u_n$. From the corollary, if $ k\geq 3 + \log_2 n$ agents randomize, at least $n$ generic objectives can be improved. Hence, provided that the projection onto the payoff space acts as if it were generic, the payoff vector $u(\nu)$ is in the interior of the set of correlated equilibrium payoffs. The genericity requirement we need is that the dimension of this face equals the dimension of its image. As we show in the appendix, the payoff projection behaves in this desired way under the additional assumption that $k\geq 12$.\footnote{In Proposition~\ref{prop:payoff_extreme}, moving to a strategically equivalent game ensured that the payoff projection behaved generically. For Proposition~\ref{prop_Pareto}, this is insufficient because all $n$ payoff dimensions must be tracked simultaneously. The formal proof therefore takes a different route: starting from $S = S_1 \times \cdots \times S_n$, we construct a single game $\G^*$ with a Nash equilibrium $\nu_S$ supported on $S$ whose payoff vector lies in the interior of $U^{\mathrm{CE}}(\G^*)$, and then use algebraic arguments to extend the conclusion to generic games. Such an extension is possible since the game $\G^*$ is constructed so that it has the maximal number of independent incentive constraints active at~$\nu_S$. Achieving this requires at least $12$ mixing agents, which gives rise to the condition $k \geq 12$.} We conclude that the equilibrium admits a Pareto improvement if $k \geq 3 + \log_2 n$ and $k\geq 12$,  conditions satisfied simultaneously by requiring $k \geq 8 + \log_2 n$.

\section{Symmetric Games}\label{sec:symmetric}

Many games considered in the literature exhibit symmetry: an agent's payoffs depend only on her actions and the actions others choose, but not on her or others' identity. Put differently, permuting the labels of agents and, accordingly, the actions they take, makes no difference to payoffs; see, for example, \cite*{dasgupta1986existence}. 

Formally, we say that a game $\G$ is symmetric if it is invariant under all permutations of agents $\pi$. That is, all action sets are identical, and for every permutation $\pi$ of agents we have $u_i(a) = u_{\pi(i)}(a_{\pi(1)}, \ldots, a_{\pi(n)})$. Similarly, a distribution $\mu$ over actions is symmetric if, for any permutation $\pi$, we have $\mu(a)=\mu(a_{\pi(1)},\ldots, a_{\pi(n)})$ for all $a\in A$. 

Consider a symmetric $n$-agent game $\G$ with $n\geq 3$. Let $\nu$ be a symmetric mixed Nash equilibrium, where agents randomize over all available actions, so that $\nu$ is totally mixed. Our Corollary~\ref{cor_fully_mixed_asymmetric} indicates that $\nu$ is not extreme within the set of correlated equilibria. What if symmetry is desirable in and of itself? A designer may want to treat agents symmetrically and favor equity over any other objective. Additionally, if the objective treats agents symmetrically---e.g., the social (utilitarian) welfare---its optimum can be achieved via symmetric correlated equilibria. As we show in Appendix~\ref{app_symm_games}, a totally mixed equilibrium is not extreme even within the set of symmetric correlated equilibria. Thus, we have the following corollary.

\begin{corollary}\label{cor_not_extreme_totally_mixed}
Let $\Gamma$ be a symmetric game with $n\geq 3$ agents. If $\nu$ is a symmetric totally-mixed Nash equilibrium, then $\nu$ is not extreme within the set of correlated equilibria. Furthermore, $\nu$ is not extreme within the set of symmetric correlated equilibria.   
\end{corollary}

We now show that extreme symmetric correlated equilibria take a very particular form, regardless of the underlying game's details. The key observation is that any symmetric correlated equilibrium corresponds to an exchangeable distribution over actions. That is, the distribution of $(a_1,\ldots, a_n)$ coincides with that of $(a_{\pi(1)},\ldots, a_{\pi(n)})$ for any permutation $\pi$.

\begin{proposition}\label{prop_structure_of_extreme_symmetric_CE}
    Any extreme symmetric correlated equilibrium of a symmetric game $\G$ with $n$ agents and $m$ actions $a_1,\ldots, a_m$ arises from a fixed collection of $M\leq  m(m-1)+1$ urns, each containing $n$ balls labeled with actions, as follows:
\begin{itemize}
\item An urn is selected at random according to some distribution $p \in \Delta_M$. Agents do not know which urn is selected.
\item Agents approach the selected urn in some order and draw balls without replacement.
\item Each agent takes an action matching her ball’s label.
\end{itemize}
\end{proposition}
The proposition is proved in Appendix~\ref{app_symm_games}. Intuitively, the set of exchangeable distributions over actions is convex, so any such distribution can be expressed as a mixture of extreme exchangeable distributions. For infinite sequences of exchangeable random variables, the famous de Finetti theorem characterizes the extreme points as i.i.d. product distributions. For finite $n$, the extreme exchangeable distributions correspond to sampling without replacement from an urn with a given composition \citep{diaconis1980finite}. Hence, any exchangeable distribution over action profiles corresponds to a mixture over urns. 

Symmetric correlated equilibria form a convex subset of the set of exchangeable distributions, defined by incentive constraints. By symmetry, it is enough to impose these constraints on a single agent. For each action taken with positive probability, there are at most $m - 1$ constraints---one for each possible deviation---yielding a total of at most $m(m - 1)$ constraints. By \citet{winkler1988extreme}, any extreme point of this constrained set must be a mixture over at most $m(m - 1) + 1$ extreme exchangeable distributions. This gives the desired representation.

The proposition provides a simple parametrization of extreme symmetric correlated equilibria, particularly convenient when the number of actions is small. We illustrate the implications of the proposition in the following section, where we explicitly construct improving correlated equilibria in binary-action games. Indeed, when $m=2$, we only need $M=3$ urns. The composition of each urn is determined by a single number: the fraction of balls marked $a_1$ in the urn. A distribution $p$ over three urns adds $2$ more parameters. As a result, all extreme symmetric correlated equilibria form a $5$-parametric family. More generally, we need $(m^2+1)(m-1)$ parameters. Importantly, this number does not depend on the number of agents $n$.

Proposition \ref{prop_structure_of_extreme_symmetric_CE} implies further simplification when the number of agents is large relative to the number of available actions. In that case, the joint distributions corresponding to draws with and without replacement become close to one another. Thus, we are back to the framework of the classical de~Finetti theorem with infinite products of i.i.d. distributions instead of urns. 

Formally, a collection of random variables $\xi_1,\ \xi_2,\ldots, \xi_n$ with values in an $m$-element finite set is said to be {$\varepsilon$-close} to a mixture of i.i.d. distributions if, for any $k=1,\ldots, n$, the total variation distance between the joint distribution of any $k$ random variables $\xi_{i_1},\ldots, \xi_{i_k}$ and the closest mixture of i.i.d. distributions is at most $\varepsilon k$.\footnote{The total variation distance between two probability distributions is defined as half the sum of the absolute differences between probabilities across all possible outcomes.} \cite*{diaconis1980finite} showed that any exchangeable sequence of length~$n$  is {$\varepsilon$-close} to a mixture of i.i.d. distributions with $\varepsilon=2m/n$. Combining their insight with Proposition~\ref{prop_structure_of_extreme_symmetric_CE}, we obtain the following corollary.

\begin{corollary}\label{th_de_Finetti_CE}
Any symmetric correlated equilibrium of a symmetric game with $n$ agents, each having $m$ strategies, is $\varepsilon$-close to a mixture of i.i.d. distributions with
$\varepsilon={2m}/{n}.$
For extreme symmetric correlated equilibria, mixing at most $m(m-1)+1$ distributions suffices.
\end{corollary}

The corollary implies that, in the large-population limit, the set of symmetric correlated equilibria can be described via the following simple procedure. Pick a set of states $\Theta$ with prior $p \in \Delta(\Theta)$, a signal set $S = A_1 = \ldots = A_n$ equal to the common action set, and a family of distributions $\nu_\theta \in \Delta(S)$ for each $\theta \in \Theta$.
The state $\theta$ is realized according to the distribution $p$, and i.i.d. signals $s_i\sim \nu_\theta$ are privately released to each agent $i$. Agents, knowing the distributions over states and signals, best reply to their observed signals. 

Thus, symmetric correlated equilibria coincide asymptotically with the Bayesian Nash equilibria of a game obtained after adding a payoff-irrelevant random state variable $\theta$ with conditionally independent and identically distributed private signals. We use this observation in the following section.

\section{Construction of Improving Correlated Equilibria}\label{sec_binary_action}

When more than two agents mix in a Nash equilibrium, it is generically improvable. How do we identify a superior correlated equilibrium? 

In this section, we focus on utilitarian and Pareto welfare objectives and consider symmetric binary-action games, a class that captures a range of settings---including public goods provision, political participation (e.g., protests), and congestion scenarios---where individuals choose between two alternatives. In a rich subset of these environments, Nash equilibria are inherently sub-optimal: \textit{no} Nash equilibrium achieves utilitarian efficiency. We present a geometric approach for constructing desirable correlated equilibria.

Formally, we study the class of symmetric games in which each agent chooses an action from $\{0,1\}$.
For simplicity, we focus on games where the payoff of action~$0$ is normalized to~$0$. Thus, the game can be parameterized by a function $f\colon [0,1]\to \R$ that maps the overall fraction of agents choosing action 1 to the resulting payoff from that action:
\begin{equation}\label{eq_binary_congestion}
u_i(a_i,a_{-i})=\begin{cases}
    f\left(\frac{|\{j\in N\colon a_j=1\}|}{n}\right),& a_i=1\\
    0, & a_i=0.
\end{cases}
\end{equation}
Any symmetric binary-action game is strategically equivalent to a game of this form.\footnote{In general, agent $i$'s choice between the two actions is driven by the payoff difference $u_i(1,a_{-i})-u_i(0,a_{-i})$. Setting the payoff of action $0$ to $0$ and the payoff of $1$ to this difference, we maintain all agents' incentives and get a game of the form of~\eqref{eq_binary_congestion}. Naturally, for objectives that do not respect strategic equivalence, such as utilitarian welfare considerations, utility from the $0$-action may influence assessments. Our approach can accommodate general utilities for the $0$-action, though doing so entails more tedious calculations.} 

We are interested in the large-population behavior of the game. We assume that $f$ takes both positive and negative values, $f(1)<0$, and $f$ is continuous. These requirements reflect that action $1$ is desirable only if not too many agents choose it, and small shifts in behavior have limited impact on others' payoffs. 

Such payoff structures arise naturally across a range of settings. In models of political protests, individuals choose whether to participate, with success determined by turnout \citep{chwe1999structure}; when participation is costly, high turnout can generate free-riding incentives, reducing the net value of joining. In public goods provision, agents decide whether to contribute, with the value of contributions depending on the share of contributors and potentially declining as participation grows \citep{kurz1994game}. Similarly, in congestion environments, agents choose between routes whose travel times depend on the fraction of users, so that heavily used routes become increasingly unattractive \citep{rosenthal1973class}.

In this class of games, decentralized outcomes exhibit inherent inefficiencies when the population grows large.

\begin{proposition}\label{prop:binary_payoffs}
As~$n\to \infty$, all agents' equilibrium payoffs at all Nash equilibria converge to $0$. 
\end{proposition}

The proposition speaks to \textit{all} equilibria, pure or mixed, symmetric or asymmetric. The proof, contained in Appendix~\ref{app_binary}, also implies that, for Lipschitz continuous $f$ with constant $L$, each agent's equilibrium payoff is between $0$ and~$\frac{L}{n}$ for any finite $n$.
The proof intuition can be explained using Figure~\ref{fig:Nash_binary}. 
 \begin{figure}
 \begin{center}
\begin{tikzpicture}[scale=1.2]
    \draw[->, thick] (-0.5,0) -- (7.5,0);
    \draw[->, thick] (0,-1.0) -- (0,2) node[right] {Payoff};
    \node[below] at (7.0,0) {$1$};
    \node[below] at (-0.15,0) {$0$};

    \draw[thick] (7.0,-0.05) -- (7.0,0.05);
    
    \draw[thick] plot[smooth] coordinates {
        (0,1.5)
        (1,0) 
        (2,-1) 
        (3,0.5)
        (3.5,0)
        (4,-0.5) 
        (5,0)
        (6,0.5)
        (6.5,0)
        (7.0,-1)
    };
    
    \foreach \x/\y in {1/0, 3.5/0, 6.5/0} {
        \filldraw[black] (\x,\y) circle (0.1);
    };
    
    \begin{scope}
        \clip (0,0) -- (0,2.5) -- 
            plot[smooth] coordinates {
                (0,1.5) (1,0)
                (2,-1)(3,0.5) (3.5,0) (4,-0.5)
                (5,0) (6,0.5) (6.5,0)} --
            cycle;
        \fill[pattern=north east lines, pattern color=black, opacity=0.7] 
            (-0.5,0) rectangle (7.5,3);
    \end{scope}

    \node[right] at (7.5,0) {Fraction};
    \node[right] at (7.5,-0.4) {choosing 1};
\end{tikzpicture}

\end{center}
     \caption{{Payoff from action $1$ as a function of the fraction of agents playing this action. Highlighted zeros correspond to stable Nash equilibria, others to unstable.}}
     \label{fig:Nash_binary}
 \end{figure}
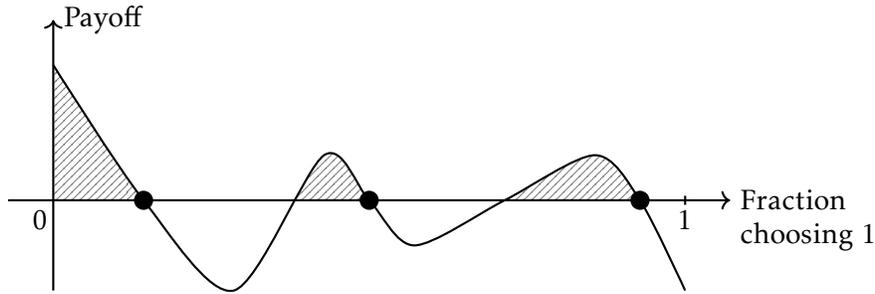
For profiles of actions falling into the shaded areas, agents choosing the action $0$ would have an incentive to switch to action $1$. Similarly, for profiles falling into the blank areas below the axis, where $f<0$, agents choosing $1$ would benefit from switching to action~$0$. Accordingly, limit Nash equilibria correspond to zeros of $f$. The highlighted downward intersection points reflect stable equilibria that can be closely matched by pure asymmetric equilibria. The upward intersection points represent unstable mixed equilibria, which cannot be approximated via pure Nash equilibria in the finite game. 

Any symmetric Nash equilibrium in which agents choose action~$1$ with positive probability necessarily involves all agents using mixed strategies. By Corollary~\ref{cor_not_extreme_totally_mixed}, such equilibria are not extreme, and hence, any non-degenerate objective can be improved. As we now show, this conclusion applies to the utilitarian welfare objective. Specifically, we explicitly construct symmetric correlated equilibria that provide a strict improvement over \textit{any} Nash equilibrium, or convex combinations thereof. 

Our approach combines dimension reduction via the de Finetti theorem described in the previous section, with a geometric insight connecting the optimum to the convex hull of an auxiliary set, an idea reminiscent of \cite*{aumann1995repeated} and \cite*{kamenica2011bayesian}. While we focus on utilitarian welfare for illustrative purposes, the approach can be extended to other objectives.

By Corollary \ref{th_de_Finetti_CE}, to identify a beneficial symmetric correlated equilibrium, we can focus on equilibria that are generated via a two-step process. First, the fraction $\theta$ of agents choosing $1$ is realized from some distribution. Second, a fraction $\theta$ of agents, determined uniformly at random, is prescribed the action $1$, while others are prescribed the action $0$. 

The per-capita utilitarian welfare is given by $W=\E[\theta f(\theta)]$, where $\theta$ is the (random) fraction of agents choosing the action $1$. The expectation is taken with respect to the distribution of $\theta$, which we aim to determine. Any symmetric correlated equilibrium satisfies two constraints. First, any agent prescribed the action $1$ should prefer it to action $0$. Being prescribed the action $1$ clearly provides an agent useful information on the realized fraction of agents directed at the same action. Using Bayes' rule, the incentive constraint for an agent prescribed the action $1$ is $\E[\theta f(\theta)]\geq 0$. Similarly, the constraint for an agent prescribed the action $0$ is  $\E[(1-\theta)f(\theta)] \leq 0$.
The first of these two constraints can be rewritten as $W\geq 0$. Since the value of $0$ is attained at Nash equilibria, the first incentive constraint is automatically satisfied at the optimum. Thus, our ultimate design problem is:
\begin{equation*}
\text{ max } W=\E[\theta f(\theta)] \quad \text{over distributions on $[0,1]$ subject to } \quad \E[(1-\theta)f(\theta)] \leq 0. 
\end{equation*}
We denote $IC=\E[(1-\theta)f(\theta)]$. To find the optimal $W$, we describe all the possible joint \emph{feasible} values $(W, IC)\in \R^2$; that is, pairs that arise for some distribution of $\theta$.
Consider the curve $\varphi\colon [0,1]\to \R^2$ given by $\varphi(\theta)=\big(\theta f(\theta),\ (1-\theta)f(\theta)\big)$ and let $\Phi\subset \R^2$ be its image:
\begin{equation*} \Phi=\Big\{\varphi(\theta)\in\R^2\ \colon \  \theta\in [0,1]  \Big\}.\end{equation*} 
Each point of $\Phi$ corresponds to the feasible pair $(W, IC)$ generated by a point-mass distribution at $\theta$. A general distribution of $\theta$ is a convex combination of point-mass distributions. Therefore, the set of feasible $(W,IC)$ is the convex hull of $\Phi$. That is, a point $(W, IC)$ is feasible if and only if
    $(W, IC)\in \conv[\Phi].$
Finding the optimum under the incentive constraint~$IC\leq 0$ reduces to
\begin{equation*} \text{ {max} } W \quad \text{{over}}\quad (W,IC)\in\conv[\Phi]\cap (\R\times \R_-).\end{equation*} 
This problem can be solved geometrically: one plots the curve $\Phi$, computes its convex hull, and intersects the hull with the lower half-plane. The optimal value $W^*$ is thus the rightmost point of the resulting set.

\begin{figure}
    \centering
\begin{tikzpicture}[
    scale=3.5,
    every node/.style={inner sep=1pt},
    main_curve/.style={draw=black, ultra thick, line cap=round, line join=round},
    axis/.style={draw=black, thick, -latex},
    marker/.style={draw=black, very thick}
]

\tikzset{
    declare function={
        f(\t) = (8*\t*(1-\t) - 1);
        x(\t) = \t * f(\t);
        y(\t) = (1-\t) * f(1-\t);
    }
}

\definecolor{fill_color}{RGB}{150, 240, 240}

\foreach \i in {1,...,100} {
    \fill[black!10] 
        ({x(0)}, {y(0)}) --
        ({x(\i/100)}, {y(\i/100)}) -- 
        ({x(1-\i/100)}, {y(1-\i/100)}) --
        ({x(1)}, {y(1)}) -- cycle;
}

\begin{scope}
    \clip (-1.2, -1.2) rectangle (1, 0); 
    \foreach \i in {1,...,100} {
    \fill[pattern=north west lines, pattern color=black!50]
        ({x(0)}, {y(0)}) --
        ({x(\i/100)}, {y(\i/100)}) -- 
        ({x(1-\i/100)}, {y(1-\i/100)}) --
        ({x(1)}, {y(1)}) -- cycle;
}
\end{scope}

\node[anchor=south west] at (0.15, 0.65) {$\conv[\Phi]\cap (\R\times\R_-)$};

\draw[main_curve, variable=\t, domain=0:1, samples=201, smooth, black!90]
    plot ({x(\t)}, {y(\t)});

\draw[axis] (-1.2, 0) -- (0.9, 0) node[below=2pt] {W};
\draw[axis] (0, -1.2) -- (0, 0.9) node[left=2pt] {IC};

\node[fill=black, circle, minimum size=4pt] at (0,0) {};

\coordinate (P_start)  at ({x(1)}, {y(1)});
\coordinate (P_end)    at ({x(0)}, {y(0)});
\coordinate (P_star)   at ({x(0.65)}, {y(0.65)});
\coordinate (P_star_inv)   at ({x(1-0.65)}, {y(1-0.65)});
\coordinate (P_firstbest)   at ({x(0.59685647168)}, {y(0.59685647168)});
\coordinate (P_firstbest_welfare)   at ({x(0.59685647168)}, 0);

\draw[black!90] (P_end) -- (P_star);
\draw[black!90] (P_start) -- (P_star_inv);
\draw[black!90] (P_start) -- (P_end);

\node[fill=black, circle, minimum size=6pt] at (P_start) {};
\node[anchor=north, yshift=-2pt] at ($(P_start)-(0.05,0.02)$) {$\varphi(1)$};

\node[fill=black, circle, minimum size=6pt] at (P_end) {};
\node[anchor=west, xshift=6pt] at (P_end) {$\varphi(0)$};

\node[fill=black, circle, minimum size=6pt] at (P_star) {};
\node[anchor=west, xshift=3pt] at (P_star) {$\varphi(\theta^*)$};

\node[fill=black, circle, minimum size=6pt] at (P_firstbest) {};
\node[anchor=west, xshift=3pt, yshift=2] at (P_firstbest) {first best};

\coordinate (Intersection) at (0.41, 0);
\node[marker, fill=white, circle, minimum size=6pt] at (Intersection) {};
\node (w_label) at ($(Intersection)+(0.13,-0.1)$) {W$^{*}$};
\end{tikzpicture}
    \caption{The set of feasible $(W,IC)$ and the optimal welfare $W^{*}$ for $f(\theta)=8\theta(1-\theta)-1$.   
    }
    \label{fig:optimal_W}
\end{figure}
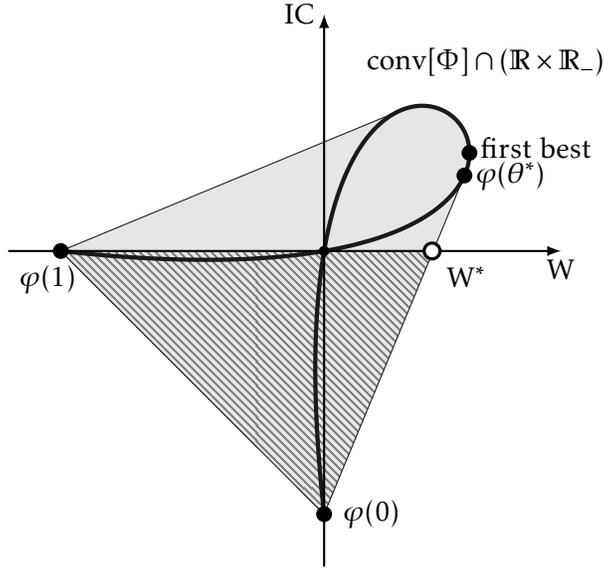

To illustrate this procedure, Figure~\ref{fig:optimal_W} depicts the curve $\Phi$, its convex hull, and the optimal welfare $W^*$ for $f(\theta)=8\theta(1-\theta)-1$. While the function~$f$ is chosen arbitrarily, it reflects settings in which the action $1$ is desirable if the fraction of others choosing~$1$ is not too low nor too high, roughly between $0.15$ and $0.85$. Such a concave~$f$ could emerge if action $1$ is costly with increasing marginal costs and decreasing marginal benefits.

The underlying optimal correlated equilibrium can be read from the figure. Indeed, the optimal welfare $W^*$ belongs to the segment of the convex hull spanned by~$\varphi(0)$ and $\varphi(\theta^*)$, where $\varphi(\theta^*)$ is in the positive quadrant. Since $IC=0$ at $W^*$---meaning that the incentive constraint binds---the weight $p^*$ of the point $\theta^*$ in the convex combination can be determined from the identity $(1-p^*)f(0)+p^*(1-\theta^*)f(\theta^*)=0$. In this example, we get $W^*=\sqrt{2}-1\approx 0.41$, attained at $\theta^*=1-\sqrt{2}/4\approx 0.65$ and $p^*=(4+\sqrt{2})/7\approx 0.77$.  

Intuitively, the correlated equilibrium we construct prescribes either that all agents choose action~$0$ or that a fraction~$\theta^*$ of agents choose action~$1$. An agent who receives the recommendation to choose~$1$ therefore infers that a fraction~$\theta^*$ of the population does the same. By construction, since $f(\theta^*)>0$, it is in her interest to follow this recommendation. An agent recommended to choose~$0$, however, cannot infer the population share of each action with certainty. The value~$\theta^*$ is selected to be high enough that the possibility that others are choosing~$1$, occurring with probability~$p^*$, deters her from deviating. Naturally, the resulting welfare~$W^*$ falls short of the first best (unconstrained) welfare, which maximizes~$\theta f(\theta)$. Indeed, $\theta^* \approx 0.65$ exceeds the value of $\theta$ corresponding to the first best: $\arg\max_{\theta} \theta f(\theta)=\frac{4 + \sqrt{10}}{12} \approx 0.60$.

The construction illustrated in the example applies more generally. Identical arguments generate the following corollary.

\begin{corollary}
    Let $f$ be any continuous function satisfying $f(0),f(1)<0$ and having precisely two zeros in the interval $[0,1]$. Then, an optimal correlated equilibrium for the induced game places probability $p^*$ on $\theta^*$ and $1-p^*$ on $0$, where $p^*$ and $\theta^*$ solve the following problem:
    \begin{equation*}  \max \ p \theta f(\theta)\quad \text{\emph{over \ \  $(\theta,p)\ \in \ [0,1]^2$ \ \ subject to }} \quad (1-p)f(0)+p(1-\theta)f(\theta)=0.\end{equation*} 
    Furthermore, if $f$ is differentiable, $\theta^*$ is strictly above the first best, $\arg\max\, \theta f(\theta)$.
\end{corollary}

For functions $f$ that change sign more than twice, the solution may instead involve randomization between two interior points $\theta^*_1$ and $\theta^*_2$. These points can be identified using the geometric method described above.\footnote{The sufficiency of a distribution supported on two points follows from Carath\'eodory's theorem, which states that any point in the convex hull of a $k$-dimensional set can be expressed as a convex combination of at most $k+1$ points. The relevant set is a (one-dimensional) boundary of $\conv[\Phi]$.} In our setting, $\conv[\Phi]$ is two-dimensional because, under the utilitarian welfare objective, the incentive constraint $\E[\theta f(\theta)] \ge 0$ is automatically satisfied at the optimum. For alternative objectives, however, both incentive constraints may bind. In that case, an additional coordinate,  $IC' = \E[\theta f(\theta)]$, must be introduced, yielding a three-dimensional set $\Phi \subset \R^3$. Consequently, the optimal solution may require randomization over three values of~$\theta$.

For finite and possibly small $n$, the fraction $\theta$ of agents taking any action is constrained to be a rational number with denominator $n$. With this modification, the same approach applies in this setting. In particular, the correlated equilibrium we constructed above can be approximated via a sequence of correlated equilibria in games played in finite populations of size $n$. 

More generally, the approach extends naturally to settings with more than two actions. For a symmetric game with $m$ actions per agent, one introduces a coordinate for each of the $\ell\leq m(m-1)$ incentive constraints that may bind at the optimum, as well as a coordinate for the objective value. This yields a set $\Phi\subset \R^{1+\ell}$, where each point represents a feasible tuple arising from a deterministic assignment of actions. The optimal correlated equilibrium can be identified geometrically: compute $\conv[\Phi]$, intersect it with the halfspaces $IC_k\leq 0$ for $k=1,\ldots,\ell$ corresponding to the incentive constraints, and identify the point with the maximal objective value $W$ in the resulting region. 

\section{Discussion}

In our analysis, we focused on the improvability of Nash equilibria via correlation. While the previous literature identified various settings where Nash equilibria are extreme, our paper highlights a general link between randomness in Nash equilibria and their improvability, regardless of the designer's objective. Our analysis is detail-free in that the characterization of improvable Nash equilibria does not depend on specific features of the game or the designer’s objective.

It would be interesting to extend the analysis to games with a continuum of actions. While our focus is on finite games, our arguments suggest that as agents randomize over richer action sets, equilibria lie on higher-dimensional faces of the correlated equilibrium polytope, opening additional directions for improvement. This points to potentially greater scope for improvability in continuous-action settings. However, even the existence of Nash equilibrium in these settings is not guaranteed. Therefore, such extensions must be tailored to specific classes of games and inevitably would no longer be fully detail-free.

There are several other natural directions for extending our analysis: focusing on particular classes of games or objectives and assessing the gains from correlation, incorporating incomplete information, and considering alternative equilibrium notions.\footnote{For example, the follow-up paper by \cite{sandomirskiy2026delegation} considers coarse correlated equilibria (CCE), where agents commit ex ante to follow a mediator’s recommendation. By definition, $\ce \subseteq \mathrm{CCE}$, implying that Nash equilibria are, if anything, even less likely to be extreme.} We hope the tools we develop prove useful for such future investigations.

\newpage

\bibliographystyle{ecta}	
\bibliography{main}

@article{winkler1988extreme,
  title={Extreme points of moment sets},
  author={Winkler, Gerhard},
  journal={Mathematics of Operations Research},
  volume={13},
  number={4},
  pages={581--587},
  year={1988},
  publisher={INFORMS}
}

@article{kurz1994game,
  title={Game theory and public economics},
  author={Kurz, Mordecai},
  journal={Handbook of Game Theory with Economic Applications},
  volume={2},
  pages={1153--1192},
  year={1994},
  publisher={Elsevier}
}

@book{aumann1995repeated,
  title={Repeated games with incomplete information},
  author={Aumann, Robert J and Maschler, Michael and Stearns, Richard E},
  year={1995},
  publisher={MIT press}
}

@article{rosenthal1974correlated,
  title={Correlated equilibria in some classes of two-person games},
  author={Rosenthal, Robert W},
  journal={International Journal of Game Theory},
  volume={3},
  number={3},
  pages={119--128},
  year={1974},
  publisher={Springer}
}

@article{chin1974structure,
  title={Structure of equilibria in N-person non-cooperative games},
  author={Chin, HH and Parthasarathy, T and Raghavan, TES},
  journal={International Journal of Game Theory},
  volume={3},
  number={1},
  pages={1--19},
  year={1974},
  publisher={Springer}
}

@article{kamenica2011bayesian,
  title={Bayesian persuasion},
  author={Kamenica, Emir and Gentzkow, Matthew},
  journal={American Economic Review},
  volume={101},
  number={6},
  pages={2590--2615},
  year={2011},
  publisher={American Economic Association}
}

@article{chwe1999structure,
  title={Structure and strategy in collective action},
  author={Chwe, Michael Suk-Young},
  journal={American journal of sociology},
  volume={105},
  number={1},
  pages={128--156},
  year={1999},
  publisher={The University of Chicago Press}
}

@article{rosenthal1973class,
  title={A class of games possessing pure-strategy Nash equilibria},
  author={Rosenthal, Robert W},
  journal={International Journal of Game Theory},
  volume={2},
  pages={65--67},
  year={1973},
  publisher={Physica-Verlag}
}

@inproceedings{kolumbus2022auctions,
  title={Auctions between regret-minimizing agents},
  author={Kolumbus, Yoav and Nisan, Noam},
  booktitle={Proceedings of the ACM Web Conference 2022},
  pages={100--111},
  year={2022}
}

@article{aumann1987correlated,
  title={Correlated equilibrium as an expression of Bayesian rationality},
  author={Aumann, Robert J},
  journal={Econometrica: Journal of the Econometric Society},
  pages={1--18},
  year={1987},
  publisher={JSTOR}
}

@book{guide2006infinite,
  title={Infinite dimensional analysis: A hitchhiker's guide},
  author={Border, Kim},
  year={2006},
  publisher={Springer}
}

@article{gerardi2007deliberative,
  title={Deliberative voting},
  author={Gerardi, Dino and Yariv, Leeat},
  journal={Journal of Economic theory},
  volume={134},
  number={1},
  pages={317--338},
  year={2007},
  publisher={Elsevier}
}

@article{harsanyi1973oddness,
  title={Oddness of the number of equilibrium points: a new proof},
  author={Harsanyi, John C},
  journal={International Journal of Game Theory},
  volume={2},
  pages={235--250},
  year={1973},
  publisher={Springer}
}

@article{neyman1997correlated,
  title={Correlated equilibrium and potential games},
  author={Neyman, Abraham},
  journal={International Journal of Game Theory},
  volume={26},
  pages={223--227},
  year={1997},
  publisher={Springer}
}

@article{nau2004geometry,
  title={On the geometry of Nash equilibria and correlated equilibria},
  author={Nau, Robert and Canovas, Sabrina Gomez and Hansen, Pierre},
  journal={International Journal of Game Theory},
  volume={32},
  pages={443--453},
  year={2004},
  publisher={Springer}
}

@article{viossat2008having,
  title={Is having a unique equilibrium robust?},
  author={Viossat, Yannick},
  journal={Journal of Mathematical Economics},
  volume={44},
  number={11},
  pages={1152--1160},
  year={2008},
  publisher={Elsevier}
}

@inproceedings{feldman2016correlated,
  title={Correlated and coarse equilibria of single-item auctions},
  author={Feldman, Michal and Lucier, Brendan and Nisan, Noam},
  booktitle={Web and Internet Economics: 12th International Conference, WINE 2016, Montreal, Canada, December 11-14, 2016, Proceedings 12},
  pages={131--144},
  year={2016},
  organization={Springer}
}

@article{jann2015correlated,
  title={Correlated equilibria in homogeneous good Bertrand competition},
  author={Jann, Ole and Schottm{\"u}ller, Christoph},
  journal={Journal of Mathematical Economics},
  volume={57},
  pages={31--37},
  year={2015},
  publisher={Elsevier}
}

@inproceedings{wu2008correlated,
  title={Correlated equilibrium of {B}ertrand competition},
  author={Wu, John},
  booktitle={Internet and Network Economics: 4th International Workshop, WINE 2008, Shanghai, China, December 17-20, 2008. Proceedings 4},
  pages={166--177},
  year={2008},
  organization={Springer}
}

@article{lehrer2011equilibrium,
  title={Equilibrium payoffs of finite games},
  author={Lehrer, Ehud and Solan, Eilon and Viossat, Yannick},
  journal={Journal of Mathematical Economics},
  volume={47},
  number={1},
  pages={48--53},
  year={2011},
  publisher={Elsevier}
}

@article{einy2022strong,
  title={Strong robustness to incomplete information and the uniqueness of a correlated equilibrium},
  author={Einy, Ezra and Haimanko, Ori and Lagziel, David},
  journal={Economic Theory},
  volume={73},
  number={1},
  pages={91--119},
  year={2022},
  publisher={Springer}
}

@article{pavlov2023correlated,
  title={Correlated equilibria and communication equilibria in all-pay auctions},
  author={Pavlov, Gregory},
  journal={Review of Economic Design},
  pages={1--33},
  year={2023},
  publisher={Springer}
}

@article{ashlagi2008value,
  title={On the value of correlation},
  author={Ashlagi, Itai and Monderer, Dov and Tennenholtz, Moshe},
  journal={Journal of Artificial Intelligence Research},
  volume={33},
  pages={575--613},
  year={2008}
}

@book{vives1999oligopoly,
  title={Oligopoly Pricing: Old Ideas and New Tools},
  author={Vives, X},
  year={1999},
  publisher={MIT Press}
}

@article{palfrey1984participation,
  title={Participation and the provision of discrete public goods: a strategic analysis},
  author={Palfrey, Thomas R and Rosenthal, Howard},
  journal={Journal of public Economics},
  volume={24},
  number={2},
  pages={171--193},
  year={1984},
  publisher={Elsevier}
}

@article{diekmann1985volunteer,
  title={Volunteer's dilemma},
  author={Diekmann, Andreas},
  journal={Journal of conflict resolution},
  volume={29},
  number={4},
  pages={605--610},
  year={1985},
  publisher={Sage Publications 275 South Beverly Drive, Beverly Hills, California 90212}
}

@article{roberson2006colonel,
  title={The {C}olonel {B}lotto game},
  author={Roberson, Brian},
  journal={Economic Theory},
  volume={29},
  number={1},
  pages={1--24},
  year={2006},
  publisher={Springer}
}

@article{papadimitriou2008computing,
  title={Computing correlated equilibria in multi-player games},
  author={Papadimitriou, Christos H and Roughgarden, Tim},
  journal={Journal of the ACM (JACM)},
  volume={55},
  number={3},
  pages={1--29},
  year={2008},
  publisher={ACM New York, NY, USA}
}

@article{aumann1974subjectivity,
  title={Subjectivity and correlation in randomized strategies},
  author={Aumann, Robert J},
  journal={Journal of mathematical Economics},
  volume={1},
  number={1},
  pages={67--96},
  year={1974},
  publisher={Elsevier}
}

@article{peeters1999structure,
  title={On the structure of the set of correlated equilibria in two-by-two bimatrix games},
  author={Peeters, RJAP and Potters, JAM},
  year={1999}
}

@article{dasgupta1986existence,
  title={The existence of equilibrium in discontinuous economic games, I: Theory},
  author={Dasgupta, Partha and Maskin, Eric},
  journal={The Review of economic studies},
  volume={53},
  number={1},
  pages={1--26},
  year={1986},
  publisher={Wiley-Blackwell}
}

@article{hendrickx2002relation,
  title={A relation between Nash equilibria and correlated equilibria},
  author={Hendrickx, Ruud and Peeters, Ronald and Potters, Jos},
  journal={International game theory review},
  volume={4},
  number={04},
  pages={405--413},
  year={2002},
  publisher={World Scientific}
}

@article{ahunbay2024uniqueness,
  title={On the Uniqueness of Bayesian Coarse Correlated Equilibria in Standard First-Price and All-Pay Auctions},
  author={Ahunbay, Mete {\c{S}}eref and Bichler, Martin},
  journal={arXiv preprint arXiv:2401.01185},
  year={2024}
}

@article{lopomo2011bidder,
  title={Bidder collusion at first-price auctions},
  author={Lopomo, Giuseppe and Marx, Leslie M and Sun, Peng},
  journal={Review of Economic Design},
  volume={15},
  number={3},
  pages={177--211},
  year={2011},
  publisher={Springer}
}

@article{kreps1981,
  title={Finite N-person non-cooperative games with unique equilibrium points},
  author={Kreps, Victoria},
  journal={International Journal of Game Theory },
  volume={10},
  pages={125--129},
  year={1981},
  publisher={Springer}
}

@article{he2021private,
  title={Private private information},
  author={He, Kevin and Sandomirskiy, Fedor and Tamuz, Omer},
    journal={Journal of Political Economy},
  volume={forthcoming},
  number={},
  pages={},
  year={2025},
  publisher={The University of Chicago Press Chicago, IL}
}

@techreport{moulin2014coarse,
  title={Coarse correlated equilibria in an abatement game},
  author={Moulin, Herve and Ray, Indrajit and Gupta, Sonali Sen},
  year={2014},
  institution={Cardiff Economics Working Papers}
}

@article{gerard1978correlation,
  title={Correlation and duopoly},
  author={G{\'e}rard-Varet, Louis-Andr{\'e} and Moulin, Herv{\'e}},
  journal={Journal of economic theory},
  volume={19},
  number={1},
  pages={123--149},
  year={1978},
  publisher={Elsevier}
}

@article{moulin1978strategically,
  title={Strategically zero-sum games: the class of games whose completely mixed equilibria cannot be improved upon},
  author={Moulin, Herv{\'e} and Vial, J -P},
  journal={International Journal of Game Theory},
  volume={7},
  pages={201--221},
  year={1978},
  publisher={Springer}
}

@article{diaconis1980finite,
  title={Finite exchangeable sequences},
  author={Diaconis, Persi and Freedman, David},
  journal={The Annals of Probability},
  pages={745--764},
  year={1980},
  publisher={JSTOR}
}

@article{mckelvey1997maximal,
  title={The maximal number of regular totally mixed Nash equilibria},
  author={McKelvey, Richard D and McLennan, Andrew},
  journal={Journal of Economic Theory},
  volume={72},
  number={2},
  pages={411--425},
  year={1997},
  publisher={Elsevier}
}

@article{wilson1971computing,
  title={Computing equilibria of n-person games},
  author={Wilson, Robert},
  journal={SIAM Journal on Applied Mathematics},
  volume={21},
  number={1},
  pages={80--87},
  year={1971},
  publisher={SIAM}
}

@article{govindan2003short,
  title={A short proof of Harsanyi's purification theorem},
  author={Govindan, Srihari and Reny, Philip J and Robson, Arthur J},
  journal={Games and Economic Behavior},
  volume={45},
  number={2},
  pages={369--374},
  year={2003},
  publisher={Citeseer}
}

@article{harsanyi1973games,
  title={Games with randomly disturbed payoffs: A new rationale for mixed-strategy equilibrium points},
  author={Harsanyi, John C},
  journal={International journal of game theory},
  volume={2},
  number={1},
  pages={1--23},
  year={1973},
  publisher={Springer}
}

@book{van1991stability,
  title={Stability and perfection of Nash equilibria},
  author={Van Damme, Eric},
  volume={339},
  year={1991},
  publisher={Springer}
}

@article{evangelista1996note,
  title={A note on correlated equilibrium},
  author={Evangelista, Fe S and Raghavan, TES},
  journal={International Journal of Game Theory},
  volume={25},
  pages={35--41},
  year={1996},
  publisher={Springer}
}

@article{nash1950non,
  title={Non-cooperative games},
  author={Nash, John F},
  year={1950},
  publisher={Princeton University}
}

@article{calvo2006,
  title={The set of correlated equilibria of 2x2 games},
  author={Calvó-Armengol, Antoni 
},
  journal={mimeo},
  year={2006},
}

@article{liu1996correlated,
  title={Correlated equilibrium of Cournot oligopoly competition},
  author={Liu, Luchuan},
  journal={Journal of Economic Theory},
  volume={68},
  number={2},
  pages={544--548},
  year={1996},
  publisher={Elsevier}
}

@article{agranov2018collusion,
  title={Collusion through communication in auctions},
  author={Agranov, Marina and Yariv, Leeat},
  journal={Games and Economic Behavior},
  volume={107},
  pages={93--108},
  year={2018},
  publisher={Elsevier}
}

@article{goeree2011experimental,
  title={An experimental study of collective deliberation},
  author={Goeree, Jacob K and Yariv, Leeat},
  journal={Econometrica},
  volume={79},
  number={3},
  pages={893--921},
  year={2011},
  publisher={Wiley Online Library}
}

@article{kleiner2021extreme,
  title={Extreme points and majorization: Economic applications},
  author={Kleiner, Andreas and Moldovanu, Benny and Strack, Philipp},
  journal={Econometrica},
  volume={89},
  number={4},
  pages={1557--1593},
  year={2021},
  publisher={Wiley Online Library}
}

@article{kleiner2024extreme,
  title={The Extreme Points of Fusions},
  author={Kleiner, Andreas and Moldovanu, Benny and Strack, Philipp and Whitmeyer, Mark},
  journal={arXiv preprint arXiv:2409.10779},
  year={2024}
}

@article{siegel2009all,
  title={All-pay contests},
  author={Siegel, Ron},
  journal={Econometrica},
  volume={77},
  number={1},
  pages={71--92},
  year={2009},
  publisher={Wiley Online Library}
}

@article{stahl1989oligopolistic,
  title={Oligopolistic pricing with sequential consumer search},
  author={Stahl, Dale O},
  journal={The American Economic Review},
  pages={700--712},
  year={1989},
  publisher={JSTOR}
}

@article{baye1996all,
  title={The all-pay auction with complete information},
  author={Baye, Michael R and Kovenock, Dan and De Vries, Casper G},
  journal={Economic Theory},
  volume={8},
  number={2},
  pages={291--305},
  year={1996},
  publisher={Springer}
}

@book{bochnak2013real,
  title={Real algebraic geometry},
  author={Bochnak, Jacek and Coste, Michel and Roy, Marie-Fran{\c{c}}oise},
  volume={36},
  year={2013},
  publisher={Springer Science \& Business Media}
}

@article{varian1980model,
  title={A model of sales},
  author={Varian, Hal R},
  journal={The American economic review},
  volume={70},
  number={4},
  pages={651--659},
  year={1980},
  publisher={JSTOR}
}

@article{shilony1977mixed,
  title={Mixed pricing in oligopoly},
  author={Shilony, Yuval},
  journal={Journal of Economic Theory},
  volume={14},
  number={2},
  pages={373--388},
  year={1977},
  publisher={Elsevier}
}

@article{adachi2004costly,
  title={Costly participation in voting and equilibrium abstention: a uniqueness result},
  author={Adachi, Takanori},
  journal={Economics Bulletin},
  volume={4},
  number={2},
  pages={1--5},
  year={2004},
  publisher={AccessEcon}
}

@article{augias2025economics,
  title={The economics of convex function intervals},
  author={Augias, Victor and Uhe, Lina},
  journal={arXiv preprint arXiv:2510.20907},
  year={2025}
}

@article{yang2026stochastic,
  title={Stochastic Optimization and Coupling},
  author={Yang, Frank and Yang, Kai Hao},
  journal={arXiv preprint arXiv:2603.11448},
  year={2026}
}

@article{sandomirskiy2026delegation,
  title={Delegation in Strategic Environments and Equilibrium Uniqueness},
  author={Sandomirskiy, Fedor and Wincelberg, Ben},
  journal={arXiv preprint arXiv:2602.21470},
  year={2026}
}

@article{yang2024monotone,
  title={Monotone Function Intervals: Theory and Applications},
  author={Yang, Kai Hao and Zentefis, Alexander K},
  journal={American Economic Review},
  volume={114},
  number={8},
  pages={2239--2270},
  year={2024},
  publisher={American Economic Association 2014 Broadway, Suite 305, Nashville, TN 37203}
}

@article{arieli2023optimal,
  title={Optimal persuasion via bi-pooling},
  author={Arieli, Itai and Babichenko, Yakov and Smorodinsky, Rann and Yamashita, Takuro},
  journal={Theoretical Economics},
  volume={18},
  number={1},
  pages={15--36},
  year={2023},
  publisher={Wiley Online Library}
}

@article{forges2020correlated,
  title={Correlated equilibria and communication in games},
  author={Forges, Fran{\c{c}}oise},
  journal={Complex Social and Behavioral Systems: Game Theory and Agent-Based Models},
  pages={107--118},
  year={2020},
  publisher={Springer}
}

@article{agranov2014communication,
  title={Communication in multilateral bargaining},
  author={Agranov, Marina and Tergiman, Chloe},
  journal={Journal of Public Economics},
  volume={118},
  pages={75--85},
  year={2014},
  publisher={Elsevier}
}

@article{lahr2024extreme,
  title={Extreme Points in Multi-Dimensional Screening},
  author={Lahr, Patrick and Niemeyer, Axel},
  journal={arXiv preprint arXiv:2412.00649},
  year={2025}
}

@article{kamenica2024commitment,
  title={Commitment and Randomization in Communication},
  author={Kamenica, Emir and Lin, Xiao},
  journal={arXiv preprint arXiv:2410.17503},
  year={2024}
}

@article{milgrom1990rationalizability,
  title={Rationalizability, learning, and equilibrium in games with strategic complementarities},
  author={Milgrom, Paul and Roberts, John},
  journal={Econometrica: Journal of the Econometric Society},
  pages={1255--1277},
  year={1990},
  publisher={JSTOR}
}

@article{kajii1997robustness,
  title={The robustness of equilibria to incomplete information},
  author={Kajii, Atsushi and Morris, Stephen},
  journal={Econometrica: Journal of the Econometric Society},
  pages={1283--1309},
  year={1997},
  publisher={JSTOR}
}

@article{bergemann2015limits,
  title={The limits of price discrimination},
  author={Bergemann, Dirk and Brooks, Benjamin and Morris, Stephen},
  journal={American Economic Review},
  volume={105},
  number={3},
  pages={921--957},
  year={2015},
  publisher={American Economic Association 2014 Broadway, Suite 305, Nashville, TN 37203}
}

@article{forges2024subjectivity,
  title={“Subjectivity and correlation in randomized strategies”: Back to the roots},
  author={Forges, Fran{\c{c}}oise and Ray, Indrajit},
  journal={Journal of Mathematical Economics},
  volume={114},
  pages={103044},
  year={2024},
  publisher={Elsevier}
}

@article{dubey1978finiteness,
 ISSN = {0364765X, 15265471},
 URL = {http://www.jstor.org/stable/3690047},
 author = {Pradeep Dubey},
 journal = {Mathematics of Operations Research},
 number = {1},
 pages = {1--8},
 publisher = {INFORMS},
 title = {Inefficiency of Nash Equilibria},
 urldate = {2025-04-04},
 volume = {11},
 year = {1986}
}

@article{cao2025correlated,
  title={Correlated equilibrium of games with concave potential functions},
  author={Cao, Zhigang and Tan, Zhibin and Zhou, Jinchuan},
  journal={Operations Research Letters},
  volume={60},
  pages={107241},
  year={2025},
  publisher={Elsevier}
}

@article{rivera2018efficiency,
  title={Efficiency of correlation in a bottleneck game},
  author={Rivera, Thomas J and Scarsini, Marco and Tomala, Tristan},
  journal={HEC Paris Research Paper No. ECO/SCD-2018-1289},
  year={2018}
}

@inproceedings{nikzad2022constrained,
  title={Constrained majorization: Applications in mechanism design},
  author={Nikzad, Afshin},
  booktitle={Proceedings of the 23rd ACM Conference on Economics and Computation},
  pages={330--331},
  year={2022}
}

@article{manelli2007multidimensional,
  title={Multidimensional mechanism design: Revenue maximization and the multiple-good monopoly},
  author={Manelli, Alejandro M and Vincent, Daniel R},
  journal={Journal of Economic theory},
  volume={137},
  number={1},
  pages={153--185},
  year={2007},
  publisher={Elsevier}
}

@article{felzenbaum1993packing,
  title={Packing lines in a hypercube},
  author={Felzenbaum, Alexander and Holzman, Ron and Kleitman, Daniel J},
  journal={Discrete mathematics},
  volume={117},
  number={1-3},
  pages={107--112},
  year={1993},
  publisher={Elsevier}
}

@article{fudenberg1999conditional,
  title={Conditional universal consistency},
  author={Fudenberg, Drew and Levine, David K},
  journal={Games and Economic Behavior},
  volume={29},
  number={1-2},
  pages={104--130},
  year={1999},
  publisher={Elsevier}
}

@article{foster1997calibrated,
  title={Calibrated learning and correlated equilibrium},
  author={Foster, Dean P and Vohra, Rakesh V},
  journal={Games and Economic Behavior},
  volume={21},
  number={1-2},
  pages={40--55},
  year={1997},
  publisher={Academic Press}
}

@article{hart2000simple,
  title={A simple adaptive procedure leading to correlated equilibrium},
  author={Hart, Sergiu and Mas-Colell, Andreu},
  journal={Econometrica},
  volume={68},
  number={5},
  pages={1127--1150},
  year={2000},
  publisher={Wiley Online Library}
}

@article{canovas1999nash,
  title={Nash equilibria from the correlated equilibria viewpoint},
  author={Canovas, Sabrina Gomez and Hansen, Pierre and Jaumard, Brigitte},
  journal={International Game Theory Review},
  volume={1},
  number={01},
  pages={33--44},
  year={1999},
  publisher={World Scientific}
}

@article{ui2008correlated,
  title={Correlated equilibrium and concave games},
  author={Ui, Takashi},
  journal={International Journal of Game Theory},
  volume={37},
  pages={1--13},
  year={2008},
  publisher={Springer}
}

@article{georgalos2020nash,
  title={Nash versus coarse correlation},
  author={Georgalos, Konstantinos and Ray, Indrajit and SenGupta, Sonali},
  journal={Experimental Economics},
  volume={23},
  pages={1178--1204},
  year={2020},
  publisher={Springer}
}

@article{friedman2022empirical,
  title={On the empirical relevance of correlated equilibrium},
  author={Friedman, Daniel and Rabanal, Jean Paul and Rud, Olga A and Zhao, Shuchen},
  journal={Journal of Economic Theory},
  volume={205},
  pages={105531},
  year={2022},
  publisher={Elsevier}
}

@article{dokka2023equilibrium,
  title={Equilibrium design in an n-player quadratic game},
  author={Dokka, Trivikram and Moulin, Herv{\'e} and Ray, Indrajit and SenGupta, Sonali},
  journal={Review of economic design},
  volume={27},
  number={2},
  pages={419--438},
  year={2023},
  publisher={Springer}
}

@misc{cripps1995extreme,
  title={Extreme correlated and Nash equilibria in two-person games},
  author={Cripps, M},
  year={1995},
  publisher={working paper, University of Warwick}
}

@article{palfrey1983strategic,
  title={A Strategic Calculus of Voting},
  author={Palfrey, Thomas R and Rosenthal, Howard},
  journal={Public Choice},
  volume={41},
  number={1},
  pages={7--53},
  year={1983},
  publisher={Springer}
}

@article{baye1993rigging,
  title={Rigging the Lobbying Process: An Application of the All-Pay-Auction},
  author={Baye, Michael and Kovenock, Dan and de Vries, Casper},
  journal={The American Economic Review},
  pages={289--294},
  year={1993}
}
	
\pagebreak

\appendix

\section{Extremality in the Space of Action Distributions}
\label{app_th1_proof}

In this appendix we formulate and prove results pertaining to extremality of Nash equilibria in the space of action distributions. Our ultimate goal is to prove the following theorem, which generalizes Theorem~\ref{th_almost_pure}. 

\begin{theorem}\label{th_almost_pure_appendix}
Consider a regular Nash equilibrium supported on 
$S = S_1 \times \ldots \times S_n$ and let 
$k = \bigl|\{i \colon |S_i|\geq 2\}\bigr|$ be the number of agents who randomize their actions. This equilibrium is an extreme point of the correlated equilibrium polytope if and only if $k\leq 2$. For $k\geq 3$, the equilibrium lies in the interior of a face of dimension greater than or equal to 
\begin{equation}\label{eq_dimension}
2^k-2k-1 \ + \ \begin{cases} \prod_{i\colon |S_i|\geq 2} \  (|S_i|-1)-1 , & k\geq 4\\
     0, & k=3
\end{cases}.
\end{equation}
\end{theorem}

\vspace{3mm}

The theorem implies that, when $k \geq 3$ agents are mixing non-trivially, the equilibrium lies in the interior of a face whose dimension grows at least exponentially in the number $k$ of mixers since
\begin{equation*}
    2^{k} - 2k - 1 \ \geq \  2^{k-3} \qquad \text{for any}\quad k\geq 3.
\end{equation*}
Furthermore, the bound given by \eqref{eq_dimension} indicates that this growth rate further increases with the number of actions over which agents randomize. For example, if each of the $k\geq 4$ agents randomizes over $m$ actions, then the dimension of the face is at least $(m-1)^k$.

The dimension of the face containing a Nash equilibrium represents the number of linearly independent perturbations that preserve the incentive constraints; essentially, the number of directions in which potential improvements can occur. We now formalize the notion of such perturbations.

\subsection{The Tangent Space} 
Consider a correlated equilibrium $\mu\in \Delta(A)$. A vector $\tau \in \R^A$ represents a feasible perturbation direction at $\mu$ if, for a sufficiently small $\varepsilon > 0$, the distributions $\mu + \varepsilon \tau$ and $\mu - \varepsilon \tau$ are both correlated equilibria. The set of such vectors $\tau$ forms a linear space, which is the tangent space to the face of the polytope $\ce(\G)$ containing $\mu$.
Let $T_\mu$ be the set of all such directions $\tau$:
\begin{equation}\label{eq_tangent_space_mu}
T_\mu = \big\{\tau\in \R^A\colon \mu\pm\varepsilon \tau \ \ \text{ is a correlated equilibrium for small enough  \ \  $\varepsilon>0$}\big\}.
\end{equation}
This is the  \textbf{tangent space} to the set of correlated equilibria at $\mu$. The dimension of a face of the correlated equilibrium polytope to which $\mu$ belongs is the dimension of the tangent space $T_\mu$. In particular, $\mu$ is an extreme correlated equilibrium if and only if $T_\mu=\{0\}$; that is, the dimension of the tangent space is zero. We are interested in extremality of Nash equilibria and thus aim to understand whether the dimension $\dim T_\nu$ is positive or zero.

A Nash equilibrium $\nu$ supported on $S=S_1\times \ldots \times S_n\subseteq A$  is \textbf{quasi-strict} if no action outside $S_i$ is a best response to $\nu_{-i}$. That is, 
\begin{equation}\label{eq_quasistrict}
    \sum_{a_{-i}} \nu(a_i,a_{-i}) (u_i(a_i,a_{-i}) - u_i(a_i',a_{-i})) > 0, \quad \forall i,  a_i\in S_i,  a_i' \in A_i\setminus S_i.
\end{equation}
The following lemma describes the tangent space~$T_\nu$ to a quasi-strict Nash equilibrium $\nu$ and offers a lower bound on its dimension. 
\begin{lemma}\label{lm_cone}
   The tangent space $T_\nu$ to a quasi-strict Nash equilibrium $\nu$ supported on $S=S_1\times \ldots \times S_n\subseteq A$ consists of all vectors $\tau\in \R^A$ solving the following system:
   \begin{equation}\label{eq_lin_independent_delta}
    \left\{
    \begin{array}{l}
          \sum_{a_{-i}} \tau(a_i,a_{-i}) (u_i(a_i,a_{-i}) - u_i(a_i',a_{-i})) = 0, \quad \forall i, \forall a_i \neq a_i' \in S_i \\
          \tau(a)=0, \quad a\not\in S\\ 
          \sum_{a \in S} \tau(a) = 0
    \end{array}.
    \right.
\end{equation}   
The dimension of the tangent space admits the following lower bound:
\begin{equation}\label{eq_cone_bound}
    \dim T_\nu\geq \prod_i |S_i|-\left(1+\sum_i |S_i|(|S_i|-1)\right).
\end{equation}
\end{lemma}

\vspace{2mm}

The system~\eqref{eq_lin_independent_delta} implies that only incentive constraints on the support are relevant; these must bind as equalities, and $\tau$ must vanish outside the support. The bound on the dimension of the tangent space in~\eqref{eq_cone_bound} follows from the fact that the dimension of the space of solutions is at least the dimension of the space minus the number of constraints imposed. The actual dimension may in fact be larger, since we do not yet exclude the possibility that some of these constraints are linearly dependent.
\begin{proof}

Perturbations $\nu + \varepsilon \tau$ and $\nu - \varepsilon \tau$ lead to correlated equilibria for small $\varepsilon$ if and only if a direction $\tau\in \R^A$ satisfies several conditions. First, since $\nu$ is a Nash equilibrium, the incentive constraints are binding on its support:
\begin{equation*} 
\sum_{a_{-i}} \nu(a_i, a_{-i}) u_i(a_i, a_{-i}) = \sum_{a_{-i}} \nu(a_i, a_{-i}) u_i(a_i', a_{-i})\quad \forall i \in N, \forall a_i\ne a_i' \in S_i.
\end{equation*}
Hence, for $\nu \pm \varepsilon \tau$ to remain a CE, the same constraints must continue to hold as equalities resulting in the first equation in~\eqref{eq_lin_independent_delta}.
By quasi-strictness of $\nu$, the incentive constraints outside of the support hold as strict inequalities and thus continue to hold for $\nu \pm \varepsilon \tau$, provided that $\varepsilon$ is small enough.

Second, for $\nu \pm \varepsilon \tau$ to have no negative weights, the support of $\tau$ must be contained within the support of $\nu$; that is, $\tau(a) = 0$ for all $a \notin S$.

Finally, to guarantee that the weights in $\nu \pm \varepsilon \tau$ sum up to one, $\tau$ must satisfy the normalization condition  $\sum_{a \in S} \tau(a) = 0$. We conclude that $\tau$ belongs to $T_\nu$ if and only if it solves the system~\eqref{eq_lin_independent_delta}.

Since $\tau(a)=0$ for $a\in A\setminus S$, we can identify $\tau$ with a vector in $\R^S$ by excluding the components outside of $S$. The total number of constraints imposed on this vector is $\sum_i |S_i|(|S_i|-1)$, with the additional constraint of zero total weight. Thus, the dimension of the space of solutions is at least $|S|$ minus  $\sum_i |S_i|(|S_i|-1)+1$, which yields the bound in~\eqref{eq_cone_bound}. 
\end{proof}

The lower bound~\eqref{eq_cone_bound} does not, by itself, guarantee that $\dim T_\nu>0$ when three or more agents mix. For instance, if $n=3$ and $|S_2|=|S_3|=2$ while $|S_1|=4$, the right hand side of~\eqref{eq_cone_bound} equals $-1$, so the bound is vacuous. However, as we next discuss, such a configuration of support sizes cannot arise at a regular Nash equilibrium. 

\subsection{Regular Equilibria and the Polygon Inequality}~\label{app_polygon}
Consider a Nash equilibrium $\nu$ supported on $S=S_1\times\cdots\times S_n$. Informally, $\nu$ is called regular if all incentive constraints for actions outside of $S$ are slack and the implicit function theorem can be applied on $S$ to conclude that the equilibrium changes smoothly under small perturbations of utility functions. The slackness requirement is captured by quasi-strictness of $\nu$ (see equation~\eqref{eq_quasistrict}). We therefore focus on formalizing the robustness to small perturbations.

To this end, for any agent $i$, we introduce local coordinates on the simplex $\Delta(S_i)$ around $\nu_i$ as follows. Fix a reference action profile $a^*\in S$. Define $\tilde{\nu}_i(\delta_i)$ for $\delta_i\in\R^{S_i\setminus\{a_i^*\}}$ by setting $\tilde{\nu}_i(\delta_i)(a_i)=\nu_i(a_i)+\delta_i(a_i)$ for all $a_i\in S_i\setminus\{a_i^*\}$ and choosing $\tilde{\nu}_i(\delta_i)(a_i^*)$ so that probabilities sum to one. Hence, $\tilde{\nu}_i(\delta_i)$ satisfies $\tilde{\nu}_i(0)=\nu_i$ and belongs to $\Delta(S_i)$ for~$\delta_i$ in a neighborhood of $0$. Viewing $\delta=(\delta_1,\ldots,\delta_n)$ as a vector in $\R^d$ with $d=\sum_{i=1}^n (|S_i|-1)$, we define the \textbf{indifference map} $F_i\colon \R^d\to \R^{S_i\setminus\{a_i^*\}}$ of agent $i$ by
\begin{equation*}
F_i(\delta)
\;=\;
\Big(u_i\big(a_i,\tilde{\nu}_{-i}(\delta_{-i})\big)-u_i\big(a_i^*,\tilde{\nu}_{-i}(\delta_{-i})\big)\Big)_{a_i\in S_i\setminus\{a_i^*\}}.
\end{equation*}
The condition that agent $i$ is indifferent among actions in $S_i$ at $\nu$ is equivalent to $F_i(0)=0$. Let $F=(F_1,\ldots,F_n)$ denote the stacked indifference maps $F\colon \R^d\to\R^d$. The equilibrium indifference conditions at $\nu$ are then equivalent to the single system $F(0)=0$. 

We can now state the formal definition. A Nash equilibrium $\nu$ is \textbf{regular} if it is quasi-strict (equation~\eqref{eq_quasistrict} holds) and the Jacobian matrix $DF(0)$ of $F$ at $\delta=0$ satisfies
\begin{equation}\label{eq_regualr}
    \det DF(0)\neq 0.
\end{equation}

Quasi-strictness ensures that actions outside of $S_i$ remain strictly suboptimal under small payoff perturbations. On the support, the non-singularity of $DF(0)$ implies, via the implicit function theorem, that $\nu$ is locally unique and varies smoothly with sufficiently small perturbations of $u=(u_1,\ldots, u_n)$. Furthermore, the regularity conditions---given by strict inequalities---are maintained by small perturbations of payoffs. This robustness of regular equilibria, and the corresponding fragility of irregular equilibria, underlies the result of \cite{harsanyi1973oddness} that, in generic games, all Nash equilibria are regular and finite in number.

\begin{lemma}[Polygon Inequality]\label{lm_polygon}
If $\nu$ is a regular Nash equilibrium supported on 
$S = S_1 \times \ldots \times S_n$, then 
\begin{equation}\label{eq_polygon_app}
    |S_i|-1\leq \sum_{j\ne i} (|S_j|-1)\qquad \text{for all $i=1,\ldots, n$},
\end{equation}
that is, one can construct a polygon with $n$ sides such that each side length is $|S_i|-1$. 
\end{lemma}
This result follows from the analysis of \cite{mckelvey1997maximal}, who bound the maximal number of Nash equilibria with a given support in a generic game. They show that this number vanishes for equilibria that violate \eqref{eq_polygon_app}, and hence the inequality holds generically. In their work, this conclusion arises as a byproduct of a technically demanding argument that relies on advanced tools from algebraic geometry. For completeness, we provide a brief and direct proof. 
\begin{proof}
Let $F=(F_1,\ldots,F_n)$ be the indifference map defined above, written in local coordinates
$\delta\in\R^d$ with $d=\sum_i (|S_i|-1)$. Its Jacobian at $\delta=0$ is the $d\times d$ matrix $DF(0)$, which we
 represent in row blocks $DF_i(0)$ corresponding to agent $i$'s indifference equations:
\begin{equation*}
DF(0)
  \;=\;
  \begin{bmatrix}
    DF_1(0) \\ DF_2(0) \\ \vdots \\ DF_n(0)
  \end{bmatrix},
  \qquad
  DF_i(0)\in\R^{(|S_i|-1)\times d}.
\end{equation*}
By regularity (see ~\eqref{eq_regualr}), $\det DF(0)\neq 0$. Hence $DF(0)$ has full rank and its rows are linearly
independent. In particular, each block $DF_i(0)$ has full row rank of $|S_i|-1$.

Row rank equals column rank, and the latter is bounded above by the number of nonzero columns. Since $F_i$ depends only on agent $i$'s opponents’ probabilities, it is invariant to $\delta_i$. Accordingly, the columns of $DF_i(0)$ corresponding to $\delta_i$ are zero. Consequently, $DF_i(0)$ has at most $\sum_{j\ne i}(|S_j|-1)$ nonzero columns, which implies that $|S_i|-1 \;\leq\; \sum_{j\ne i}(|S_j|-1)$ for every $i$.
\end{proof}

We next show how the lower bound~\eqref{eq_cone_bound} tightens for three or more mixing agents once the polygon inequality rules out highly asymmetric equilibria.

\subsection{Auxiliary Inequalities for Three or More Mixing Agents}

Consider a regular Nash equilibrium $\nu$ supported on $S=S_1\times \ldots \times S_n$. By Lemma~\ref{lm_cone}, the dimension of the tangent space at $\nu$ satisfies the following bound
\begin{equation*} 
    \dim T_\nu\geq \prod_i |S_i|-\left(1+\sum_i |S_i|(|S_i|-1)\right).
\end{equation*} 
Here, we explore the values that the right-hand-side of this inequality can take as a function of the number of randomizing agents (agents $i$ with $|S_i|\geq 2$), by combining this inequality with the polygon inequality from Lemma~\ref{lm_polygon}. Since the analysis does not rely on our particular game-theoretic interpretation, we formulate the result in terms of an arbitrary integer sequence $m_1,\ldots,m_n$, which we later set to satisfy $m_i=|S_i|$.

\begin{proposition}\label{lm_inequality}
 Let  $m_1,\ldots, m_n$ be a sequence of integers $m_i\geq 1$ satisfying the polygon inequality
 $m_i-1\ \leq \ \sum_{j\ne i} (m_j-1)$. 
   Denote the number of indices with $m_i\geq 2$ by $k=\bigl|\{i \colon m_i\geq 2\}\bigr|$ and suppose $k\geq 3$. Then, 
\begin{equation}\label{eq_integer_sequence_lower_bound}
 \prod_{i=1}^n m_i-\left(1+\sum_{i=1}^n m_i(m_i-1)\right)\ \ \  \geq \ \ \  2^k-2k-1+\begin{cases} \left(\prod_{i\colon m_i\geq 2} \  (m_i-1)\right)-1  & k\geq 4\\
     0 & k=3
\end{cases}.
\end{equation}
\end{proposition}

\vspace{2mm}

For the proof of Theorem~\ref{th_almost_pure_appendix}, the weaker lower bound without the additional term for $k\geq 4$ suffices. The extra term will be needed  in the discussion of payoff-space extremality appearing in Appendix~\ref{app_payoff_space}.

\begin{proof}
After reindexing, assume $m_i \ge 2$ for $ i=1,\ldots,k$ and $m_i=1$ for $i>k$. Set $x_i=m_i-1$ for all $i$. Then $x_i \geq 1$ and $x_i \le \sum_{j \ne i} x_j$ for all $i$.

If $k=3$, a direct expansion shows that the difference between the two sides of \eqref{eq_integer_sequence_lower_bound} is equal to
\begin{equation*}
    \sum_{i=1}^3 (x_i-1)\Bigl(\sum_{j\ne i}x_j-x_i\Bigr)\;+\;(x_1-1)(x_2-1)(x_3-1).
\end{equation*}
The polygon inequality, together with the fact that $x_i \ge 1$ for $i \leq 3$, imply that each term in this expression is nonnegative. The claim then follows.

Now, assume that $k \ge 4$. Define
\begin{equation*}
    F(x)=\prod_{i=1}^k(1+x_i)-\prod_{i=1}^k x_i-\sum_{i=1}^k(x_i+1)x_i-2^k+2k+1
\end{equation*}
on $\mathcal P=\big\{x \in [1,\infty)^k : x_i \le \sum_{j \ne i} x_j \ \text{ for all } i\big\}$. Then \eqref{eq_integer_sequence_lower_bound} is equivalent to $F(x)\ge 0$ for integer-valued~$x \in \mathcal P$.

Let $\mathbf 1=(1,\dots,1)$. Since $\mathbf 1 \in \mathcal P$ and $F(\mathbf 1)=0$, it suffices to show that $F$ is non-decreasing along each segment $x(t)=\mathbf 1+t(y-\mathbf 1)$ with $y \in \mathcal P$ and $0 \le t \le 1$. Along such a segment, $\frac{d}{dt}F(x(t))=\sum_{i=1}^k (y_i-1)\,\partial_iF(x(t))$,
so it suffices to prove that $\partial_iF \ge 0$ on $\mathcal P$.
From the definition of~$F$,
\begin{equation*}
\partial_iF(x)=\prod_{j \ne i}(1+x_j)-\prod_{j \ne i}x_j-(2x_i+1).
\end{equation*}
After expanding $\prod_{j \ne i}(1+x_j)$, the constant term and the top-degree term cancel. Thus,
\begin{equation*}
\partial_iF(x)\ge \sum_{j \ne i}x_j+\sum_{\substack{j<\ell\\ j,\ell \ne i}}x_jx_\ell-2x_i
\end{equation*}
since all omitted terms are nonnegative. Also,
\begin{equation*}
    2\sum_{\substack{j<\ell\\ j,\ell \ne i}}x_jx_\ell
=\sum_{\substack{j \ne \ell\\ j,\ell \ne i}}x_jx_\ell
\ge (k-2)\sum_{j \ne i}x_j.
\end{equation*}
Indeed, each fixed $x_j$ occurs in exactly $k-2$ ordered products $x_jx_\ell$. Furthermore, $x_jx_\ell \ge x_j$ because $x_\ell \ge 1$. Since $k \ge 4$, this yields $\sum_{j<\ell,\ j,\ell \ne i} x_jx_\ell \ge \sum_{j \ne i}x_j$. Hence,
\begin{equation*}
    \partial_iF(x) \ge 2\sum_{j \ne i}x_j-2x_i \ge 0
\end{equation*}
by the polygon inequality.
Therefore $F \ge 0$ on  $\mathcal P$. Substituting back $m_i=x_i+1$ proves~\eqref{eq_integer_sequence_lower_bound}.
\end{proof}

\subsection{Regular Equilibria with at Most Two Mixing Agents}

Analyzing regular Nash equilibria with at most two mixing agents reduces to the two-agent case, since agents playing pure actions can be replaced with dummies. Nash equilibria in two-agent games are very well understood; for completeness, we provide short self-contained proofs of the specific facts we need, using the tools developed above.

Consider a regular Nash equilibrium $\nu$ supported on $S=S_1\times\cdots\times S_n$ and suppose that the number of mixing agents $k = \bigl|\{i \colon |S_i|\geq 2\}\bigr|$ is at most $2$. The polygon inequality determines the support structure in this case.

\begin{corollary}\label{cor_few_mixers}
If $\nu$ is a regular Nash equilibrium with  $k\leq 2$ mixing agents, then either $\nu$ is pure, or exactly two agents mix and their supports have the same cardinality.
\end{corollary}
Indeed, $k=0$ corresponds to a pure $\nu$. The case of one randomizing agent is ruled out since $|S_i|\geq 2$ for only one agent $i$ contradicts the polygon inequality
$|S_i|-1\leq \sum_{j\ne i}(|S_j|-1)$: the left-hand side is positive while the right-hand side equals zero.  Finally, if $k=2$ and $i\ne j$ are the mixing agents, then the polygon inequality yields
$|S_i|-1\leq |S_j|-1$ and $|S_j|-1\leq |S_i|-1$, hence $|S_i|=|S_j|$.

With $k=2$ mixing agents, the indifference map from the regularity condition~\eqref{eq_regualr} becomes especially simple: agent $i$'s indifference equations depend only on agent $j$'s mixing weights, and vice versa. As a result, the Jacobian from the definition of regularity does not depend on the equilibrium mixing weights and its non-singularity simplifies to linear independence of payoff-difference vectors introduced in the following lemma.

\begin{lemma}\label{lm_two_mixing_regular}
Let $\nu$ be a quasi-strict Nash equilibrium in which agents $i$ and $j$ mix on supports $S_i$ and $S_j$ with
$|S_i|=|S_j|=m\ge 2$, while every $\ell\notin\{i,j\}$ plays $a_\ell^*$ deterministically. Fix reference actions
$a_i^*\in S_i$ and $a_j^*\in S_j$. For each $a_i\in S_i\setminus\{a_i^*\}$ and $a_j\in S_j\setminus\{a_j^*\}$, define
payoff-difference vectors $X_{a_i}\in\R^m$ and $Y_{a_j}\in\R^m$ by
\begin{equation}\label{eq_XY_blocks}
\begin{aligned}
X_{a_i}(a_j)
&=u_i(a_i,a_j,a_{-ij}^*)-u_i(a_i^*,a_j,a_{-ij}^*),\\
Y_{a_j}(a_i)
&=u_j(a_i,a_j,a_{-ij}^*)-u_j(a_i,a_j^*,a_{-ij}^*).
\end{aligned}
\end{equation}
Then $\nu$ is regular if and only if the families of vectors $X_{a_i}$, $a_i\in S_i\setminus\{a_i^*\}$, and
$Y_{a_j}$, $a_j\in S_j\setminus\{a_j^*\}$, are each linearly independent in $\R^m$.
\end{lemma}

\begin{proof}
Differentiating the indifference map $F\colon \R^{2(m-1)}\to \R^{2(m-1)}$ at $\delta=0$ yields the block Jacobian
\begin{equation*}
DF(0)=
\begin{bmatrix}
0 & \widehat X\\
\widehat Y & 0
\end{bmatrix},
\end{equation*}
where $\widehat X$ and $\widehat Y$ are $(m-1)\times (m-1)$ matrices with entries
$\widehat X(a_i,a_j)=X_{a_i}(a_j)-X_{a_i}(a_j^*)$ and
$\widehat Y(a_j,a_i)=Y_{a_j}(a_i)-Y_{a_j}(a_i^*)$.
Thus, $\nu$ is regular if and only if both $\widehat X$ and $\widehat Y$ are non-singular.

We show that $\widehat X$ is non-singular if and only if the vectors $X_{a_i}$ are linearly independent. If the vectors $X_{a_i}$ are linearly dependent, then some nontrivial combination $\sum_{a_i}c(a_i)X_{a_i}$ vanishes. Looking at $X_{a_i}$, subtracting its $a_j^*$-coordinate from each $a_j$-coordinate shows that the same coefficients give a nontrivial linear relation among the rows of $\widehat X$, so $\widehat X$ is singular. Symmetric arguments follow for the non-singularity of $\widehat Y$.

Conversely, if $\widehat X$ is singular, there exist coefficients $c(a_i)$, not all zero, such that, for all $a_j\neq a_j^*$, we have $\sum_{a_i}c(a_i)(X_{a_i}(a_j)-X_{a_i}(a_j^*))=0$. Hence, the vector
$z=\sum_{a_i}c(a_i)X_{a_i}\in\R^m$ is constant across coordinates. On the other hand, indifference of agent $i$ at the Nash equilibrium implies  that $\sum_{a_j}\nu_j(a_j)X_{a_i}(a_j)=0$ for every $a_i\neq a_i^*$. After multiplying by $c(a_i)$ and summing up, we get $\sum_{a_j}\nu_j(a_j)z(a_j)=0$. Since $z$ is constant and $\nu_j$ is a probability distribution, this constant must be zero. Thus $z=0$, so the vectors $X_{a_i}$ are linearly dependent.
\end{proof}

We can now prove that regular Nash equilibria with two mixing agents are extreme correlated equilibria.

\begin{lemma}\label{lm_extreme_two_mixing}
If $\nu$ is a regular Nash equilibrium in which exactly two agents mix, then $\nu$ is an extreme point of the correlated
equilibrium polytope.
\end{lemma}
Versions of this result appear in \cite*{cripps1995extreme}, \cite*{evangelista1996note}, and \cite*{canovas1999nash}.

\begin{proof}
Extremality is equivalent to showing that the tangent space $T_\nu$ from~\eqref{eq_tangent_space_mu} is trivial. Let $\tau\in T_\nu$. By Lemma~\ref{lm_cone}, $\tau$ vanishes outside the support $S$ and by Corollary~\ref{cor_few_mixers}, when two agents mix, they have identically-sized supports. Therefore, we can identify $\tau$ with the $m\times m$ matrix $T$ on $S_i\times S_j$ defined by $T(a_i,a_j)=\tau(a_i,a_j,a_{-ij}^*)$.

Fix $a_i\in S_i$. The tangent-space conditions in Lemma~\ref{lm_cone} imply that the row $T(a_i,\,\cdot\,)$ is orthogonal to all payoff-difference vectors relevant for agent $i$'s deviations on the support. These vectors span the same subspace as $\{X_{a_i'}:a_i'\neq a_i^*\}$. By Lemma~\ref{lm_two_mixing_regular}, this subspace has dimension $m-1$. Its orthogonal complement is therefore one-dimensional. Since $\nu_j$ is orthogonal to every $X_{a_i'}$ by agent $i$'s indifference at $\nu$, each row of $T$ must be a scalar multiple of $\nu_j$. Thus $T(a_i,a_j)=\alpha(a_i)\nu_j(a_j)$ for some coefficients $\alpha(a_i)$.

Applying the same argument to agent $j$, each column of $T$ must be a scalar multiple of $\nu_i$. Hence, there exists $c\in\R$ such that $T(a_i,a_j)=c\,\nu_i(a_i)\nu_j(a_j)$ for all $(a_i,a_j)$.

Finally, Lemma~\ref{lm_cone} also imposes the normalization condition $\sum_{a_i,a_j}T(a_i,a_j)=0$. Since
$\sum_{a_i,a_j}\nu_i(a_i)\nu_j(a_j)=1$, this forces $c=0$. Therefore, $T=0$, so that $\tau=0$. We conclude that $T_\nu=\{0\}$, and hence $\nu$ is extreme.
\end{proof}

\subsection{Proof of Theorem~\ref{th_almost_pure_appendix}}

Theorem~\ref{th_almost_pure_appendix} now follows easily from a combination of the above preparatory results.

\begin{proof}[Proof of Theorem~\ref{th_almost_pure_appendix}]
Consider a regular Nash equilibrium $\nu$ supported on $S = S_1 \times \cdots \times S_n$ with $k = \bigl|\{i : |S_i| \geq 2\}\bigr|$ mixing agents.
The dimension of the face of the correlated equilibrium polytope $\ce(\G)$ containing
$\nu$ equals $\dim T_\nu$, where $T_\nu$ is the tangent space defined
in~\eqref{eq_tangent_space_mu}. In particular, $\nu$ is an extreme point if and only
if $T_\nu = \{0\}$.

Assume first that $k\geq 3$. By Lemma~\ref{lm_cone},
\begin{equation*}
  \dim T_\nu
  \;\geq\;
  \prod_i |S_i| - \Bigl(1 + \sum_i |S_i|(|S_i|-1)\Bigr).
\end{equation*}
By regularity of $\nu$ and Lemma~\ref{lm_polygon}, the support sizes satisfy the
polygon inequality $|S_i| - 1 \leq \sum_{j \neq i}(|S_j| - 1)$.
Therefore, applying Proposition~\ref{lm_inequality} to the preceding bound yields
\begin{equation}\label{eq_dimension_th_proof}
\dim T_\nu\;\ge\;
2^k-2k-1+\begin{cases}
\prod_{i:\,|S_i|\ge 2} (|S_i|-1)-1 & k\ge 4\\
0 & k=3
\end{cases}.
\end{equation}
The right-hand side is strictly positive since $2^k-2k-1>0$ for $k\geq 3$. Hence, $\dim T_\nu > 0$, so $\nu$ is not extreme and lies in the interior of a face
of dimension at least as high as the right-hand side of~\eqref{eq_dimension_th_proof}.

Now suppose $k \leq 2$. Corollary~\ref{cor_few_mixers} implies that $\nu$ is either
pure or has exactly two mixing agents.
If $\nu$ is pure, it is extreme because point masses are extreme among all probability
distributions, hence also within $\ce(\G)$.
If exactly two agents mix, Lemma~\ref{lm_extreme_two_mixing} provides extremality directly.

Combining both cases, a regular Nash equilibrium is an extreme point of $\ce(\G)$ if and only if $k \leq 2$.
\end{proof}

\section{Proofs Pertaining to Extremality in Payoff Space}\label{app_payoff_space}

We demonstrate that, in generic $n$-agent games, Nash equilibria with $k\geq 3$ mixing agents can be welfare-improved by correlation, while those with $k\geq 8 + \log_2 n$ admit a Pareto improvement. The welfare improvement construction for games with $n\geq 4$ agents combines a dimension-counting argument with the invariance of the correlated equilibrium set under strategic equivalence; for $n = 3$, we rely on a more abstract algebraic genericity argument. The Pareto improvement results draw on a similar abstract argument and require constructing a canonical game with a Pareto-improvable Nash equilibrium of prescribed support, which we accomplish via a new cylinder-packing combinatorial technique.

\subsection{Welfare Improvements: Proof of Proposition~\ref{prop:payoff_extreme}}\label{app_welfare}

We start with a high-level overview of the proof of Proposition~\ref{prop:payoff_extreme}, formalizing the steps through a sequence of technical lemmas below.

Fix a number of agents $n\geq 3$ and a game $\Gamma$. 
Let $\nu$ be a regular Nash equilibrium supported on $S = S_1 \times \ldots \times S_n\subseteq A$ and denote the  number of agents who randomize their actions by $k = \bigl|\{i \colon |S_i|\geq 2\}\bigr|$.

Recall that $T_\nu$ defined in~\eqref{eq_tangent_space_mu} denotes the tangent space to the set of correlated equilibria at $\nu$, consisting of all perturbations $\tau\in \R^A$ such that $\nu\pm\varepsilon \tau$ are correlated equilibria for small enough $\varepsilon>0$. For quasi-strict Nash equilibria, Lemma~\ref{lm_cone} characterizes $T_\nu$ and shows that it depends on $\nu$ only through its support $S$. Moreover, the linear system in that lemma is well-defined regardless of whether a Nash equilibrium supported on $S$ exists. We therefore write $T_S$ for the solution set of this system, so that $T_S = T_\nu$ whenever $\nu$ is a quasi-strict Nash equilibrium supported on $S$.

Given $S\subseteq A$, let $\mathfrak{U}_{S\subseteq A}\subset \big(\R^A\big)^n$ be the set of utilities $u=(u_i)_i$ admitting some $\tau \in T_S$ with $\sum_i u_i(\tau)\ne 0$, where $u_i(\tau)=\sum_a u_i(a)\tau(a)$. For $u\in \mathfrak{U}_{S\subseteq A}$, any quasi-strict Nash equilibrium $\nu$ supported on $S$ admits perturbations $\mu_\pm=\nu\pm\varepsilon \tau$ that are both correlated equilibria for small enough $\varepsilon$, one with strictly higher and one with strictly lower welfare than $\nu$. Thus, $\nu$ is not welfare-optimal among correlated equilibria, and since $u(\nu)$ is the average of the distinct points $u(\mu_+)$ and $u(\mu_-)$, it is not payoff-extreme either.

We will show that whenever $S$ with $k=\bigl|\{i\colon |S_i|\geq 2\}\bigr|\geq 3$ supports a regular Nash equilibrium for some $u$, the set $\mathfrak{U}_{S\subseteq A}$ contains an open set of full measure; that is, such games are generic. Intersecting over all such $S$ and restricting to games where all Nash equilibria are regular, we conclude that for a generic game, every Nash equilibrium with $3$ or more mixers is neither payoff-extreme nor welfare-optimal.

The key difficulty in showing that $\mathfrak{U}_{S\subseteq A}$ is generic is demonstrating that it is dense. Once density is established, genericity follows by a standard algebraic argument.

\begin{lemma}\label{lm_density_genericity}
Density of $\mathfrak{U}_{S\subseteq A}$ in $(\R^A)^n$ implies that it contains an open subset of full measure.
\end{lemma}

\begin{proof}
$\mathfrak{U}_{S\subseteq A}$ is a semialgebraic set since the condition ``$\exists\tau\in T_S$ with $\sum_i u_i(\tau)\neq 0$'' can be expressed by finitely many polynomial equalities and inequalities in~$u$.\footnote{A semialgebraic set in $\R^d$ is any subset that can be described by finitely many polynomial equalities and inequalities. By the ``quantifier-elimination'' theorem of Tarski and Seidenberg (Theorem~1.4.2 in \cite{bochnak2013real}), sets defined by first-order logical formulas with polynomial equalities and inequalities and quantifiers (for example, $\exists \tau$ or $\forall \tau$) are also semialgebraic; that is, such sets can be expressed without using these quantifiers.}

It therefore suffices to show that every dense semialgebraic subset of $\R^d$ contains an open set whose complement has measure zero. Consider any such set, and let $C$ denote its complement. Then $C$ is a semialgebraic set with an empty interior. Every semialgebraic set can be written as a finite union of nonempty sets, each defined by finitely many polynomial equalities and \emph{strict} inequalities.\footnote{Represent each weak inequality as either an equality or strict inequality and discard empty pieces;
see \cite[Chapter~2.3]{bochnak2013real} for details.} Any piece defined solely by strict inequalities is open in $\R^d$ and therefore cannot appear in a decomposition of $C$. Hence, every nonempty piece of $C$ is contained in the zero set of a polynomial that is not identically zero. It follows that $C$ is contained in a finite union of algebraic surfaces of positive codimension, which has measure zero. Since $\mathfrak{U}_{S\subseteq A}$ is a dense semialgebraic set, it contains an open set with a zero-measure complement.
\end{proof}

The density of $\mathfrak{U}_{S\subseteq A}$ is established by exploiting the invariance of correlated equilibria under strategic equivalence. In particular, consider modifying agent $i$'s payoff by adding a function that depends only on the opponents' actions, $u'_i(a)=u_i(a)+\delta_i(a_{-i})$, yielding a strategically equivalent game $\Gamma'$ in the sense of \cite{moulin1978strategically}. This transformation leaves agents' incentives unchanged, since all payoff differences of the form $u_i(a_i,a_{-i})-u_i(a_i',a_{-i})$ remain the same. Consequently, the sets $\nash(\Gamma)$, $\ce(\Gamma)$, and $T_S$ are preserved. At the same time, the transformation alters the direction in which these equilibrium sets are projected into the payoff space.

We show that even when the welfare of a regular Nash equilibrium $\nu$ with three mixers in a game $\Gamma$ cannot be improved, it becomes improvable under a small strategically equivalent perturbation of the form $u'_i(a)=u_i(a)+\delta_i(a_{-i})$. Indeed, consider $\tau\in T_S$ and a perturbation $\delta$. Let $\delta_i(\tau)$ be the weighted average by $\tau$ of the perturbation $\delta_i$. That is, $\delta_i(\tau)=\sum_a \delta_i(a_{-i}) \tau(a)$.
Then, for any $i$,
\begin{equation}\label{eq_SE_pairing}
 \delta_i(\tau)= \sum_{a_{-i}}\delta_i(a_{-i})\Big(\sum_{a_i}\tau(a_i,a_{-i})\Big).
\end{equation}
Hence,  we can find $\delta$ such that $\sum_i \delta_i(\tau)\ne 0$ if and only if there exists some $i$ for which the $(n-1)$-dimensional
marginal $a_{-i}\to \sum_{a_i}\tau(a_i,a_{-i})$ is not identically zero. Consider, then, the following subspace: 
\begin{equation*}
    T_{S,0}=\Big\{\tau\in\R^A:\ \ \sum_{a_i}\tau(a_i,a_{-i})=0\ \text{for all $i$ and all $a_{-i}\in A_{-i}$ \  and \  $\tau(a)=0$ for $a\notin S$}\Big\}.
\end{equation*}
By \eqref{eq_SE_pairing}, $T_{S,0}$ is exactly the space of $\tau$ supported on $S$ and orthogonal to all strategically equivalent perturbations. Consequently, if there exists $\tau\in T_S\setminus T_{S,0}$, then we can pick $\delta_i$ so that $\sum_i \delta_i(\tau)\neq 0$.  Thus, either $\sum_i u_i(\tau)\neq 0$ already in $\Gamma$ or else $\sum_i u'_i(\tau)\neq 0$ in the strategically equivalent~$\Gamma'$.  Because the strategically equivalent perturbation can be chosen arbitrarily small, this establishes the density of $\mathfrak{U}_{S\subseteq A}$ once we know that $T_S\setminus T_{S,0}\neq\varnothing$ for the equilibria of interest.

It remains to show that $T_S\setminus T_{S,0}\neq\varnothing$ whenever at least three agents mix at the regular equilibrium. First, we compute the dimension of $T_{S,0}$.

\begin{lemma}\label{lm_zero_marginals}
The dimension of $T_{S,0}$ is
\begin{equation*}
\dim T_{S,0}\;=\;\prod_{i}\big(|S_i|-1\big).
\end{equation*}
In particular, $T_{S,0}=\{0\}$ whenever some agent $i$ has $|S_i|=1$.
\end{lemma}
\begin{proof}
We assume $S=A$ without loss. If $S_i=\{a_i^\ast\}$ for some $i$, then every $\tau\in T_{S,0}$ satisfies
$\tau(a_i^\ast,a_{-i})=\sum_{a_i}\tau(a_i,a_{-i})=0$ for all $a_{-i}$,
so $\tau=0$ and $\dim T_{S,0}=0$.
This agrees with the formula, since the product contains the factor $|S_i|-1=0$.

Now assume $|S_i|\geq 2$ for all $i$. Fix a reference action $a_i^\ast\in S_i$ for each $i$, and let 
$C=\prod_i (S_i\setminus\{a_i^\ast\})$.  
Every $\tau\in T_{S,0}$ is uniquely determined by its restriction to $C$.  
Indeed, if $a$ has at least one coordinate $a_i=a_i^\ast$, then from 
$\sum_{a_i}\tau(a_i,a_{-i})=0$ we get that
$
\tau(a_i^\ast,a_{-i})=-\sum_{a_i\ne a_i^\ast}\tau(a_i,a_{-i}),
$
so $\tau(a)$ is determined by values at profiles with fewer reference coordinates.  
By induction, all values of $\tau$ are determined once the restriction of $\tau$ to $C$ is known, implying that
\begin{equation*} \dim T_{S,0}\le |C|=\prod_i(|S_i|-1).\end{equation*} 
To prove the reverse inequality, we construct $|C|$ linearly independent vectors as follows. For each $c=(c_1,\dots,c_n)\in C$, define $\phi_c\in T_{S,0}$ by
\begin{equation*} 
\phi_c(a)=
\begin{cases}
(-1)^{|\{i:\,a_i=a_i^\ast\}|}, &\text{if }a_i\in\{c_i,a_i^\ast\}\ \forall i,\\[3pt]
0, &\text{otherwise.}
\end{cases}
\end{equation*} 
Each $\phi_c$ satisfies the zero-marginal conditions, and   $\phi_c$ coincides with the indicator of $c$ on~$C$.  
Thus, $\{\phi_c\}_{c\in C}$ are linearly independent, so $\dim T_{S,0}\ge |C|$. Therefore, 
$\dim T_{S,0}=|C|=\prod_i(|S_i|-1)$ and $\{\phi_c\}_{c\in C}$ is a basis of $T_{S,0}$.
\end{proof}

The following lemma takes care of the case of $n \geq 4$ agents and $k\geq 3$ mixers. 
\begin{lemma}\label{lm_4agents_zero_marginals}
Let $\nu$ be a regular Nash equilibrium supported on $S=S_1\times\cdots\times S_n$ with $n\geq 4$. If $k\geq 3$ agents use mixed strategies, then $T_S\setminus T_{S,0}\neq\varnothing$.
\end{lemma}
The idea is to show that the lower bound on $\dim T_S$ from Theorem~\ref{th_almost_pure_appendix} exceeds $\dim T_{S,0}$, so $T_S$ cannot be contained in $T_{S,0}$.

\begin{proof}
If $k<n$, some agent does not randomize, so by Lemma~\ref{lm_zero_marginals}, we know that $T_{S,0}=\{0\}$. By Theorem~\ref{th_almost_pure_appendix}, for any regular equilibrium with at least three mixers we have $\dim T_S>0$. Therefore, $T_S\setminus T_{S,0}\neq\varnothing$.

If $k=n$, then all agents mix, so $k=n\geq 4$.
Theorem~\ref{th_almost_pure_appendix} gives
\begin{equation*}
    \dim T_S\;\geq\; 2^k-2k-2 \;+\prod_{i\colon |S_i|\geq 2}(|S_i|-1).
\end{equation*}
By Lemma~\ref{lm_zero_marginals}, $\dim T_{S,0}=\prod_{i\colon |S_i|\geq 2}(|S_i|-1)$, so
\begin{equation}\label{eq_dimension_difference}
    \dim T_S - \dim T_{S,0} \;\geq\; 2^k-2k-2 \;\geq\; 6
    \qquad\text{for } k\geq 4.
\end{equation}
Hence, $T_S\not\subseteq T_{S,0}$. It follows that $T_S\setminus T_{S,0}\neq\varnothing$.
\end{proof}

It remains to consider the case of $n=3$ agents, all of whom mix. In this case, the dimension of $T_{S,0}$ can be as large as that of $T_S$.\footnote{It may even happen that $T_\nu=T_{S,0}$. For instance, let $S_i=\{0,1\}$ and define $u_i(a)=(-1)^{a_i+a_j}$ for each triple of distinct indices $i,j,k$. As in a generic $2\times 2\times 2$ game, both $T_S$ and $T_{S,0}$ are one-dimensional; in this example, they coincide and are spanned by $\tau(a)=(-1)^{a_1+a_2+a_3}$.} Accordingly, ruling out the inclusion $T_S\subseteq T_{S,0}$ requires a more refined analysis of the structure of these spaces. We show that this inclusion holds only for utility functions $u$ satisfying a certain algebraic condition, and therefore fails generically. This distinction between the cases $n=3$ and $n\geq 4$ reflects a fundamental difference in strategic structure: for $n\geq 4$, the inclusion is ruled out for all $u$ by Lemma~\ref{lm_4agents_zero_marginals}, whereas for $n=3$ it may arise.

To show that $T_S$ is not contained in $T_{S,0}$ generically when $n=k=3$, we first establish this on a set of games of positive measure and then extend the conclusion to all generic games via an algebraic argument. Our starting point is the set of games admitting at least one additional Nash equilibrium $\nu'$ with the same support as $\nu$. The difference $\tau=\nu'-\nu$ then lies in $T_S$, and its marginals are nonzero because $\nu$ and $\nu'$ are distinct, so $\tau\notin T_{S,0}$.\footnote{Although we use a second Nash equilibrium $\nu'$ to find a welfare improvement of $\nu$, the improvement does not necessarily correspond to moving towards $\nu'$: it can occur in the direction of $-\tau=\nu-\nu'$.}

The next two lemmas apply to any number $n\geq k\geq 3$ of agents, although we only need it for $n=k=3$.

\begin{lemma}\label{lm_generic_two_Nash}
Let $S = S_1 \times \cdots \times S_n \subseteq A$ satisfy $k = \bigl|\{i \colon |S_i| \geq 2\}\bigr| \geq 3$, and suppose there exists $u \in (\R^{A})^n$ admitting a regular Nash equilibrium with support $S$. Then the set of $u$ for which $T_S \setminus T_{S,0} \neq \varnothing $ has positive Lebesgue measure.
\end{lemma}

\begin{proof}
Since $S$ supports at least one regular Nash equilibrium, it satisfies the polygon inequality (Lemma~\ref{lm_polygon}). As the defining conditions of both $T_S$ and $T_{S,0}$ involve only the restriction of $u$ to $S$, we assume $A_i = S_i$  for all $i$ without loss. Furthermore, we can focus on $|S_i| \geq 2$ by replacing agents with singleton supports with dummy agents.

\cite*{mckelvey1997maximal} characterize the maximal number of fully-supported Nash equilibria that a generic game can have. By their Theorem~5.3, if $A_i = S_i$ satisfies the polygon inequality, there exists a positive-measure set of utility functions $u \in (\R^{S})^n$ for which the game admits at least $(n-1)!$ distinct regular Nash equilibria. In particular, since $k \geq 3$ implies that $(k-1)! \geq 2$, there are at least two such equilibria. Let $\nu$ and $\nu'$ be any two such regular equilibria, and set $\tau = \nu' - \nu$. The equilibrium indifference conditions are linear in the joint distribution, so together with $\sum_{a} \tau(a) = 0$, we have that $\tau$ satisfies the linear system given in Lemma~\ref{lm_cone}. This implies that $\tau \in T_S$ at $\nu$.

It remains to show $\tau \notin T_{S,0}$. Since $\nu \neq \nu'$ and both are probability distributions over $S$, there exists an agent $i$ such that $\nu_i \neq \nu_i'$: the $i$-th marginal of $\tau$ is nonzero. For $j\ne i$, the $(n-1)$-dimensional marginal of $\tau$ on $A_{-j}$
is nonzero as well: summing that marginal over all coordinates other than $i$
recovers the $i$-th marginal of $\tau$, which is nonzero, so the marginal on
$A_{-j}$ cannot vanish. Hence $\tau \notin T_{S,0}$, and therefore
$T_S \setminus T_{S,0} \neq \varnothing$.
\end{proof}

The following lemma shows that, in fact, the conclusion of Lemma~\ref{lm_generic_two_Nash} established for a positive-measure set of utility functions holds for generic utility profiles~$u$.

\begin{lemma}\label{lm_three_equal}
Let $S = S_1 \times \cdots \times S_n \subseteq A$ satisfy $k = \bigl|\{i \colon |S_i| \geq 2\}\bigr| \geq 3$, and suppose there exists $u \in (\R^{A})^n$ admitting a regular Nash equilibrium with support $S$. Then $T_S \setminus T_{S,0} \neq \varnothing$ for generic~$u$.
\end{lemma}

\begin{proof}
It suffices to show that the inclusion $T_S\subseteq T_{S,0}$ imposes a non-trivial algebraic condition on $u$. Consider the linear systems defining the two spaces $T_S$ and $T_{S,0}$. We identify a perturbation $\tau$ with an element of $\R^S$ by ignoring coordinates outside of $S$. A perturbation $\tau$ belongs to $T_S$ if it solves the homogeneous linear system of Lemma~\ref{lm_cone}. This system can be written as $B_S(u)\tau=0$, where $B_S(u)$ is a $\left(1+\sum_i |S_i|(|S_i|-1)\right)\times|S|$ matrix. In other words, $T_S=\ker B_S(u)$. Similarly, $T_{S,0}=\ker D_S$ for some $\left(\sum_i\prod_{j\neq i}|S_j|\right)\times|S|$ matrix $D_S$ (independent of $u$).

The inclusion $T_S\subseteq T_{S,0}$ holds if and only if the rank of the stacked matrix of $B_S(u)$ and $D_S$ equals the rank of $B_S(u)$:
\begin{equation*} 
    \mathrm{rank}\begin{bmatrix}B_S(u)\\D_S\end{bmatrix}=\mathrm{rank}\,B_S(u).
\end{equation*} 
Let $r$ denote the maximum value of $\mathrm{rank} B_S(u)$ over all $u \in (\R^{A})^n$. The entries of $B_S(u)$ depend linearly on $u$, so each $r \times r$ minor of $B_S(u)$ is a polynomial in $u$. By definition of $r$, at least one such minor is not identically zero. Therefore, $\{u \colon \mathrm{rank} B_S(u) < r\}$ is the common zero set of a collection of polynomials that are not all identically zero. Hence, it is an algebraic subset of $(\R^{A})^n$ of Lebesgue measure zero.

On the full-measure set $\{u\colon \mathrm{rank}\,B_S(u)=r\}$, the inclusion $T_S\subseteq T_{S,0}$ requires every $(r+1)\times(r+1)$ minor of the stacked matrix of $B_S(u)$ and $D_S$ to vanish. Denote by $q_S(u)$ the polynomial in $u$ defined as the sum of squares of all $(r+1)\times(r+1)$ minors of the stacked matrix. This polynomial depends only on the support $S$, and $T_S\subseteq T_{S,0}$ can only hold when $q_S(u)=0$.

It remains to show that $q_S$ is not identically zero. By Lemma~\ref{lm_generic_two_Nash}, there is a positive-measure set of utility functions for which $T_S\not\subseteq T_{S,0}$. Intersecting this set with the full-measure set $\{u\colon \mathrm{rank}\,B_S(u)=r\}$ yields a point $u^*$ at which the stacked matrix has rank of at least $r+1$. Hence, $q_S(u^*)\neq 0$, so $q_S$ is not identically zero. It follows that the set $\{u\colon q_S(u)=0\}$ is an algebraic set of measure zero.

We conclude that $T_S \setminus T_{S,0} \neq \varnothing$ for all $u$ outside the measure-zero algebraic set $\{u \colon q_S(u) = 0\} \cup \{u \colon \operatorname{rank} B_S(u) < r\}$.
\end{proof}

We can now provide the proof of the proposition.

\begin{proof}[Proof of Proposition~\ref{prop:payoff_extreme}]
By Lemmas~\ref{lm_4agents_zero_marginals} and~\ref{lm_three_equal}, if $S \subseteq A$ with $k = \bigl|\{i \colon |S_i| \geq 2\}\bigr| \geq 3$ supports a regular Nash equilibrium for some utilities, then $T_S \setminus T_{S,0} \neq \varnothing$ for generic~$u \in (\R^A)^n$. Fix such~$u$. We now show that there is a point in $\mathfrak{U}_{S \subseteq A}$ arbitrarily close to it so that $\mathfrak{U}_{S \subseteq A}$ is dense. If $u \in \mathfrak{U}_{S \subseteq A}$, the claim is immediate. Suppose $u \notin \mathfrak{U}_{S \subseteq A}$ and thus $\sum_i u_i(\tau)=0$ for all $\tau\in T_S$. Choose a perturbation $\tau \in T_S \setminus T_{S,0}$. By~\eqref{eq_SE_pairing}, there exist functions $\delta_i \colon A_{-i} \to \R$ such that $\sum_i \delta_i(\tau) > 0$. Let $u'_i = u_i + \delta_i$ be the resulting strategically equivalent game. The correlated equilibrium polytope, and hence $T_S$, is unchanged, and
for the same $\tau \in T_S$,
\begin{equation*} 
    \sum_i u'_i(\tau)
    = \sum_i u_i(\tau) + \sum_i \delta_i(\tau)
    = \sum_i \delta_i(\tau)
    > 0.
\end{equation*} 
Since the values of $\delta_i$ can be set to be arbitrarily small and $u$ is generic, $\mathfrak{U}_{S \subseteq A}$ is dense.
Lemma~\ref{lm_density_genericity} then implies that $\mathfrak{U}_{S \subseteq A}$
contains an open set of full measure and is therefore generic.

Define $\mathfrak{U}_A$ as the intersection of $\mathfrak{U}_{S \subseteq A}$ over all $S \subseteq A$ with $k = \bigl|\{i \colon |S_i| \geq 2\}\bigr| \geq 3$ that support a regular Nash equilibrium for at least some $u \in (\R^A)^n$, together with the set of utility profiles under which Nash equilibria are all regular. Each factor is generic, and the collection is finite, so $\mathfrak{U}_A$ is itself generic.

Fix $u \in \mathfrak{U}_A$ and let $\nu$ be a Nash equilibrium with $k \geq 3$ mixing agents. Denote its support by $S$. Because $\nu$ is regular and $u \in \mathfrak{U}_{S \subseteq A}$, there exists $\tau \in T_S$ with $\sum_i u_i(\tau) \neq 0$. Choose $\varepsilon > 0$ small enough so that $\mu_\pm = \nu \pm \varepsilon\tau$ are both correlated equilibria. Then,
\begin{equation*} 
    \sum_i u_i(\mu_\pm)
    = \sum_i u_i(\nu) \pm \varepsilon \sum_i u_i(\tau),
\end{equation*} 
so one of $\mu_\pm$ yields strictly higher utilitarian welfare than $\nu$,
establishing the first claim.
For the second, observe that
$\nu = \tfrac{1}{2}(\mu_+ + \mu_-)$ with $u(\mu_+)\neq u(\mu_-)$,
so the payoff vector of $\nu$ is the midpoint of the payoff vectors of
$\mu_+$ and $\mu_-$ and is therefore not an extreme point of the
correlated equilibrium payoff set.
\end{proof}

\subsection{Pareto Improvements: Proof of Proposition~\ref{prop_Pareto}}\label{app_Pareto}

Proposition~\ref{prop_Pareto} shows that under substantial randomization, the expected payoff vector of a Nash equilibrium lies in the interior of the payoff polytope $U^{CE}(\G)$. To prove this, we first establish the following result giving a lower bound on the dimension of the face of $U^{CE}(\G)$ containing the Nash equilibrium payoff vector, and then derive Proposition~\ref{prop_Pareto} as a corollary.

\begin{proposition}\label{prop_payoff_dimension}
Let $\nu$ be a Nash equilibrium of a generic $n$-agent game $\G$, and let $S_i$ be the support of agent $i$'s mixed strategy. Suppose that at least $12$ agents mix. Then the expected payoff vector $u(\nu)$ belongs to a face of $U^{CE}(\G)\subset \R^n$ of dimension at least
\begin{equation*}
\min\{n,Q\}, \qquad \text{where}\qquad Q=\prod_{i=1}^n |S_i|-1-\sum_{i=1}^n |S_i|(|S_i|-1).
\end{equation*}
Moreover, $\nu$ itself lies in a face of $\ce(\G)$ of dimension $Q$.
\end{proposition}

In particular, if $Q\ge n$, then $u(\nu)$ lies in the interior of $U^{CE}(\G)$.

As discussed in Appendix~\ref{app_th1_proof}, a regular equilibrium $\nu$ supported on $S = S_1 \times \cdots \times S_n$ admits a tangent space $T_\nu$ that captures every perturbation direction $\tau$ along which moving from $\nu$ by $\pm\varepsilon\tau$ for small enough $\varepsilon > 0$ remains inside the correlated equilibrium polytope. The tangent space is characterized as a solution of a linear system of equations (Lemma~\ref{lm_cone}), and the face of $\ce(\G)$ containing $\nu$ has dimension $\dim T_\nu$.

An analogous description applies to the payoff polytope. The image of $T_\nu$ under the linear map $L_u \colon T_\nu \to \R^n$ defined by $L_u(\tau) = u(\tau) = (u_1(\tau), \ldots, u_n(\tau))$ provides directions along which $u(\nu)$ can be perturbed while remaining in $U^{CE}(\Gamma)$. Consequently, the dimension of the face of $U^{CE}(\Gamma)$ containing $u(\nu)$ is at least $\dim(\operatorname{Im}(L_u))$. This bound need not be tight: the dimension of that face can strictly exceed $\dim(\operatorname{Im}(L_u))$, for instance when the entire face of $\operatorname{ce}(\Gamma)$ containing $\nu$ projects into the relative interior of a higher-dimensional face of the payoff polytope.

The proof of Proposition~\ref{prop_payoff_dimension} proceeds as follows. For any admissible support profile, we construct a game for which the linear system in Lemma~\ref{lm_cone} has full rank and the associated payoff map has rank $\min\{n,Q\}$. We then invoke an algebraic argument to conclude that these properties hold for generic games.

\begin{lemma}\label{lm_non_degenerate_game}
For any non-empty finite sets $S_i$, $i=1,\ldots,n$, with at least $12$ sets containing more than one element and satisfying $|S_i|-1\leq \sum_{j\ne i}(|S_j|-1)$ for all $i$, there exists a game $\Gamma^*$ with action sets $S_i$ such that
\begin{itemize}
    \item the uniform distribution $\nu$ is a Nash equilibrium;
    \item $\nu$ belongs to a face of $\ce(\Gamma^*)$ of dimension $Q$ given in~\eqref{eq_cone_bound};
    \item the image of the tangent space $T_\nu$ under the payoff map is of dimension $\min\{n,\,Q\}$.
\end{itemize}
\end{lemma}
To construct such a game we need the following combinatorial lemma on packing disjoint cylinders in a product set $S = S_1\times\cdots\times S_n$. For a subset $B_{-i}\subseteq S_{-i}$, we call $C_i = S_i\times B_{-i}$ an $i$-cylinder. In the following game-theoretic construction, $i$-cylinders capture the coordinates constrained by the incentive constraints of agent $i$, and we require these cylinders to be disjoint so that the corresponding incentive constraints act on different coordinates.

\begin{lemma}\label{lm_packing}
For sets $S_1,\ldots, S_n$ with $n\geq 3$ and $|S_i|\geq 2$, there exist disjoint $i$-cylinders $C_i$, $i=1,\ldots,n$, of sizes 
\begin{equation}\label{eq_pack_bound}
    |C_i|\geq\frac{1}{(2n)^{\log_2 3}} \prod_{j=1}^n |S_j|.
\end{equation}
\end{lemma}

For binary $S_i = \{0,1\}$, cylinder packing was studied by~\cite*{felzenbaum1993packing}. To the best of our knowledge, Lemma~\ref{lm_packing} is the first result covering the non-binary and possibly asymmetric case.

The idea of the proof is to designate $k \simeq \log_2 n$ coordinates as separation coordinates, and to split each corresponding set $S_i$ into two halves. A binary code on $k$ bits then assigns to each cylinder which half to occupy at each separation coordinate, thereby ensuring separation. The remaining coordinates are common across all cylinders and contribute to volume. The main subtlety arises for cylinders $C_i$ where the index $i$ itself is one of the separation coordinates. That is, the code must allow recovery of the intended assignment even when one coordinate is unavailable.

\begin{proof}
We use $\lceil x \rceil$ and $\lfloor x \rfloor$ to denote the smallest integer above $x\in \R $ and largest integer below $x$, respectively, and use $[m]$ to denote the set $\{1,\ldots, m\}$.

Choose the first $k=\lceil\log_2 n\rceil+1$ coordinates as separation coordinates and designate the remaining $n-k$ as {capacity} coordinates. The assumption that $n\geq 3$ ensures that $3\leq k\leq n$. For each $j\in[k]$, we split  $S_j=S_{j,0}\,\cup\, S_{j,1}$ so that  $|S_{j,0}|=\Big\lfloor\frac{|S_j|}{2}\Big\rfloor$ and $|S_{j,1}|=\Big\lceil\frac{|S_j|}{2}\Big\rceil.$

We now construct a mapping $f:[n]\to\{0,1\}^{k}$. The vector $f(i)$ can be thought of as a codeword, and is effectively a vector of dummy variables indicating, for each $j \in [k]$, whether the cylinder $C_i$ contains $S_{j,0}$ or $S_{j,1}$ in its $j$-th coordinate. Namely, for each $i\in[n]$, we define:
\begin{equation*}
B_{-i}\ =\ \Big(\prod_{\substack{1\le j\le k\\ j\ne i}} S_{j,\,f_j(i)}\Big)\times \Big(\prod_{\substack{k< j\le n\\ j\ne i}} S_j\Big),
\qquad
C_i\ =\ S_i\times B_{-i}\ \subset\ S.
\end{equation*}

We construct the mapping $f$ so that it respects the following separation property:
\begin{equation}\label{eq:sep_prop}
\forall\, i\ne i'\ \ \  \exists\, j\in[k]\setminus\{i,i'\}\quad \text{ such that }\quad f_j(i)\ne f_j(i').
\end{equation}
This separation property guarantees that the cylinders are disjoint. 
To see this, fix any $i \ne i'$ and let $j \in [k]\setminus\{i,i'\}$ be the index whose existence is guaranteed by~\eqref{eq:sep_prop}. The projections of $C_i$ and $C_{i'}$ onto the $j$-th coordinate are $S_{j,0}$ and $S_{j,1}$ and since these sets are disjoint, we conclude $C_i \cap C_{i'} = \varnothing$.

Specifically, we define $f(i)=e_i+e_{i+1\bmod k}$ for $i\in[k]$, where $e_1,\ldots, e_k$ are standard basis vectors. We then pick distinct $f(i)$ for all $i>k$ in the set $\{0,1\}^k\setminus  F$, where
\begin{equation*} F = \big\{e_j,\  j\in[k]\big\}\ \cup\  \big\{f(j), \ j\in[k]\big\}.\end{equation*} 
Since $F$ contains $2k$ elements, that is the number of excluded codewords, leaving $2^k-2k$ possible codewords for $n-k$ values of $i>k$. Since $k=\lceil\log_2 n\rceil+1$, we have $2^k\ge 2n$, hence $2^k-2k\ge 2n-2k\geq n-k$, so the assignment is possible.

We now verify that \eqref{eq:sep_prop} holds for the chosen $f$. For $i,i'\in [k]$, the non-zero entries are  $\{i,\ (i+1) \bmod k\}$ and $\{i', \ (i'+1)\bmod k\}$. Since $k\ge 3$, the union of these two sets contains some $j\notin\{i,i'\}$ such that $f_j(i) \neq f_j(i')$. Suppose $i\le k<i'$. Then, $j \neq i'$ holds automatically for any $j\in[k]$. 
By construction, the vectors $e_{i+1}$ and $e_i+e_{i+1\bmod k}$ are excluded from the codewords assigned to indices $i'>k$. Therefore, either $f(i')$ takes the value $1$ at some entry $j\in [k]\setminus\{i,(i+1) \bmod k\}$, in which case $f_j(i)\neq f_j(i')$, or $f(i')\equiv 0$, in which case $f_{j}(i)=1\neq 0=f_{j}(i')$ for $j=(i+1) \bmod k$. Finally, consider $i,i'>k$. By construction, $f(i)\neq f(i')$, so  there exists $j\in[k]$ such that $f_j(i)\neq f_j(i')$. Since $j\in[k]$, we have $j\notin\{i,i'\}$.

It remains to bound the size of each constructed cylinder. For each $i$, we restrict at most $k$ separation coordinates, yielding
\begin{equation*}
|C_i|\ \ge\ \Big(\prod_{j=1}^n |S_j|\Big)\,\prod_{j=1}^{k}\frac{\big\lfloor |S_j|/2\big\rfloor}{|S_j|}.
\end{equation*}
Since $\lfloor m/2\rfloor = m/2$ for even $m$ and $\lfloor m/2\rfloor \ge \frac{m-1}{2}= \frac{m}{2}\left(1-1/m\right)$ for odd $m$,
\begin{equation*}
|C_i|\ \ge\ 2^{-k}\left(\prod_{\substack{1\le j\le k\\ |S_j|\ \text{is odd}}}\big(1-\tfrac{1}{|S_j|}\big)\right)\,\prod_{j=1}^n |S_j|.
\end{equation*}
If $|S_j|$ is odd, then $|S_j|\geq 3$ and thus $(1-1/|S_j|)\geq \frac{2}{3}$. We obtain 
\begin{equation*}
|C_i|\ \ge\ 3^{-k}\,\prod_{j=1}^n |S_j|.
\end{equation*}
Since $2^k\leq 2n$, we get $3^k\leq (2n)^{\log_2 3}$, which results in~\eqref{eq_pack_bound}.
\end{proof}

The following result, which will be useful in proving Lemma~\ref{lm_non_degenerate_game}, is deduced from Lemma~\ref{lm_packing} under the additional assumptions that $n\geq 12$ and the sizes $|S_i|$ satisfy the polygon inequality.

\begin{lemma}\label{lm_packing_quadratic}
For sets $S_1,\ldots, S_n$ with $n\geq 12$ satisfying
$|S_i|\geq 2$ and
$|S_i|-1\leq \sum_{j\ne i } (|S_j|-1),$ there exist disjoint $i$-cylinders $C_i=S_i\times B_{-i}$, $i=1,...,n$, of sizes 
    $|C_i|> |S_i|^2.$
\end{lemma} 
\begin{proof}
By Lemma~\ref{lm_packing}, we know that there exist disjoint $C_i$ of size given by the right-hand side of~\eqref{eq_pack_bound}. Thus, it suffices to demonstrate that under our assumptions, $|S_i|^2$ does not exceed this right-hand side. That is, it suffices to show that, for all $i$, 
\begin{equation*}     |S_i|^2<\frac{1}{(2n)^{\log_2 3}} \prod_{j=1}^n |S_j|.
\end{equation*} 
Dividing both sides by $|S_i|^2$ and taking into account that $|S_i|\leq 1+\sum_{j\ne i } (|S_j|-1)$ from the polygon inequality, it suffices to verify that
\begin{equation*}     
1<\frac{1}{(2n)^{\log_2 3}}\  \frac{\prod_{j\ne i} |S_j|}{1+\sum_{j\ne i } (|S_j|-1)}.
\end{equation*} 
Denote the value of the denominator by $M$, which satisfies $M\geq n$ since $|S_j|\geq 2$ for all $j$. For a given value of $M$, the minimal value of the right-hand side is attained when all $|S_j|$ with $j\ne i$ are equal to $2$ except one, which equals $M-n+2$. In this case, the right-hand side is equal to
\begin{equation*} \frac{2^{n-2}(M-n+2)}{M(2 n)^{\log_2 3}}.\end{equation*} 
The minimum over $M$ is attained at $M=n$ and is equal to
\begin{equation*} 
\frac{2^{n}}{(2 n)^{1+\log_2 3}}.
\end{equation*} 
Therefore, it suffices to show that this ratio is bigger than $1$ for $n\geq 12$. Equivalently, we need to show that $f(n)> 0$, where
\begin{equation*}
f(x)=\log_2\!\left(\frac{2^{x}}{(2x)^{1+\log_2 3}}\right).
\end{equation*}
A direct evaluation gives $f(12)\approx 0.15>0$. The derivative of $f$,
\begin{equation*}
f'(x)=1-\frac{1+\log_2 3}{x \ln 2},
\end{equation*}
is increasing in $x$ and $f'(12)\approx 0.69>0$. Thus, $f(x)$ is increasing for $x\geq 12$ and  $f(n)>0$ for all $n\ge 12$.
\end{proof}

With the help of Lemma~\ref{lm_packing_quadratic}, we can now prove lemma~\ref{lm_non_degenerate_game}.

\begin{proof}[Proof of Lemma~\ref{lm_non_degenerate_game}]
The game $\Gamma^*$ is constructed as follows. Agent $i$'s action set is $S_i$. First, we assume that each $S_i$ contains at least two elements; that is, there are no dummy agents who have utilities over action profiles but do not take any strategic actions.

By Lemma~\ref{lm_packing_quadratic}, for each $i$ we can pick $B_{-i}\subset S_{-i}$ such that $C_i=S_i\times B_{-i}$ are disjoint $i$-cylinders, each of size at least $|S_i|^2$. By shrinking each $B_{-i}$ if necessary, we may assume without loss of generality that $|B_{-i}|=|S_i|$.

Let $f_i$ be a bijection $S_i\to B_{-i}$. Define $i$'s utility by
\begin{equation*} 
u_i^*(a_i, a_{-i})=\left\{\begin{array}{cc} 1,  & a_{-i}=f_i(a_i)  \\
   0,  & \text{otherwise}
\end{array}
\right.
\end{equation*} 
In words, agent~$i$ aims to predict the profile of actions of others. She can only announce one of $|S_i|$ predictions from $B_{-i}\subset S_{-i}$, and gets utility of $1$ if she accurately predicts what the others do and $0$, otherwise.

If others randomize over all profiles uniformly, no prediction is better than the other. Hence, each agent playing the uniform distribution $\nu_i$ over $S_i$ is a Nash equilibrium of this game.

Now consider the tangent space $T_\nu$ at $\nu$ described in Lemma~\ref{lm_cone}. The incentive constraint 
\begin{equation*} 
\sum_{a_{-i}} \tau(a_i,a_{-i}) (u_i^*(a_i,a_{-i}) - u_i^*(a_i',a_{-i})) = 0, \quad \forall i, \forall a_i \neq a_i' \in S_i
\end{equation*} 
reduces to 
\begin{equation*} 
\tau\big(a_i,f_i(a_i)\big)-\tau\big(a_i,f_i(a_i')\big)=0, \quad \forall i, \forall a_i \neq a_i' \in S_i.
\end{equation*} 
That is, $\tau(a_i,\,\cdot\,)$ is equal to some constant $c_i(a_i)$ on $B_{-i}$. Since the sets $S_i \times B_{-i}$ do not overlap, these constancy constraints
act on disjoint coordinates of $\tau$ and are therefore linearly independent across
different $i$. The condition $\sum_a \tau(a)=0$  boils down to 
\begin{equation}\label{eq_zero_weight}
    \sum_{a\not\in \cup_i S_i\times B_{-i}} \tau(a) + \sum_i |S_i|\sum_{a_i}  c_i(a_i)=0.
\end{equation}
Accordingly choosing $c_i(a_i)\in \R$ and $\tau(a)\in\R$ for $a\not\in \cup_i S_i\times B_{-i}$ so that~\eqref{eq_zero_weight} holds, we obtain all the vectors from the tangent space. We conclude that its dimension is given by the right-hand side of~\eqref{eq_cone_bound}
\begin{equation*}
    \dim (T_\nu)=\sum_i |S_i| + \left(\prod_i |S_i|-\sum_i |S_i|^2\right)-1=Q.
\end{equation*}
Now consider the image of $T_\nu$ under the payoff map $L_{u^*}\colon T_\nu\to \R^n$ mapping $\tau$ to $(u_1^*(\tau), \dots, u_n^*(\tau))$. The dimension of the face of the payoff polytope $U^{CE}(\G)$ containing $u^*(\nu)$ is the dimension of the image of this map, $\dim(\text{Im}(L_{u^*}))$.

By the rank-nullity theorem, $\dim(\text{Im}(L_{u^*})) = \dim(T_\nu) - \dim(\ker(L_{u^*}))$. A vector $\tau\in T_\nu$ is in the kernel if 
$\sum_a u_i^*(a)\tau(a)=0$ for all $i$. Since $u_i^*(a)$ only equals $1$ on $a=(a_i, f_i(a_i))$ and $\tau(a)=c_i(a_i)$ on such $a$, we get
\begin{equation}\label{eq_zero_component_weight}
\sum_{a_i} c_i(a_i)=0.
\end{equation}
and thus~\eqref{eq_zero_weight} simplifies to
\begin{equation}\label{eq_zero_weight_simplified}
       \sum_{a\not\in \cup_i S_i\times B_{-i}} \tau(a)=0.
\end{equation}
We note that the left-hand side of~\eqref{eq_zero_weight_simplified} is not identically zero since $S\setminus \cup_i S_i\times B_{-i}\ne \varnothing$. To see this, note that by 
Lemma~\ref{lm_packing_quadratic}, $S$ admits a partition into cylinders $C_i$ with $|C_i| > |S_i|^2$ 
for each $i$, which gives $|S| = \sum_i |C_i| > \sum_i |S_i|^2$. On the other hand, 
$|\bigcup_i S_i \times B_{-i}| = \sum_i |S_i|^2$, so the complement is nonempty.

We obtain that $\tau $ is from the kernel if and only if $c_i(a_i)\in \R$ and $\tau(a)\in \R$ for $a\not\in \cup_i S_i\times B_{-i}$ are chosen arbitrary so that~\eqref{eq_zero_component_weight} and~\eqref{eq_zero_weight_simplified} are satisfied. As these conditions constrain disjoint sets of coordinates, we obtain that
\begin{equation*} 
\dim(\ker(L_{u^*}))=\sum_i (|S_i|-1) + \left(\prod_i |S_i|-\sum_i |S_i|^2\right)-1=Q-n.
\end{equation*} 
We conclude that $\dim(\text{Im}(L_{u^*}))=n$ and so $u^*(\nu)$ is contained in the face of the payoff polytope $U^{CE}(\G^*)$ of dimension~$n$, that is, in its interior.

It remains to consider the case when some agents are dummy agents with $|S_i|=1$, who make no strategic decisions but still receive some payoffs. Suppose the total number of agents is $n'$ where agents $1,\ldots, n$ with $n<n'$ are not dummies and $n+1,\ldots n'$ are dummies.

Consider the game $\Gamma^*$ constructed for agents $1,\ldots, n$ as above. Now we add $n'-n$ dummy agents sequentially. These agents do not make any strategic decisions and thus do not contribute to incentive constraints and to the dimension of $T_\nu$. Now we choose utilities $u_i^*$ of dummy agents. There are no restrictions on them and thus they are just arbitrary functionals on $\R^S$. We pick these utilities so that $\dim(\ker(L_{u^*}))$ decreases by one with each new dummy agent unless $\dim(\ker(L_{u^*}))$ is already zero. By doing so we ensure that 
$\dim(\ker(L_{u^*}))=\sum_i (|S_i|-1) + \left(\prod_i |S_i|-\sum_i |S_i|^2\right)-1=\max\{Q-n',0\}$ and thus  
$\dim(\text{Im}(L_{u^*}))=\min\{n',\,Q\}$.
Hence $u^*(\nu)$ lies in a face of dimension at least $\min\{n',\,Q\}$ of the payoff polytope.
\end{proof}

Using an algebraic argument, we can now prove
Proposition~\ref{prop_payoff_dimension} by showing that a generic game shares
the properties of the game constructed in Lemma~\ref{lm_non_degenerate_game}.

\begin{proof}[Proof of Proposition~\ref{prop_payoff_dimension}]
Let $\nu$ be a Nash equilibrium of a generic $n$-agent game $\G$ with support
$S = S_1\times\cdots\times S_n$. Recall that the dimension of the face of the
payoff polytope $U^{CE}(\G)$ containing $u(\nu)$ is at least
$\dim\operatorname{Im}(L_u)$, where $L_u\colon T_\nu\to\R^n$ is the linear
map defined by $L_u(\tau) = (u_1(\tau),\ldots,u_n(\tau))$ with
$u_i(\tau)=\sum_a u_i(a)\tau(a)$. By the rank-nullity theorem,
\begin{equation*}
\dim\operatorname{Im}(L_u) = \dim T_\nu - \dim(\ker(L_u)),
\end{equation*}
so it suffices to determine the generic values of $\dim T_\nu$ and
$\dim(\ker(L_u))$ separately.

By the discussion in Appendix~\ref{app_polygon}, in a generic game all Nash
equilibria are regular. For such a game, Lemma~\ref{lm_polygon} implies that
the support sizes satisfy the polygon inequality. Fix a support profile $S$
with $|S_i|\geq 2$ for at least $12$ agents, and set
$M = 1+\sum_{i}|S_i|(|S_i|-1)$ and $Q = |S|-M$.

Let $B_S(u)$ denote the $M\times|S|$ coefficient matrix of the
support-restricted linear system from Lemma~\ref{lm_cone}, so that
$T_\nu = \ker B_S(u)$ whenever $u$ admits a regular equilibrium with support
$S$. The entries of $B_S(u)$ are linear in $u$. By
Lemma~\ref{lm_non_degenerate_game}, there exists a utility profile $u^*$ with
$\operatorname{rank}B_S(u^*)=M$. Hence the polynomial $p_S(u)$, defined as
the sum of squares of all $M\times M$ minors of $B_S(u)$, is not identically
zero. On the complement of the algebraic set $\{p_S=0\}$, the matrix
$B_S(u)$ has full row rank, and therefore $\dim T_\nu = |S|-M = Q$.

To compute the generic value of $\dim(\ker(L_u))$, let $U_S(u)$ be the
$n\times|S|$ matrix whose $i$-th row is $(u_i(a))_{a\in S}$. For a regular
equilibrium with support $S$,
\begin{equation*}
\ker(L_u) = \ker\begin{bmatrix}B_S(u)\\U_S(u)\end{bmatrix}.
\end{equation*}
Set $M' = \min\{|S|,\,M+n\}$. By Lemma~\ref{lm_non_degenerate_game}, there
exists $u^*$ for which the stacked matrix has rank $M'$. Hence the polynomial
$q_S(u)$, defined as the sum of squares of all $M'\times M'$ minors of the
stacked matrix, is not identically zero. Outside the algebraic set
$\{q_S=0\}$, the stacked matrix has rank $M'$, so
$\dim(\ker(L_u)) = |S|-M' = \max\{0,\,Q-n\}$.

Combining the two computations,
\begin{equation*}
\dim\operatorname{Im}(L_u) = \dim T_\nu - \dim(\ker(L_u)) = \min\{n,\,Q\}.
\end{equation*}
Since there are only finitely many support profiles $S$, excluding finitely many algebraic sets $\{p_S=0\}$ and $\{q_S=0\}$ and intersecting
with the full-measure set of utility profiles under which Nash equilibria are all regular yields an open set of full measure
on which every Nash equilibrium with at least $12$ mixing agents satisfies
the stated conclusion.
\end{proof}

Proposition~\ref{prop_Pareto} follows directly by combining Proposition~\ref{prop_payoff_dimension} with the bound in Theorem~\ref{th_almost_pure_appendix} on the dimension of a face of the correlated equilibrium polytope that contains a Nash equilibrium.
\begin{proof}[Proof of Proposition~\ref{prop_Pareto}]
Consider a Nash equilibrium $\nu$ in a generic $n$-agent game $\Gamma$ in which at least $k\geq 8+\log_2 n$ agents randomize their actions. Since $k\le n$, it follows that $n\geq 8+\log_2 n$; in particular, both $k$ and $n$ are at least $12$. Hence we may apply Proposition~\ref{prop_payoff_dimension} to conclude that $u(\nu)$ lies on a face of the correlated-equilibrium payoff polytope $U^{CE}(\Gamma)$ of dimension at least $\min\{n,\, Q\}$. It remains to show that $Q\geq n$. By Proposition~\ref{prop_payoff_dimension}, $Q$ is the dimension of the face of $\ce(\Gamma)$ containing $\nu$. By Theorem~\ref{th_almost_pure_appendix}, we have the lower bound
\begin{equation*} Q\geq 2^{k-3}.\end{equation*} 
Since $k\geq 8+\log_2 n$, we get
$Q\geq 32 n$
and therefore $\min\{n,\,Q\}=n$. Thus, $u(\nu)$ lies in the interior of the correlated-equilibrium payoff polytope $U^{CE}(\Gamma)$.
\end{proof}

\section{Proofs Pertaining to Symmetric Games}\label{app_symm_games}

We start by proving the following proposition, which implies 
Corollary~\ref{cor_not_extreme_totally_mixed}. 
\begin{proposition}\label{prop_symmetric}
Let $\Gamma$ be a symmetric game with $n\geq 3$ agents. If $\nu$ is a symmetric totally-mixed Nash equilibrium, then $\nu$ is not extreme within the set of symmetric correlated equilibria.   
\end{proposition}

This proposition does not follow from the preceding analysis. The set of symmetric correlated equilibria is generally smaller than the set of all correlated equilibria, so $\nu$ 
may potentially be non-extreme in the larger set while remaining extreme in the symmetric subset.

Let $X$ be a finite set with $m$ elements. A distribution $\mu$ over an $n$-fold product space $X^n = X\times \ldots \times X$ is called \textbf{exchangeable} if for any element $(x_1,\ldots, x_n) \in X^n$ and any 
     permutation $\pi$, we have
     \begin{equation*}
\mu(x_1,\ldots, x_n)=\mu(x_{\pi(1)},\ldots, x_{\pi(n)}).
\end{equation*}
The set of exchangeable distributions is a convex set. Furthermore, any symmetric correlated equilibrium is an exchangeable distribution.
\begin{lemma}\label{lm_extreme_symmetric_correlated}
    If $\mu$ is an extreme point of the set of symmetric correlated equilibria, then $\mu$ can be represented as a convex combination of at most 
    \begin{equation*}
m(m-1)+1
\end{equation*}
extreme exchangeable distributions, where $m$ is the number of actions of each agent.
\end{lemma}
The set of exchangeable distributions is universal in that it only depends on the number of agents and actions in the game, but not on their associated payoffs.
\begin{proof}[Proof of Lemma \ref{lm_extreme_symmetric_correlated}]
To check that an exchangeable distribution is a correlated equilibrium of a symmetric game, it suffices to check the incentive constraints for a single agent since, by symmetry, any two agents are equivalent up to relabeling. For a single agent, there
are $m(m-1)$ incentive constraints. The result then follows from \cite*{winkler1988extreme}.
\end{proof}   
Exchangeable distributions are often studied on the infinite product space $X^{\Z}$. In this setting, de Finetti's Theorem shows that exchangeable sequences admit a representation as conditionally independent draws given a latent variable. Equivalently, the extreme points of the set of exchangeable distributions on $X^{\Z}$ are i.i.d. distributions. For a finite product space, the set of extreme exchangeable distributions has a different structure analyzed by \cite*{diaconis1980finite}.
\begin{lemma}[\citet*{diaconis1980finite}]\label{lm_extreme_exchangeable}
Denote by $\delta_{(x_1,\ldots, x_n)}$ the distribution that places unit mass on the profile $(x_1,\ldots, x_n)$. Any distribution $\mu$ over $X^n$ is an extreme point of exchangeable distributions if and only if it can be represented as
\begin{equation}\label{eq_extreme_exchangeable}
\mu=\frac{1}{n!}\sum_{\pi} \delta_{(x_{\pi(1)},\ldots, x_{\pi(n)})}.
\end{equation}
\end{lemma}  

Extreme exchangeable distributions may have different supports because not all terms in the sum in~\eqref{eq_extreme_exchangeable} are distinct. For example, suppose $n=3$ and $X=\{a,b\}$. Then, there are four extreme exchangeable distributions
\begin{equation*}
    \delta_{(a,a,a)}, \quad \frac{1}{3}\left(\delta_{(a,a,b)}+\delta_{(a,b,a)}+\delta_{(b,a,a)}\right), \quad \frac{1}{3}\left(\delta_{(a,b,b)}+\delta_{(b,a,b)}+\delta_{(b,b,a)}\right), \quad\text{and}\quad \delta_{(b,b,b)}.
\end{equation*}
Without loss of generality, enumerate the elements of $X$ thus identifying $X$ and $\{1,\ldots,m\}$. Let $\Delta_{m,n}$ be the discrete simplex, which consists of all nonnegative integer vectors $(k_1,\ldots, k_m)$ such that $k_1+\ldots+k_m=n$. For each $x=(x_1,\ldots, x_n)$, we can assign the frequency vector $f(x)$ in $\Delta_{m,n}$ that counts the number of times each $j \in X=\{1,\dots,m\}$ appears in $x$. That is, 
\begin{align*}
f(x)_j=|\{i\colon x_i=j\}|.
\end{align*}
Since an extreme exchangeable distribution $\mu$ is obtained via symmetrization of some $x=(x_1,\ldots, x_n)$, we can associate the frequency vector~$f(x)$ with $\mu$.

The total number of different frequency vectors is given by\footnote{This can be derived by what is often referred to as the Stars and Bars Lemma in combinatorics. Namely, one can think of $n$ items to be split into $m$ separate buckets. If the items are thought of as stars on a line, the buckets can be represented by $m-1$ bars that are interspersed among the stars and partition them into $m$ subsets. Our derivation is then equivalent to identifying the number of choices of these $m-1$ locations out of $n+m-1$ possibilities.}
\begin{equation}\label{eq_bars_and_stars}
    |\Delta_{m,n}|=\binom{n+m-1}{m-1}.
\end{equation}
By Lemma~\ref{lm_extreme_exchangeable}, there is a natural bijection between extreme exchangeable distributions and $\Delta_{m,n}$ and thus equality~\eqref{eq_bars_and_stars} provides the total number of extreme exchangeable distributions.

\begin{lemma}\label{lm_extreme_exchangeable_product_decomposition}
Consider a product distribution $\tau=\nu\times \ldots \times \nu$ over $X^n$ with full support. If $\tau$ is represented as a mixture of extreme exchangeable distributions, then all
\begin{equation*}
\binom{n+m-1}{m-1}
\end{equation*}
of them must enter the mixture with a positive weight.
\end{lemma}   
\begin{proof}
Let $x=(x_1,\ldots,x_n)$ be drawn from $\tau$, and let $f(x)\in \Delta_{m,n}$ 
denote its frequency vector. Since $\tau$ has full support, every $x\in X^n$ 
occurs with positive probability, and hence $f(x)$ takes every value in 
$\Delta_{m,n}$ with positive probability.

Each extreme exchangeable distribution is supported on the set of sequences 
with a fixed frequency vector. Therefore,  
each extreme exchangeable distribution must enter the mixture with a strictly 
positive weight. The number of such distributions is $|\Delta_{m,n}|$, given by 
\eqref{eq_bars_and_stars}.
\end{proof}

We now combine these lemmas to prove the proposition.

\begin{proof}[Proof of Proposition \ref{prop_symmetric}]
Consider a symmetric game with $n\geq 3$ agents each having $m$ actions and  a symmetric totally-mixed  Nash equilibrium $\nu$.

By Lemma~\ref{lm_extreme_symmetric_correlated}, any extreme symmetric correlated equilibrium is a convex combination  of at most $m(m-1)+1$ extreme exchangeable distributions. On the other hand, by Lemma~\ref{lm_extreme_exchangeable_product_decomposition}, we need at least $\binom{n+m-1}{m-1}$ extreme exchangeable distributions to represent $\nu$ as a mixture.
 Thus, $\nu$ is not an extreme point of the set if  
\begin{equation}\label{eq_symm_condition_not_extreme}
    \binom{n+m-1}{m-1}>m(m-1)+1.
\end{equation}
 We now demonstrate that~\eqref{eq_symm_condition_not_extreme} is satisfied for any $n\geq 3$ and $m\geq 2$.
The left-hand side is monotone-increasing in $n$, and thus, it is enough to consider the case of $n=3$. We obtain
\begin{align*}
\binom{m+2}{m-1}>m(m-1)+1,
\end{align*}
or, equivalently,
\begin{align*}
\frac{(m+2)(m+1)m}{6}>m(m-1)+1.
\end{align*}
Elementary computations show that this inequality holds for all $m\geq 2$.
\end{proof}   

 Proposition~\ref{prop_structure_of_extreme_symmetric_CE}, which represents extreme symmetric correlated equilibria via urns, follows directly from the lemmas above.

\begin{proof}[Proof of Proposition~\ref{prop_structure_of_extreme_symmetric_CE}]
By Lemma~\ref{lm_extreme_exchangeable}, every extreme exchangeable distribution admits the urn representation claimed in the proposition. Specifically, a measure $\mu$ of the form~\eqref{eq_extreme_exchangeable} corresponds to an urn containing $n$ balls labeled by $x_1,\ldots, x_n$, from which balls are drawn without replacement. The result then follows by combining this observation with Lemma~\ref{lm_extreme_symmetric_correlated}.
\end{proof}

\subsection{Proofs Pertaining to Binary Action Games}\label{app_binary}

\begin{proof}[Proof of Proposition~\ref{prop:binary_payoffs}]
Consider a Nash equilibrium $\nu$, where agent $i$ chooses action $1$ with probability $p_i\in[0,1]$ which may vary across agents. 

Consider an agent $i$. Recall that the  payoff of action $0$ is always zero. Hence, if $p_i=0$, i.e., $i$ plays action $0$ deterministically, then her equilibrium payoff is $0$. If $p_i\in (0,1)$, i.e., $i$ mixes between $0$ and $1$, her expected payoff is also zero as she must be indifferent between the two actions. 

It remains to consider the case of $p_i=1$, that is, $i$ plays action $1$ deterministically. Her expected payoff 
\begin{equation}\label{eq_payoff_i}
    \E[u_i(1,a_{-i})]=\E\left[f\left(\frac{1+\sum_{k\ne i}{a_k}}{n}\right)\right]
\end{equation}
cannot be negative as otherwise she would switch to action~$0$. We now bound it from above. Since $f(1)<0$, all agents choosing $1$ cannot be a Nash equilibrium. Thus, there must be some agent $j$ choosing $0$ with positive probability. Since this agent cannot benefit from switching to action $1$, we have
\begin{equation}\label{eq_payoff_j}
0=\E[u_j(a_j,a_{-j})]\geq \E[u_j(1,a_{-j})]=\E\left[f\left(\frac{1+\sum_{k\ne j}{a_k}}{n}\right)\right]
\end{equation}
The arguments of $f$ in~\eqref{eq_payoff_i} and~\eqref{eq_payoff_j} differ by at most $1/n$. Let $\delta_n$ be the maximal value of $|f(x)-f(y)|$ over $x,y\in [0,1]$ with $|x-y|\leq 1/n$.\footnote{For functions $f$, Lipschitz with parameter $L$, we have $\delta_n\leq L/n$. For continuously differentiable $f$, the Lipschitz constant $L$ equals $\max_{[0,1]} |f'(x)|$.} We obtain that 
\begin{equation*}  \E[u_i(1,a_{-i})]\leq \delta_n,\end{equation*} 
and conclude that all agents' equilibrium payoffs are between $0$ and $\delta_n$. The continuity of $f$ on $[0,1]$ implies the uniform continuity by the Heine–Cantor theorem, and thus $\delta_n\to 0$ as $n\to \infty$, completing the proof.
\end{proof}

\end{document}